\PassOptionsToPackage{nosumlimits,nonamelimits}{amsmath}
\PassOptionsToPackage{colorlinks,linkcolor={blue},citecolor={blue},urlcolor={red},breaklinks=true,final}{hyperref}
\documentclass[sigconf,screen,nonacm]{acmart}
\usepackage{ifthen}
\newboolean{arxiv}
\setboolean{arxiv}{true} %
\newcommand{\ifarx}[2]{\ifthenelse{\boolean{arxiv}}{#1}{#2}}
\ifarx{
\settopmatter{printacmref=false}}{} %
\renewcommand\footnotetextcopyrightpermission[1]{} %
\copyrightyear{2024}
\acmYear{2024}
\setcopyright{acmlicensed}\acmConference[LICS '24]{39th Annual ACM/IEEE Symposium on Logic in Computer Science}{July 8--11, 2024}{Tallinn, Estonia}
\acmBooktitle{39th Annual ACM/IEEE Symposium on Logic in Computer Science (LICS '24), July 8--11, 2024, Tallinn, Estonia}
\acmDOI{10.1145/3661814.3662099}
\acmISBN{979-8-4007-0660-8/24/07}

\usepackage{booktabs}   %
\usepackage{subcaption} %

\usepackage{stel-common}

\newcommand{\by}[1]{\text{/$\mspace{-2mu}$/~#1}}		       %

\usepackage{bm} %

\providecommand{\catname}{\mathbf} 
\providecommand{\clsname}{\mathcal}
\providecommand{\oname}[1]{{\mathop{\mathsf{#1}}\xspace}}

\def\defcatname#1{\expandafter\def\csname B#1\endcsname{\catname{#1}}}
\def\defcatnames#1{\ifx#1\defcatnames\else\defcatname#1\expandafter\defcatnames\fi}
\defcatnames ABCDEFGHIJKLMNOPQRSTUVWXYZ\defcatnames

\def\defclsname#1{\expandafter\def\csname C#1\endcsname{\clsname{#1}}}
\def\defclsnames#1{\ifx#1\defclsnames\else\defclsname#1\expandafter\defclsnames\fi}
\defclsnames ABCDEFGHIJKLMNOPQRSTUVWXYZ\defclsnames

\def\defbbname#1{\expandafter\def\csname BB#1\endcsname{{\mathbb{#1}}}}
\def\defbbnames#1{\ifx#1\defbbnames\else\defbbname#1\expandafter\defbbnames\fi}
\defbbnames ABCDEFGHIJKLMNOPQRSTUVWXYZ\defbbnames

\def\Set{\catname{Set}}

\providecommand{\argument}{-}

\DeclareOldFontCommand{\bf}{\normalfont\bfseries}{\mathbf}

\providecommand{\id}{\mathsf{id}}
\providecommand{\op}{\mathsf{op}}
\providecommand{\comp}{\mathbin{\circ}}

\providecommand{\xto}[1]{\,\xrightarrow{#1}\,}

\providecommand{\To}{\mathrel{\Rightarrow}}			           %

\providecommand{\dar}{\kern-1.2pt\operatorname{\downarrow}}	
\providecommand{\uar}{\kern-1.2pt\operatorname{\uparrow}}

\providecommand{\bigand}{\bigwedge}

\providecommand{\fst}{\oname{fst}}
\providecommand{\snd}{\oname{snd}}

\providecommand{\brks}[1]{\langle #1\rangle}

\providecommand{\inl}{\oname{inl}}
\providecommand{\inr}{\oname{inr}}

\DeclareSymbolFont{Symbols}{OMS}{cmsy}{m}{n}
\DeclareMathSymbol{\iobj}{\mathord}{Symbols}{"3B}
\usepackage{stmaryrd}

\providecommand{\by}[1]{\text{/\!\!/~#1}}			             %
\providecommand{\pacman}[1]{}					                     %

\newcommand{\undefine}[1]{\let #1\relax}					                       %

\providecommand{\mone}{{\text{\kern.5pt\rmfamily-}\mathsf{\kern-.5pt1}}}

\makeatletter
\def\mfix#1{\oname{#1}\@ifnextchar\bgroup\@mfix{}}	       %
\def\@mfix#1{#1\@ifnextchar\bgroup\mfix{}}			           %
\makeatother

\providecommand{\case}[3]{\mfix{case}{\mathbin{}#1}{of}{#2}{\kern-1pt;}{\mathbin{}#3}}

\usepackage{dsfont}

\DeclareMathSymbol{\mathinvertedexclamationmark}{\mathord}{operators}{'074}
\DeclareMathSymbol{\mathexclamationmark}{\mathord}{operators}{'041}
\makeatletter
\newcommand{\raisedmathinvertedexclamationmark}{%
  \mathord{\mathpalette\raised@mathinvertedexclamationmark\relax}%
}
\newcommand{\raised@mathinvertedexclamationmark}[2]{%
  \raisebox{\depth}{$\m@th#1\mathinvertedexclamationmark$}%
}
\makeatother

\newcommand{\wt}{\widetilde}

\newcommand{\Pt}{V}
\newcommand{\R}{\mathcal{R}}
\newcommand{\st}{\mathsf{st}}

\newcommand{\Pow}{\mathcal{P}}
\newcommand{\xTo}{\xRightarrow}

\newcommand{\smc}{\mathbin{;}}

\newcommand{\qand}{\quad\text{and}\quad}

\newcommand{\dleq}[1]{{\rotatebox{#1}{$\preceq$}}}
\newcommand{\dgeq}[1]{{\rotatebox{#1}{$\succeq$}}}

\newcommand{\under}[1]{\lvert#1\rvert}

\renewcommand{\L}{\mathcal{L}}

\newcommand{\Sigmas}{\Sigma^{\star}}

\newcommand{\ar}{\mathsf{ar}}

\newcommand{\seq}{\subseteq}
\newcommand{\ol}{\overline}

\newcommand{\outl}{\mathsf{l}}
\newcommand{\outr}{\mathsf{r}}

\newcommand{\dbtilde}[1]{\widetilde{\raisebox{0pt}[0.85\height]{$\widetilde{#1}$}}}
\providecommand{\C}{}
\providecommand{\D}{}

\renewcommand{\C}{{\mathbb{C}}}
\renewcommand{\D}{{\mathbb{D}}}

\renewcommand{\id}{{\mathsf{id}}}

\newcommand{\Rel}{\mathbf{Rel}}

\newcommand{\f}{\oname{f}}

\newcommand{\takeout}[1]{\empty}

\newcommand{\ini}{\iota}

\renewcommand{\rho}{\varrho}

\newcommand{\opp}{\mathsf{op}}

\newcommand{\pullbackangle}[2][]{\arrow[phantom,to path={
                     -- ($ (\tikztostart)!1cm!#2:([xshift=8cm]\tikztostart) $)
                        node[anchor=west,pos=0.0,rotate=#2,
                        inner xsep = 0]
                        {\begin{tikzpicture}[minimum
                        height=1mm,baseline=0,#1]
    \draw[-] (0,0) -- (.5em,.5em) -- (0,1em);
                        \end{tikzpicture}}}]{}}

\makeatletter
\newsavebox{\@brx}
\newcommand{\llangle}[1][]{\savebox{\@brx}{\(\m@th{#1\langle}\)}%
  \mathopen{\copy\@brx\kern-0.5\wd\@brx\usebox{\@brx}}}
\newcommand{\rrangle}[1][]{\savebox{\@brx}{\(\m@th{#1\rangle}\)}%
  \mathclose{\copy\@brx\kern-0.5\wd\@brx\usebox{\@brx}}}
\makeatother

\renewcommand{\comp}{\cdot}
\renewcommand{\c}{\colon}

\newcommand{\xra}[1]{\mathrel{\raisebox{-1.15pt}{$\xrightarrow{\;\smash{\raisebox{2.5pt}{\makebox(3,0)[t]{\scriptsize $#1$}}\;}}$}}}
\renewcommand{\xto}{\xra}

\newcommand{\monoto}{\rightarrowtail}
\newcommand{\subto}{\hookrightarrow}

\newcommand{\fset}{{\mathbb{F}}}

\newcommand{\gcat}{\mathbb{C}}

\newcommand{\mS}{{\mu\Sigma}}

\renewcommand{\epsilon}{\varepsilon}

\newcommand*\xbar[1]{%
  \kern.2em\hbox{%
    \vbox{%
      \hrule height 0.5pt %
      \kern0.5ex%
      \hbox{%
        \kern-0.2em%
        \ensuremath{#1}%
        \kern-0.4em%
      }%
    }%
  }\kern.4em %
} 

\makeatletter
\newcommand{\monto}{\@ifstar{\@mtolifted}{\@mto}}
\newcommand{\@mto}{\multimapdot}
\newcommand{\@mtolifted}{\mathbin{\xbar{\multimapdot}}}
\makeatother

\newcommand{\bool}{\mathsf{bool}}

\newcommand{\app}[2]{\,}

\makeatletter
	\newcommand{\pushright}[1]{\ifmeasuring@#1\else\omit\hfill$\displaystyle#1$\fi\ignorespaces}
	\newcommand{\pushleft}[1]{\ifmeasuring@#1\else\omit$\displaystyle#1$\hfill\fi\ignorespaces}
\makeatother

\usepackage[capitalize,noabbrev,nameinlink]{cleveref} %

\usepackage{enumitem}
\setlist[enumerate,1]{label=(\arabic*),font=\normalfont,align=left,leftmargin=0pt,labelindent=0pt,listparindent=\parindent,labelwidth=0pt,itemindent=!,topsep=2pt,parsep=0pt,itemsep=2pt,start=1}
\setlist[enumerate,2]{label=(\alph*),font=\normalfont,labelindent=*,leftmargin=*,start=1}
\setlist[itemize]{labelindent=*,leftmargin=*}
\setlist[description]{labelindent=*,leftmargin=*,itemindent=-1 em}

\usepackage{microtype}

\usepackage{ifdraft}
\ifdraft{
  \usepackage[draft]{showlabels}

  \usepackage[footnote,marginclue,nomargin]{fixme}
}{
 \usepackage[layout=footnote,final]{fixme}
}

\FXRegisterAuthor{hu}{ahu}{HU} %
\FXRegisterAuthor{sm}{asm}{SM} %
\FXRegisterAuthor{st}{ast}{ST} %
\FXRegisterAuthor{ls}{als}{LS} %
\FXRegisterAuthor{sg}{asg}{SG} %

\renewcommand{\comp}{\cdot}
\renewcommand{\c}{\colon}

\newcommand{\cprd}{\lesssim^{\raisebox{-1pt}{\scriptsize $O$}}}

\usepackage{seqsplit}
\usepackage{xstring}
\usepackage{xcolor}


\tikzstyle{shiftarr}=[
        rounded corners,%
        to path={--([#1]\tikztostart.center)
                     -- ([#1]\tikztotarget.center) \tikztonodes
                     -- (\tikztotarget)},
]

\tikzset{
    commutative diagrams/.cd,
    arrow style=tikz,
    diagrams={>={Straight Barb[length=1.75pt,width=3.85pt,inset=1.95pt]}}, %
    row sep=large,
    column sep = huge
}

\tikzset{cong/.style={draw=none,edge node={node [sloped, allow upside down, auto=false]{$\cong$}}},
         iso/.style={draw=none,every to/.append style={edge node={node [sloped, allow upside down, auto=false]{$\cong$}}}}}

\usetikzlibrary{decorations.pathmorphing}

\tikzcdset{scale cd/.style={every label/.append style={scale=#1},
    cells={nodes={scale=#1}}}}
\usepackage{hypcap}
\usepackage{stackengine}
\stackMath
\newcommand\tsup[2][2]{%
 \def\useanchorwidth{T}%
  \ifnum#1>1%
    \stackon[-1.3ex]{\tsup[\numexpr#1-1\relax]{#2}}{\scalebox{2}[1]{$\mathchar"307E$}\kern-.5pt}%
  \else%
    \stackon[-1ex]{#2}{\scalebox{2}[1]{$\mathchar"307E$}\kern-.5pt}%
  \fi%
}

\newcommand{\ST}[1]{\textcolor{purple}{ST: #1}}

\usepackage{xspace}

\usepackage{etoolbox} %

\theoremstyle{definition}
\newtheorem{defn}[theorem]{Definition} %
\newtheorem{construction}[theorem]{Construction} %
\newtheorem{rem}[theorem]{Remark} %

\newtheorem{assumption}[theorem]{Assumption}

\newcommand{\arP}{\overrightarrow{\Pow}}

\newcommand{\Tr}{\mathsf{Tr}} %

\newcommand{\stscr}{\ensuremath{\boldsymbol\mu}\textbf{TCL}\xspace}
\newcommand{\fpc}{\textbf{FPC}\xspace}

\newcommand{\Ty}{\mathsf{Ty}}
\newcommand{\Tyl}{\mathsf{Ty}}

\newcommand{\arty}[2]{#1 \rightarrowtriangle #2} %
\newcommand{\cty}[2]{#1 \rightsquigarrow #2} %

\newcommand{\Pred}[2][]{\mathbf{Pred}_{#1}(#2)}
\newcommand{\RelCat}[2][]{\mathbf{Rel}_{#1}(#2)}
\newcommand{\RelAt}[1]{\mathbf{Rel}_{#1}}

\newcommand{\pred}[1]{\overset{#1}{\rightarrowtail}}
\newcommand{\utype}{\mathsf{unit}}
\newcommand{\booltype}{\mathsf{bool}}
\newcommand{\voidtype}{\mathsf{void}}

\newcommand{\nattype}{\mathsf{nat}}

\let\oldcheckmark\checkmark
\renewcommand{\checkmark}{\raisebox{-4pt}{\scalebox{1.2}[.65]{$\oldcheckmark$}}}

\usepackage{tensor}

\newcommand{\iimg}[2]{#1{}^\star[#2]}
\newcommand{\fimg}[2]{#1{}_\star[#2]}

\makeatletter
\DeclareRobustCommand\widecheck[1]{{\mathpalette\@widecheck{#1}}}
\def\@widecheck#1#2{%
  \setbox\z@\hbox{\m@th$#1#2$}%
  \setbox\tw@\hbox{\m@th$#1%
    \widehat{%
      \vrule\@width\z@\@height\ht\z@
      \vrule\@height\z@\@width\wd\z@}$}%
  \dp\tw@-\ht\z@ppp
  \@tempdima\ht\z@ \advance\@tempdima2\ht\tw@ \divide\@tempdima\thr@@
  \setbox\tw@\hbox{%
    \raise\@tempdima\hbox{\scalebox{1}[-1]{\lower\@tempdima\box
        \tw@}}}%
  {\ooalign{\box\tw@ \cr \box\z@}}}
\makeatother

\makeatletter
\newcommand{\superimpose}[2]{{%
  \ooalign{%
    \hfil$\m@th#1\@firstoftwo#2$\hfil\cr
    \hfil$\m@th#1\@secondoftwo#2$\hfil\cr
  }%
}}
\makeatother

\let\oldhat=\hat
\renewcommand{\hat}[1]{\widehat{#1}}

\newcommand{\logp}{\square}
\newcommand{\invp}{\square}

\newcommand{\logrel}{\mathcal{L}}
\newcommand{\logreltwo}{\mathcal{M}}

\theoremstyle{definition}
\newtheorem{notation}[theorem]{Notation}
\newtheorem{assumptions}[theorem]{Assumptions}

\usepackage{scalerel,stackengine,graphicx}

\newcommand\myarr[1]{\stackrel{\mathclap{#1}}{\Mapsto}}

\usepackage[normalem]{ulem}

\def\star{\ast} %

\begin{document}
\allowdisplaybreaks

\title[Bialgebraic Reasoning on Higher-Order Program Equivalence]{Bialgebraic
  Reasoning on Higher-Order Program Equivalence}

\author{Sergey Goncharov}\authornote{Funded by the Deutsche Forschungsgemeinschaft (DFG, German
  Research Foundation) -- project number 527481841}
\orcid{0000-0001-6924-8766}             %
\affiliation{
  \institution{Friedrich-Alexander-Universität Erlangen-Nürnberg}            %
  \country{Germany}                    %
}
\email{sergey.goncharov@fau.de}          %

\author{Stefan Milius}\authornote{Funded by the Deutsche Forschungsgemeinschaft (DFG, German
  Research Foundation) -- project number 517924115} %
\orcid{0000-0002-2021-1644}             %
\affiliation{
  \institution{Friedrich-Alexander-Universität Erlangen-Nürnberg}            %
  \country{Germany}                    %
}
\email{stefan.milius@fau.de}          %

\author{Stelios Tsampas}
\authornote{Funded by the Deutsche Forschungsgemeinschaft (DFG, German
  Research Foundation) -- project number 419850228 and 527481841}  %
\orcid{0000-0001-8981-2328}             %
\affiliation{
  \institution{Friedrich-Alexander-Universität Erlangen-Nürnberg}            %
  \country{Germany}                    %
}
\email{stelios.tsampas@fau.de}          %

\author{Henning Urbat}
\authornote{Funded by the Deutsche Forschungsgemeinschaft (DFG, German
  Research Foundation) -- project number 470467389}  %
\orcid{0000-0002-3265-7168}             %
\affiliation{
  \institution{Friedrich-Alexander-Universität Erlangen-Nürnberg}            %
  \country{Germany}                    %
}
\email{henning.urbat@fau.de}          %

\sloppy

\begin{abstract}
Logical relations constitute a key method for reasoning about contextual
equivalence of programs in higher-order languages. They are usually developed on
a per-case basis, with a new theory required for each variation of the
language or of
the desired notion of equivalence. In the present paper we introduce a general
construction of (step-indexed) logical relations at the level of \emph{Higher-Order
Mathematical Operational Semantics}, a highly parametric categorical framework
for modeling the operational semantics of higher-order languages. Our main
result states that for languages whose weak operational model forms a lax
bialgebra, the logical relation is automatically sound for contextual
equivalence. Our abstract theory is shown to instantiate to combinatory logics
and $\lambda$-calculi with recursive types, and to different flavours of
contextual equivalence.
\end{abstract}

\maketitle

\section{Introduction}\label{sec:intro}
Reasoning about program equivalence is one of the primary goals of the theory of programming languages, and it is known to be particularly challenging in the presence of higher-order features. The definite operational notion of program equivalence for higher-order languages is given by \emph{contextual equivalence}~\cite{morris}: two programs $p$ and $q$ are contextually equivalent if, whichever context~$C[\cdot]$ they are executed in, they co-terminate:
\[ p\simeq q \quad\text{iff}\quad \forall C.\, (C[p] \text{ terminates} \,\Leftrightarrow\, C[q] \text{ terminates}).\]
Informally, a context $C[\cdot]$ runs tests on its input `$\cdot$', and $p\simeq q$ means that no such test leads to observations distinguishing $p$ from $q$.

While the definition of contextual equivalence is simple and natural, proving $p\simeq q$ directly can be very difficult due to the quantification over all possible contexts.
Therefore, a substantial strain of research has been devoted to developing sound (and ideally complete) proof techniques for contextual equivalence. The most widely used approach are \emph{logical relations}. Originally introduced in the context of denotational semantics~\cite{MILNER1978348,DBLP:journals/iandc/OHearnR95,ongthesis,PITTS199666,plotkin1973lambda,sieber_1992,STATMAN198585,DBLP:journals/jsyml/Tait67}, logical relations have evolved into a robust and ubiquitous operational technique~\cite{DBLP:journals/toplas/AppelM01,5230591,DBLP:conf/popl/HurDNV12,DBLP:conf/lics/Pitts96,DBLP:journals/mscs/Pitts00,10.5555/645722.666533}. Besides reasoning about contextual equivalence~\cite{10.1007/11693024_6,10.1145/3158152,DBLP:conf/popl/DevriesePP16,DBLP:conf/fossacs/BizjakB15,DBLP:journals/corr/abs-1103-0510,DBLP:conf/popl/HurDNV12,Pitts:1999:ORF:309656.309671,10.5555/1076265,DBLP:conf/icalp/Pitts98,DBLP:journals/jfp/DreyerNB12,DBLP:conf/popl/AhmedDR09}, they have been applied to proofs of strong
normalization~\cite{altenkirch_et_al:LIPIcs:2016:5972}, safety properties~\cite{10.1145/3623510,10.1145/3408996}, and formally verified
compilation~\cite{10.1145/1596550.1596567,DBLP:conf/popl/HurD11,DBLP:conf/icfp/NewBA16,10.1145/3434302}, in a variety of settings such as as effectful~\cite{5571705},
probabilistic~\cite{DBLP:journals/pacmpl/AguirreB23,10.1007/978-3-662-46678-0_18,10.1145/3236782}, and
differential programming~\cite{10.1016/j.tcs.2021.09.027,DBLP:journals/pacmpl/LagoG22,dallago_et_al:LIPIcs:2019:10687}.

A logical relation is, roughly, a type-indexed relation on
program terms that respects their operational behaviour, with the distinctive requirement that related terms
of function type should send related inputs to related outputs. In practice, logical relations
are introduced in a largely empirical fashion. Every higher-order language
comes with its own tailor-made construction of a logical relation (or
several of those), involving careful design choices to match the
features and idiosyncrasies\sgnote{Maybe, soften.} of the language, and accompanied with a
sequence of technical lemmas establishing its compatibility and
soundness properties. The required proofs are typically not inherently
difficult but notoriously long, tedious, and error-prone. Moreover,
every variation of the underlying language (e.g.\ its syntax, its type
system, its computational effects) or of the targeted form of
contextual equivalence (e.g.\ restricting the admissible contexts or
observations) requires a careful adaptation of definitions and proof
details. As a result, logical relations tend to scale rather poorly with the size and complexity of the language.  

The contribution of the present paper is a {generic}, \emph{language-independent}, theory of logical relations that aims to address these issues. It is developed at the level of generality of \emph{Higher-Order Mathematical Operational Semantics}, a.k.a.\ \emph{higher-order abstract GSOS}, a categorical framework for the operational semantics of higher-order languages recently introduced by~\citet{gmstu23} that builds on the seminal work of \citet{DBLP:conf/lics/TuriP97}. In higher-order abstract GSOS, the (small-step) operational semantics of a language is modeled abstractly as a \emph{higher-order GSOS law}, a dinatural transformation that distributes the \emph{syntax} of a language (given by an endofunctor $\Sigma$ on a category $\C$) over its \emph{behaviour} (given by a mixed variance bifunctor $B\colon \C^\op\times \C\to \C$). Every higher-order GSOS law comes equipped with a canonical operational model, which forms a \emph{(higher-order) bialgebra} 
\begin{equation}\label{eq:op-model} \Sigma(\mS)\xto{\cong} \mS \xto{\gamma} B(\mS,\mS) \end{equation}
on the object $\mS$ of program terms. Intuitively, the coalgebra $\gamma$ is the transition system that runs programs according to the operational rules encoded by the underlying higher-order GSOS law.

At this level of abstraction, we show that every higher-order GSOS law induces a
family of \emph{generic} logical relations on $\mS$ (\autoref{def:stepindexed}).
They are constructed as \emph{step-indexed logical
  relations}~\cite{DBLP:journals/toplas/AppelM01}, a general and robust kind of
logical relations that is suitable for languages with recursive types. In addition, we introduce a notion of  \emph{contextual preorder} $\cprd$ on $\mS$
that is parametric in a preorder $O$ of admissible \emph{observations} $O$ on program
contexts. We regard this notion and its properties established in \autoref{sec:ctxeq} to be of independent interest, as to our knowledge it provides the first treatment of contextual preorders as an abstract categorical notion.

Our main result (\autoref{thm:main} and its corollaries) states that
for languages modeled in higher-order abstract GSOS, the soundness of generic
logical relations often comes for free. We may state this result informally as
follows.

\smallskip
\noindent\textbf{Generic Soundness Theorem.} If the small-step operational rules of a language remain sound for weak transitions, and a generic logical relation is adequate for $O$, then it is sound for $\cprd$.
\smallskip 

In technical terms, the condition on weak transitions amounts to requiring that the weak operational model associated to \eqref{eq:op-model} forms a \emph{lax bialgebra}. This condition isolates the language-specific core of soundness proofs for logical relations, and it is usually easily verified by inspecting the operational rules of the language. Hence, our generic soundness theorem can significantly reduce the proof burden associated with soundness results for logical relations. 

Thanks to their highly parametric nature, our categorical results instantiate to a wide variety of different settings and languages. In our presentation we focus on the functional language \fpc~\cite{Gunter92,FioreP94} (a $\lambda$-calculus with recursive types), and a corresponding combinatory logic called \stscr. Let us note that the previous work on higher-order abstract GSOS has only considered untyped~\cite{gmstu23,UrbatTsampasEtAl23} and simply typed languages~\cite{gmstu24}. The insight that complex type systems featuring recursive types can be implemented in the abstract framework is thus a novel contribution in itself.

\paragraph{Related work} The majority of the literature dedicated to the
abstract study of logical relations has so far focused on
denotational~\cite{DBLP:journals/mscs/Goubault-LarrecqLN08,DBLP:journals/entcs/HermidaRR14,hermida1993fibrations,DBLP:phd/ethos/Katsumata05}
rather than operational logical relations. In recent
work~\cite{DBLP:conf/fscd/DagninoG22,DBLP:journals/corr/abs-2303-03271},
Dagnino and Gavazzo introduce a categorical notion of operational
logical relations that is largely orthogonal to ours, in particular
regarding the parametrization of the framework: In \emph{op.\ cit.},
the authors work with a fixed \emph{fine-grain call-by-value} language~\cite{LevyPowerEtAl03} parametrized by a signature
of generic effects, while the notion of logical relation is kept
variable and in fact is parametrized over a fibration; in contrast, we adhere to
conventional logical relations and parametrize over the syntax and semantics of the language. 
Unlike our present paper, Dagnino and Gavazzo have not incorporated {step-indexing}
and recursive types.

\citet{gmstu24} devise a generic construction of unary, non-step-indexed logical predicates in the framework of higher-order
abstract GSOS, and provide induction up-to techniques for reasoning on them, with \emph{strong normalization}
predicates as a key application. 
First-order lax bialgebras are initially proposed by~\citet{DBLP:conf/concur/BonchiPPR15} 
to capture coalgebraic weak bisimilarity and are subsequently generalized to higher-order lax bialgebras 
by~Urbat et al.~\cite{UrbatTsampasEtAl23}, where the authors gave an abstract proof of congruence of applicative (bi)similarity~\cite{Abramsky:lazylambda} in higher-order abstract GSOS, using a categorical generalization of Howe's method~\cite{DBLP:conf/lics/Howe89, DBLP:journals/iandc/Howe96}. The relation between the present paper and the work of Urbat et al.\ is elaborated on in \autoref{rem:logrel-vs-howe}.

While the standard version of step-indexing involves only natural numbers,
it turns out that indexing by arbitrary ordinals
emerges as a matter of course in our abstract formulation; this relates our notion to the recent
generalizations~\cite{SpiesKrishnaswamiEtAl21,DBLP:journals/pacmpl/AguirreB23}.

\section{Preliminaries}

\subsection{Category Theory}\label{sec:category-theory}
\label{sec:categories}

We assume familiarity with basic category theory~\cite{mac2013categories}, e.g.~functors, natural transformations, (co)limits, and monads. Here, we
briefly review some terminology and notation.

\paragraph*{Notation}
For objects
$X_1, X_2$ of a category $\C$, we denote their product by $X_1\times X_2$ and
the pairing of morphisms $f_i\c X\to X_i$, $i=1,2$, by $\langle f_1, f_2\rangle\c
X\to X_1\times X_2$. We write
$X_1+X_2$ for the coproduct, $\inl\c X_1\to X_1+X_2$ and
$\inr\c X_2\to X_1+X_2$ for its injections, $[g_1,g_2]\c X_1+X_2\to X$ for the copairing of morphisms $g_i\colon X_i\to X$,
$i=1,2$, and $\nabla=[\id_X,\id_X]\colon X+X\to X$ for the codiagonal. The \emph{coslice category} $X/\gcat$,
where $X\in \gcat$, has as objects all pairs $(Y,p_Y)$ consisting of an object $Y\in \gcat$
and a morphism $p_Y\c X\to Y$, and a morphism from $(Y,p_Y)$ to
$(Z,p_Z)$ is a morphism $f\c Y\to Z$ of $\gcat$ such that
$p_Z = f\comp p_Y$. The \emph{slice category} $\C/X$ is defined dually.

\paragraph*{Algebras}
Let $F$ be an endofunctor on a category $\gcat$. An \emph{$F$-algebra}
is a pair $(A,a)$ consisting of an object~$A$ (the \emph{carrier} of the
algebra) and a morphism $a\colon FA\to A$ (its \emph{structure}). A
\emph{morphism} from $(A,a)$ to an $F$-algebra $(B,b)$ is a morphism
$h\colon A\to B$ of~$\gcat$ such that $h\comp a = b\comp Fh$. Algebras
for $F$ and their morphisms form a category, whose initial object, if it exists, 
is called an \emph{initial algebra}; we denote its carrier by $\mu F$ 
and its structure by ${\ini\colon F(\mu F) \to \mu F}$.
A \emph{free $F$-algebra} on an object $X$ of $\gcat$ is an
$F$-algebra $(F^{\star}X,\iota_X)$ together with a morphism
$\eta_X\c X\to F^{\star}X$ of~$\gcat$ such that for every algebra $(A,a)$
and every morphism $h\colon X\to A$ in $\gcat$, there exists a unique
$F$-algebra morphism $h^\star\colon (F^{\star}X,\iota_X)\to (A,a)$
such that $h=h^\star\comp \eta_X$; the morphism $h^\star$ is called the \emph{free
  extension} of $h$. If free algebras
exist on every object, their formation gives rise to a monad
$F^{\star}\colon \gcat\to \gcat$, the \emph{free monad} generated by~$F$.
For every $F$-algebra $(A,a)$ we can derive an Eilenberg-Moore algebra
$\hat{a} \colon F^{\star} A \to A$ whose structure is the free extension of $\id_A\c A\to A$.

The most familiar example of functor algebras are algebras for a
signature. Given a set $S$ of \emph{sorts}, an \emph{$S$-sorted algebraic signature} consists of a set~$\Sigma$
of \emph{operation symbols} and a map $\ar\colon \Sigma\to S^{\star}\times S$
associating to every $\f\in \Sigma$ its \emph{arity}. We write $\f\colon s_1\times\cdots\times s_n\to s$ if $\ar(\f)=(s_1,\ldots,s_n,s)$, and $\f\colon s$ if $n=0$ (in which case $\f$ is called a \emph{constant}). Every
signature~$\Sigma$ induces an endofunctor on the category $\Set^S$ of $S$-sorted sets and $S$-sorted functions, denoted by the
same letter $\Sigma$, defined by $(\Sigma X)_s = \coprod_{\f\colon s_1\cdots s_n\to s} \prod_{i=1}^n X_{s_i}$ for $X\in \Set^S$ and $s\in S$. (Endofunctors of this form are called \emph{polynomial endofunctors}.) An algebra for the functor $\Sigma$ is
precisely an algebra for the signature $\Sigma$, viz.~an $S$-sorted set $A=(A_s)_{s\in S}$ equipped with an operation $\f^A\colon \prod_{i=1}^n A_{s_i}\to A_s$ for every $\f\colon s_1\times\cdots \times s_n\to s$ in $\Sigma$. Morphisms of $\Sigma$-algebras are $S$-sorted
maps respecting the algebraic structure. Given an $S$-sorted set $X$ of
variables, the free algebra $\Sigmas X$ is the $\Sigma$-algebra of
$\Sigma$-terms with variables from~$X$; more precisely, $(\Sigmas X)_s$ is inductively defined by $X_s\seq (\Sigmas X)_s$ and $\f(t_1,\ldots,t_n)\in (\Sigmas X)_s$ for all ${\f\colon s_1\times\cdots \times s_n\to s}$ and $t_i\in (\Sigmas X)_{s_i}$.
The free
algebra on the empty set is the initial algebra~$\mu \Sigma$; it is
formed by all \emph{closed terms} of the signature. We write $t\colon s$ for $t\in (\mS)_s$. For every
$\Sigma$-algebra $(A,a)$, the induced Eilenberg-Moore algebra
$\hat{a}\colon \Sigmas A \to A$ is given by the map that evaluates terms in~$A$.

An \emph{($S$-sorted) relation} $R$ on $X\in \Set^S$, denoted $R\seq X\times X$, is a family of relations $(R_s\seq X_s\times X_s)_{s\in S}$.
We write $x\, R_s\,y$ or $R_{s}(x,y)$ for $(x,y) \in
R_{s}$, and sometimes omit the subscript~$s$. Moreover we let $\Rel_{X}$ denote the complete lattice of relations on $X$, ordered by inclusion in every sort. Its top element is the full relation $\top=X\times X$, and meets are given by sortwise intersection.

A \emph{congruence} on a $\Sigma$-algebra $A$ is an $S$-sorted relation ${R\seq A\times A}$ compatible with all operations of $A$: for each $\f\colon s_1\times\cdots \times s_n\to s$ and elements $x_i,y_i\in A_{s_i}$ such that $R_{s_i}(x_i,y_i)$ ($i=1,\ldots,n$), one has $R_s(\f^A(x_1,\ldots,x_n), \f^A(y_1,\ldots,y_n))$. Unlike other authors we do not require congruences to be equivalence relations. However, every congruence on the initial algebra $\mS$ is reflexive in every sort.

\paragraph*{Coalgebras}
Dual to the notion of algebra, a \emph{coalgebra} for an endofunctor $F$ is a pair $(C,c)$ consisting of an object $C$ (the
\emph{state space}) and a morphism $c\colon C\to FC$ (its
\emph{structure}). Coalgebras are abstractions of state-based systems~\cite{DBLP:journals/tcs/Rutten00}. For instance, a coalgebra $c\colon C\to {\Pow(L\times C)}$ for the set functor $\Pow(L\times -)$, where~$\Pow$ is the power set functor and $L$ is a fixed set of labels, is precisely a labelled transition system.

\subsection{Higher-Order Abstract GSOS}
\begin{figure*}
\begin{tikzcd}[column sep=6em]
\Sigma(\mS) \ar{rrr}{\iota} \ar{d}[swap]{\Sigma\langle \id, \gamma\rangle} & & & \mS \ar[dashed]{d}{\gamma} \\
\Sigma(\mS\times B(\mS,\mS)) \ar{r}{\rho_{\mS,\mS}} & B(\mS,\Sigmas(\mS+\mS)) \ar{r}{B(\id,\Sigmas \nabla)} & B(\mS,\Sigmas(\mS)) \ar{r}{B(\id,\oldhat\ini)} & B(\mS,\mS) 
\end{tikzcd}
\caption{Operational model of a higher-order GSOS law}\label{fig:gamma}
\end{figure*}
\label{sec:abstract-gsos}
We review the fundamentals of \emph{higher-order abstract
  GSOS}~\cite{gmstu23}, a categorical framework for the operational
semantics of higher-order languages.  It is parametric in the following data:
\begin{enumerate}
\item a category $\gcat$ with finite products and coproducts;
\item an object $\Pt\in \gcat$ of \emph{variables};
\item two functors $\Sigma\c\gcat \to \gcat$ and $B\c \gcat^\opp\times \gcat\to
  \gcat$, where $\Sigma=\Pt+\Sigma'$ for some
  functor $\Sigma'\colon \gcat \to \gcat$, and free $\Sigma$-algebras exist on
  all ${X \in \gcat}$.
\end{enumerate}
The functors $\Sigma$ and $B$ model the \emph{syntax}
and the \emph{behaviour} of a higher-order language. The mixed variance of $B$ reflects the requirement that programs can occur both contravariantly (as inputs) and covariantly (as outputs) of programs. The initial
algebra~$\mS$ is the object of program terms in the language; since $\Sigma=\Pt+\Sigma'$, variables are terms. An object $(X,p_X)$ of the coslice category $V/\gcat$ (a \emph{$V$-pointed object} for short) is thought of as an abstract set $X$ of programs with an embedding $p_X\colon V\to X$ of the variables. For variable-free languages such at \stscr introduced in \autoref{sec:stscr}, we put $V=0$ (the initial object). In the setting of higher-order abstract GSOS, the operational semantics of a language is specified by a dinatural transformation that distributes syntax over behaviours:

\begin{definition}[Higher-Order GSOS Law]\label{def:ho-gsos-law}
  A \emph{($\Pt$-pointed) higher-order GSOS law} of $\Sigma$ over $B$
  is given by a family of morphisms
  \begin{align}\label{eq:ho-gsos-law}
    \rho_{(X,p_X),Y} \c \Sigma (X \times B(X,Y))\to B(X, \Sigma^\star (X+Y))
  \end{align}
  dinatural in $(X,p_X)\in \Pt/\gcat$ and natural in $Y\in \gcat$.
\end{definition}

\begin{notation}\label{not:rho}
\begin{enumerate}
\item We usually write $\rho_{X,Y}$ for $\rho_{(X,p_X),Y}$, as the
  point $p_X\c V\to X$ will always be clear from the context.
\item For every $\Sigma $-algebra $(A,a)$, we regard $A$ as a {$\Pt$-pointed} object with point
  $p_A = \bigl(\Pt\xra{\inl} \Pt+\Sigma' A = \Sigma  A \xra{a} A\bigr)$. 
\end{enumerate}
\end{notation}
A higher-order GSOS law $\rho$ is thought of as an encoding of the small-step operational rules of a higher-order language: given an operation $\f$ from $\Sigma$ and the one-step behaviours of its operands, the law $\rho$ specifies the one-step behaviour of a program $\f(-,\cdots,-)$, i.e.~the $\Sigma$-terms it transitions into. The \emph{operational model} of $\rho$ is a transition system on $\mS$ that runs programs according to the rules: 

\begin{definition}[Operational Model]\label{def:operational-model}
The \emph{operational model} of a higher-order GSOS law $\rho$ in \eqref{eq:ho-gsos-law} is the $B(\mS,-)$-coalgebra 
\begin{equation*}
\gamma\c \mS\to B(\mS,\mS)
\end{equation*}
obtained via primitive recursion~\cite[Prop.~2.4.7]{DBLP:books/cu/J2016} as the unique morphism making the diagram in
\Cref{fig:gamma} commute. Here we regard the initial algebra $\mu\Sigma$ as $V$-pointed as in \Cref{not:rho},
and $\hat{\ini}\colon \Sigmas(\mS) \to \mS$ is the $\Sigmas$-algebra induced by
$\iota\c \Sigma(\mS)\to \mS$.
\end{definition}
This makes $(\mS,\ini,\gamma)$ an (initial) \emph{$\rho$-bialgebra}; see \autoref{sec:step-indexed} for a more detailed discussion of bialgebras.

\section{Combinatory Logic}\label{sec:stscr}
 In this section we introduce a combinatory logic~\cite{hindley2008lambda} with \mbox{(iso-)}recursive types, called \stscr, that will serve as a running example for the abstract theory of logical relations developed in our paper. It forms a (computationally complete) fragment of the well-known functional language \fpc~\cite{Gunter92,FioreP94} but does not involve variables; this allows us to circumvent the technical overhead associated with binding and substitution. The fully fledged \fpc language is discussed in \autoref{sec:applications}.

\subsection{The \texorpdfstring{$\boldsymbol{\mu}$}{\mu}TCL Language}\label{sec:mu-ctl}
The type expressions of~\stscr{} are given by the grammar
\begin{equation}\label{eq:type-grammar}
  \tau_1,\tau_2,\tau,\ldots \Coloneqq \alpha \mid\tau_1\boxplus\tau_2\mid\tau_1\boxtimes\tau_2\mid \arty{\tau_1}{\tau_2}\mid \mu\alpha.\,\tau,
\end{equation}
where $\alpha$ ranges over a fixed countably infinite set of {type variables}. Free and bound 
type variables and $\alpha$-equivalence are defined in the usual way. We denote by 
$\Ty$ the set of closed type expressions modulo $\alpha$-equivalence, and refer to them simply as \emph{types}. The type constructors $\boxplus$,
$\boxtimes$, $\arty{}{}$ represent binary sums, binary products and function
spaces internal to the language, and are denoted non-standardly for distinction with the 
set constructors $+$, $\times$, $\to$. Using the recursion operator we define the empty type ($\voidtype=\mu\alpha.\,\alpha$), the unit type ($\utype=\arty{\voidtype}{\voidtype}$), and the types of booleans ($\booltype=\utype\boxplus
\utype$) and natural numbers ($\nattype=\mu\alpha.\,\utype\boxplus\alpha$).
We write $\tau_{1}[\tau_{2}/\alpha]$ for the capture-avoiding
substitution of $\tau_2$ for~$\alpha$. The constructor $\arty{}{}$
binds most weakly and is right-associative, which means that
$\arty{\tau_{1}}{\arty{\tau_{2}}{\tau_{3}}}$ is parsed as
$\arty{\tau_{1}}{({\arty{\tau_{2}}{\tau_{3}}})}$. 

The syntax of \stscr is specified by the $\Ty$-sorted signature $\Sigma$ given by the following operation symbols, with
$\tau_1,\tau_2,\tau_3$ ranging over $\Ty$, and $\tau$ over types with at most one free variable $\alpha$:
\begin{flalign*}
  \quad S_{\tau_{1},\tau_{2},\tau_{3}} \c& \arty{(\arty{\tau_{1}}{\arty{\tau_{2}}{\tau_{3}}})}{\arty{(\arty{\tau_{1}}{\tau_{2}})}{\arty{\tau_{1}}{\tau_{3}}}} &\\
  S'_{\tau_1,\tau_2,\tau_3}\c& (\arty{\tau_{1}}{\arty{\tau_{2}}{\tau_{3}}})\to (\arty{{(\arty{\tau_{1}}{\tau_{2})}}}{\arty{\tau_{1}}{\tau_{3}}})  \\
  S''_{\tau_1,\tau_2,\tau_3}\c& (\arty{\tau_{1}}{\arty{\tau_{2}}{\tau_{3}}})\times (\arty{\tau_{1}}{\tau_{2}})\to (\arty{\tau_{1}}{\tau_{3}}) \\
  K_{\tau_{1},\tau_{2}} \c& \arty{\tau_{1}}{\arty{\tau_{2}}{\tau_{1}}} \qquad K'_{\tau_1,\tau_2}\c \tau_1\to (\arty{\tau_2}{\tau_1}) \\
  I_{{\tau_1}} \c& \arty{\tau_1}{\tau_1} \\[-4.5ex]
\end{flalign*}
\begin{flalign*}
  \quad\mathsf{app}_{\tau_1,\tau_2}\c
  &
    (\arty{\tau_{1}}{\tau_{2}})\times\tau_1\to
    \tau_2
  & \mathsf{fst}_{\tau_1,\tau_2}\c
  &
    {\tau_{1}\boxtimes\tau_{1}} \to {\tau_1}
  \\ 
  \mathsf{inl}_{\tau_1,\tau_2}\c& 					{\tau_{1}}\to{\tau_1\boxplus\tau_2}
& \mathsf{snd}_{\tau_1,\tau_2}\c& {\tau_{1}\boxtimes\tau_{2}} \to {\tau_2}  \\
  \mathsf{inr}_{\tau_1,\tau_2}\c& 					{\tau_{2}}\to{\tau_1\boxplus\tau_2} & \mathsf{pair}_{\tau_1,\tau_2}\c& \tau_1\times \tau_2\to{\tau_{1}\boxtimes\tau_2} &\\[-4ex]
\end{flalign*}
\begin{flalign*}
 \quad \mathsf{case}_{\tau_1,\tau_2,\tau_3}\c& 	(\tau_1\boxplus\tau_2)\times
                                          (\arty{\tau_{1}}{\tau_3})\times
                                          (\arty{\tau_{2}}{\tau_3})\to\tau_3&\\[1ex]
  \mathsf{fold}_{\tau} \c& \tau[\mu\alpha.\,\tau/\alpha]\to{\mu\alpha.\tau} \\
  \mathsf{unfold}_{\tau} \c& \mu\alpha.\,\tau\to\tau[\mu\alpha.\tau/\alpha]
\end{flalign*}
We let $\Tr=\mS$ denote the initial algebra for $\Sigma$, carried by the $\Ty$-sorted set of closed
$\Sigma$-terms. Type indices at polymorphic operations are often omitted for the sake of readability. The operation $\mathsf{app}$ represents function application, and we write $s\, t$ for $\mathsf{app}(s,t)$.
The familiar combinators $I$, $K$, $S$ 
represent the $\lambda$-terms $\lambda x.\,x$, $\lambda x.\,\lambda y.\, x$ and $\lambda x.\,\lambda y.\,\lambda z.\, (x\, z)\, (y\, z)$.
Apart from those, we use auxiliary operations $S'$, $S''$,
$K'$; these are needed to provide a small-step semantics, 
which is instrumental for our coalgebraic approach.
\begin{figure*}[t]
  \centering
  \columnwidth=\linewidth
  \begin{gather*}
    \inference{}{S_{\tau_{1},\tau_{2},\tau_{3}}\xto{e}S'_{\tau_{1},\tau_{2},\tau_{3}}(e)}
    \qquad
    \inference{}{S'_{\tau_{1},\tau_{2},\tau_{3}}(t)\xto{e}S''_{\tau_{1},\tau_{2},\tau_{3}}(t,e)}     
    \qquad
    \inference{}{S''_{\tau_{1},\tau_{2},\tau_{3}}(t,s)\xto{e}(t\app{\tau_{1}}{\arty{\tau_{2}}{\tau_{3}}} e)\app{\tau_{2}}{\tau_{3}} (s\app{\tau_{1}}{\tau_{2}} e)}
    \qquad
    \inference{}{K_{\tau_{1},\tau_{2}}\xto{e}K'_{\tau_{1},\tau_{2}}(e)}\\[2ex]
    \inference{}{K'_{\tau_{1},\tau_{2}}(t)\xto{e}t}
    \qquad
    \inference{}{I_{\tau}\xto{e}e}
    \qquad
    \inference{t\to t'}{t \app{\tau_{1}}{\tau_{2}} s\to t' \app{\tau_{1}}{\tau_{2}} s}
    \qquad
    \inference{t\xto{s} t'}{t \app{\tau_{1}}{\tau_{2}} s\to t'}
    \qquad\inference{}{\inl_{\tau_1,\tau_2} (t)\xto{\boxplus_1} t}
    \qquad\inference{}{\inr_{\tau_1,\tau_2} (t)\xto{\boxplus_2} t}
    \\[1ex]
    \inference{t\to t'}{\mathsf{case}_{\tau_1,\tau_2,\tau_3}(t,s,r)\to \mathsf{case}_{\tau_1,\tau_2,\tau_3}(t',s,r)}
    \qquad\inference{t\xto{\boxplus_1} t'}{\mathsf{case}_{\tau_1,\tau_2,\tau_3}(t,s,r)\to s\app{\tau_{1}}{\tau_{2}} t'}
    \qquad\inference{t\xto{\boxplus_2} t'}{\mathsf{case}_{\tau_1,\tau_2,\tau_3}(t,s,r)\to r\app{\tau_{1}}{\tau_{2}} t'}
    \\[1ex]
    \inference{}{\mathsf{pair}_{\tau_1,\tau_2}(t,s)\xto{\boxtimes_1} t}
    \qquad\inference{}{\mathsf{pair}_{\tau_1,\tau_2}(t,s)\xto{\boxtimes_2} s}
    \qquad\inference{t\to t'}{\mathsf{fst}_{\tau_1,\tau_2}(t)\to \mathsf{fst}_{\tau_1,\tau_2}(t')}
    \qquad\inference{t\to t'}{\mathsf{snd}_{\tau_1,\tau_2}(t)\to \mathsf{snd}_{\tau_1,\tau_2}(t')}
    \\[1ex]
    \inference{t\xto{\boxtimes_1} t'}{\mathsf{fst}_{\tau_1,\tau_2}(t)\to t'}
    \qquad\inference{t\xto{\boxtimes_2} t'}{\mathsf{snd}_{\tau_1,\tau_2}(t)\to t'}
    \qquad
    \inference{}{\mathsf{fold}_\tau(t) \xto{\mu} t} 
    \qquad
    \inference{t \to t'}{\mathsf{unfold}_\tau(t) \to \mathsf{unfold}_\tau(t')}
    \qquad
    \inference{t \xto{\mu} t'}{\mathsf{unfold}_\tau(t) \to t'}
    \end{gather*}
  \caption{Call-by-name operational semantics of \stscr.}
  \label{fig:skirules}
\end{figure*}

We equip \stscr with a call-by-name operational semantics whose transition rules are given in
\Cref{fig:skirules}; here $e,s,t,t'$ range over appropriately typed terms in $\Tr$. There are three kinds of transitions: 
\begin{itemize}[wide]
  \item $t\xto{} s$, indicating that $t$ one-step $\beta$-reduces to $s$,
  \item $t\xto{e} s$ where $e\in \Tr$, indicating that $t$ applied to $e$ yields $s$,
  \item $t\xto{l} s$ where $l \in \{\boxtimes_1, \boxtimes_2, \boxplus_1, \boxplus_2, \mu\}$,
  which all identify $t$ as a value and provide information about its structure; for example, a transition
   $t\xto{\boxplus_1} s$ means that $t=\inl_{\tau_1,\tau_2}(s)$, and
  ${t\xto{\boxtimes_1} s}$ means that
  $t=\mathsf{pair}_{\tau_1,\tau_2}(s,e)$ for some $e$.
\end{itemize}

These transitions are deterministic: every
term $t$ either reduces to a unique term $s$, i.e. $t \to s$, or exhibits a unique
labelled transition $t \xto{e} t_{e}$ for each appropriately typed term $e$, or exhibits a unique transition $t\xto{l} s$ where  $l \in \{\boxtimes_1, \boxtimes_2, \boxplus_1, \boxplus_2, \mu\}$.

The semantics of \stscr prominently features labelled transitions
and forms a ``higher-order LTS'', i.e.\ a labelled transition system on terms whose labels may also be terms. This style of semantics draws inspiration from
the work of \citet{Abramsky:lazylambda} and
\citet{DBLP:journals/tcs/Gordon99} on the $\lambda$-calculus.
The incorporation of the auxiliary operators $S'$, $S''$ and $K'$ does not
alter the functional
behavior of programs, except for possibly adding more unlabelled
transitions. For example, the conventional rule $S\, t\, s\,
e\to (t\,e)\,(s\,e)$ for the $S$-combinator~\cite{hindley2008lambda} is rendered as the sequence of transitions
\[
S\,t\,s\,e\to S'(t)\,s\,e\to S''(t,s)\, e\to (t\,e)\,(s\,e).
\]
To model \stscr in higher-order abstract GSOS, we take the base category $\C=\Set^\Ty$ of $\Ty$-sorted sets, the polynomial functor $\Sigma\colon \Set^\Ty\to \Set^\Ty$ corresponding to the signature of \stscr, and the behaviour bifunctor $B \c (\Set^{\Ty})^{\opp} \times \Set^{\Ty}
\to \Set^{\Ty}$ given by
\begin{equation}
  \label{eq:beh}
\begin{aligned}
    B_\tau(X,Y) &\,= Y_\tau + D_\tau(X,Y),\\[.5ex]
    D_{\arty{\tau_{1}}{\tau_{2}}}(X,Y) &\,= Y_{\tau_{2}}^{X_{\tau_{1}}},&\hspace{-1.5em} 
    D_{\tau_1\boxplus\tau_2}(X,Y)  &\,= Y_{\tau_1}+Y_{\tau_2}, \\
    D_{\mu\alpha.\,\tau}(X,Y)
                    &\,= Y_{\tau[\mu\alpha.\tau/\alpha]},&\hspace{-1.5em}
    D_{\tau_1\boxtimes\tau_2}(X,Y) &\,= Y_{\tau_1}\times Y_{\tau_2}.
\end{aligned}
\end{equation}
(We denote the components of a functor $F\colon \D\to \Set^\Ty$ by $F_\tau\colon
\D\to \Set$ for $\tau\in \Ty$.) In $Y_{\tau} + D_{\tau}(X,Y)$, $Y_{\tau}$ models
$\beta$-reduction and $D_{\tau}(X,Y)$ the information carried by values. The
deterministic transition rules of \eqref{fig:skirules} induce a coalgebra
$\gamma\colon \Tr \to B(\Tr,\Tr)$ given by 
\begin{flalign}
    \gamma_\tau(t)=\;& t'
    &&\hspace{-1.2em} \text{if $t \xto{} t'$, for $t,t'\c\tau$,}&&\notag\\*  
    ~\gamma_{\arty{\tau_{1}}{\tau_{2}}}(t) =\;& \lambda e.\,t_{e}
    &&\hspace{-1.2em}\text{if
      $t \xto{e} t_{e}$ for $t \c \arty{\tau_{1}}{\tau_{2}}$, $e\c \tau_1$,}&&\notag\\*
    \gamma_{\tau_1\boxplus\tau_2}(t)=\;& t'
    &&\hspace{-1.2em} \text{if $t \xto{\boxplus_i} t'$ for $t\c\tau_1\boxplus\tau_2, t'\c\tau_i$,}\label{exa:gamma} && \\*    
    \gamma_{\tau_1\boxtimes\tau_2}(t)=\;& (t_1,t_2)
    &&\hspace{-1.2em} \text{if $t \xto{\boxtimes_1} t_1$, $t \xto{\boxtimes_2} t_2$, for $t_i\colon \tau_i$,}\notag && \\*
    \gamma_{\mu \alpha.\,\tau}(t) =\;& t'
    &&\hspace{-1.2em}\text{if
      $t \xto{\mu} t'$ for $t \c \mu \alpha.\,\tau$, \makebox[0pt][l]{$t' \c \tau[\mu \alpha.\,\tau/\alpha]$.}}&&\notag
  \end{flalign}
We omit explicit coproduct injections for better readability, e.g.\ the term $t'$ in the last clause lies in the second summand of the coproduct $B_{\mu \alpha.\tau}(\Tr,\Tr)=\Tr_{\mu \alpha.\tau} + \Tr_{\tau[\mu \alpha.\,\tau/\alpha]}$. The coalgebra $\gamma \c \Tr \to B(\Tr,\Tr)$ is the canonical
operational model of a ($0$-pointed) higher-order GSOS law of the syntax functor
$\Sigma$ over the behaviour bifunctor $B$, i.e.~a family of $\Ty$-sorted maps
\begin{equation}
  \label{eq:rhotcl}
  \rho_{X,Y} \c \Sigma (X \times B(X,Y))\to B(X, \Sigma^\star (X+Y))
\end{equation}
dinatural in $X\in \Set^\Ty$ and natural in $Y\in \Set^\Ty$. The components 
of $\rho$ are given by case distinction over the operations of \stscr{} and simply encode the rules of \autoref{fig:skirules} as functions. We list a few
selected clauses below, again omitting coproduct injections; see \ifarx{the appendix}{\cite[Appendix A]{gmtu24_arxiv}} for a complete definition.
\begin{align*}
\rho_{X,Y}(S_{\tau_{1},\tau_{2},\tau_{3}}) =&\; \lambda e.\,S'_{\tau_{1},\tau_{2},\tau_{3}}(e),\\
\rho_{X,Y}(S'_{\tau_{1},\tau_{2},\tau_{3}}(t,f)) =&\; \lambda e.\,S''_{\tau_{1},\tau_{2},\tau_{3}}(t,e),\\
\rho_{X,Y}(S''_{\tau_{1},\tau_{2},\tau_{3}}((t,f),(s,g))) =&\; \lambda e.\,\mathsf{app}_{\tau_{2},\tau_{3}}(\mathsf{app}_{\tau_1,\arty{\tau_{2}}{\tau_{3}}}(t, e),\\*&\hspace{5.2em}\mathsf{app}_{\tau_1,\tau_2}(s, e)),\\
\rho_{X,Y}(\inl_{\tau_1,\tau_2}(t,f)) 						         =&\; t,\\
\rho_{X,Y}(\mathsf{app}_{\tau_1,\tau_2}((t,f),(s,g)))									 =&\; \begin{cases}
																																							 \makebox[5em][l]{$f(s)$} \text{\quad if $f \in Y_{\tau_{2}}^{X_{\tau_{1}}}$}, \\
																																							 \makebox[5em][l]{$\mathsf{app}_{\tau_1,\tau_2}(f, s)$} \text{\quad if $f \in Y_{\arty{\tau_{1}}{\tau_{2}}}$}.
																																						  \end{cases} 
\end{align*}

\subsection{Logical Relations for \texorpdfstring{$\boldsymbol{\mu}$}{$\mu$}\textbf{TCL}}
\label{sec:tclreason}

A natural operational notion of program equivalence is given by \emph{contextual} or
\emph{observational equivalence}, a relation that identifies programs if they
behave the same in all program contexts. It comes in two different flavours: \emph{Morris-style} contextual equivalence~\cite{morris}, and a more abstract, relational approach by
\citet{10.5555/309656.309660} and \citet{lassenthesis} (see
also~\cite{10.5555/1076265}). For the purposes of \stscr we pick the
former; we elaborate on the connections in~\Cref{sec:ctxeq}.
\begin{notation}\label{not:weak-trans}
\begin{enumerate}
\item  A \emph{program context} $C\colon \cty{\tau_{1}}{\tau_{2}}$ is a $\Sigma$-term of output type $\tau_2$ with a hole of type $\tau_1$, i.e.\ if $\mathds{1}_{\tau_1}$ denotes the $\Ty$-sorted set with a single element `$\cdot$' in sort $\tau_1$ and empty otherwise,~$C$ is a term in $(\Sigmas \mathds{1}_{\tau_1})_{\tau_2}$ with at most one occurrence of the variable~`$\cdot$'. We let $C[t]$ denote the result of substituting $t\in \Tr_{\tau_1}$ for the hole.
\item \label{not:bigstep} 
  We use the following notations for \emph{weak transitions}:
\begin{itemize}
\item  $\To$ for the reflexive, transitive hull of the reduction
  relation $\to$;
\item $t \xTo{l} s$ if $t\To t'\xto{l} s$ for some
  $t'$ and $l \in \Tr \cup \{\boxtimes_1, \boxtimes_2, \boxplus_1,
  \boxplus_2, \mu\}$;
\item $t{\Downarrow}$ if $t$ terminates, i.e.\
  $t\xTo{l} s$ for some $l$ and $s$. 
\end{itemize}
\end{enumerate}
\end{notation}
The \emph{contextual preorder} $\lesssim$ and \emph{contextual equivalence} $\simeq$ are the $\Ty$-sorted relations on $\Tr$ given by
\begin{align}
  t \lesssim_{\tau} s &\text{\quad if \quad}
  \forall \tau'.\forall C \c \cty{\tau}{\tau'}.\, C[t] {\Downarrow}
  \ \implies\, C[s] {\Downarrow}, \label{eq:ctxprox}\\
  t \simeq_{\tau} s&\text{\quad if \quad}
  \forall \tau'.\forall C \c \cty{\tau}{\tau'}.\, C[t] {\Downarrow}
  \iff C[s] {\Downarrow}.  \label{eq:ctxequiv}
\end{align}
Since direct reasoning on $\lesssim$ and $\simeq$ is difficult, we introduce a \emph{(step-indexed) logical relation} for \stscr, which will give rise to a sound proof method for the contextual preorder (\cref{cor:sound}):
\begin{definition}\label{def:logrel} The \emph{step-indexed logical relation} $\logrel$ for \stscr is the family of relations $(\logrel^\alpha \seq \Tr\times \Tr)_{\alpha\leq \omega}$ defined inductively by
\begin{flalign*}
 \logrel^{0} =&\; \top,&\\
 \quad\logrel^{n+1} =&\; \logrel^{n} \cap \mathcal{E}(\logrel^{n}) \cap \mathcal{V}(\logrel^{n},\logrel^{n}),& \logrel^{\omega} =&\; \bigcap_{n < \omega} \logrel^{n},\quad
\end{flalign*}
where $\mathcal{E} \c \RelAt{\Tr} \to \RelAt{\Tr}$ and $\mathcal{V} \c (\RelAt{\Tr})^{\opp} \times \RelAt{\Tr} \to \RelAt{\Tr}$ are the monotone maps given as follows:
  \begin{align*}
    & \mathcal{E}_{\tau}(R) = \{(t,s)
      \mid \text{if $t \to t'$ then $\exists s'.\,s \To s' \land R_{\tau}(t',s')$}\} \\
    & \mathcal{V}_{\tau_1\boxplus\tau_2}(Q,R) = \{(t,s)
      \mid \text{if $t \xto{\boxplus_{1}} t'$ then $\exists s'.\,s \xTo{\boxplus_{1}} s'$}
      \land R_{\tau_{1}}(t',s'), \\
    & \qquad \text{if $t \xto{\boxplus_{2}} t'$ then $\exists s'.\,s \xTo{\boxplus_{2}} s'$}
      \land R_{\tau_{2}}(t',s') \}\\
    & \mathcal{V}_{\tau_1\boxtimes\tau_2}(Q,R) = \{(t,s)
      \mid \text{if $t \xto{\boxtimes_{1}} t_{1} \land t \xto{\boxtimes_{2}} t_{2}$ then} \\
    & \qquad \exists s_{1},s_{2}.\,s \xTo{\boxtimes_{1}} s_{1} \land s \xTo{\boxtimes_{2}} s_{2}
      \land R_{\tau_{1}}(t_{1},s_{1}) \land R_{\tau_{2}}(t_{2},s_{2}) \}\\
    & \mathcal{V}_{\arty{\tau_{1}}{\tau_{2}}}(Q,R) = \{(t,s)
      \mid \text{for all $e_{1},e_{2} \c \tau_{1},\,Q_{\tau_{1}}(e_{1},e_{2}),$} \\
    & \qquad \text{if $t \xto{e_{1}} t'$ then $\exists s'.\,s \xTo{e_{2}} s' \land R_{\tau_{2}}(t',s')$} \}\\
    & \mathcal{V}_{\mu\alpha.\tau}(Q,R) = \{(t,s)
      \mid \text{if $t \xto{\mu} t'$ then} \\*
    & \qquad \exists s'.\,s \xTo{\mu} s' \land R_{\tau[\mu\alpha.\tau/\alpha]}(t',s') \}
    \end{align*}
\end{definition} 

\begin{remark}\label{rem:logrel-props}
\begin{enumerate}\item\label{rem:logrel-props-1} While details can vary greatly in the construction of logical relations in the literature, their common key feature is that related terms of function type send related inputs to related outputs. In our case, this is reflected by the definition of $\mathcal{V}_{\arty{\tau_{1}}{\tau_{2}}}$.
\item\label{rem:logrel-props-2} One may think of extending $(\logrel^\alpha)_{\alpha\leq \omega}$ beyond $\omega$ by putting 
\[ \logrel^{\omega+1} = \logrel^{\omega}
      \cap \mathcal{E}(\logrel^{\omega})
      \cap \mathcal{V}(\logrel^{\omega},\logrel^{\omega}) \]
etc. It is, however, easy to verify that $\logrel^{\omega+1}=\logrel^\omega$. In contrast, in effectful (e.g.\ non-deterministic or probabilistic) settings, a transfinite extension may be required~\cite{SpiesKrishnaswamiEtAl21,DBLP:journals/pacmpl/AguirreB23}. Our generic logical relations introduced in~\autoref{sec:logrel} thus use indexing by arbitrary ordinals.
\item\label{rem:logrel-props-3} 
In simply typed languages, such as the fragment of \stscr without recursive
types (and with an explicit $\utype$ type), it is possible to construct an alternative
logical relation $\ol{\logrel}\seq \Tr\times \Tr$ more directly via structural induction over types. For instance, assuming that $\ol{\logrel}_{\tau_1}$ and $\ol{\logrel}_{\tau_2}$ have already been defined,  $\ol{\logrel}_{\arty{\tau_1}{\tau_2}}$ is given by all $(t,s)\in \Tr_{\arty{\tau_1}{\tau_2}}\times \Tr_{\arty{\tau_1}{\tau_2}}$ such that 
\[ \text{if~~ $t\xTo{e_1} t'\,\wedge\,\ol{\logrel}_{\tau_1}(e_1,e_2)$ ~~then~~ $\exists s'.\, s\xTo{e_2} s'\,\wedge \,\ol{\logrel}_{\tau_2}(t',s')$}.\]
This approach does not extend to untyped languages, or languages with unrestricted recursive types like \stscr, which motivates step-indexing.
From a more abstract perspective emphasized by \citet{gmstu24}, the above inductive
definition of $\ol{\logrel}$ is possible because behaviour functors $B$ for simply
typed languages are contractive w.r.t.\ a suitable
{ultrametric}~\cite{235f3d3a0dee4f339a8f59beff18cae3} on the subobject
lattices of $\C=\Set^\Ty$. This fails in presence of recursive types.
\end{enumerate}
\end{remark}
The key property of $\logrel$ is its compatibility with all operations of~$\Tr$, the initial algebra of closed \stscr-terms:

\begin{theorem}
  \label{th:l-is-cong}
 For all $\alpha\leq\omega$ the relation $\logrel^{\alpha}$ is a congruence on $\Tr$.
\end{theorem}
\begin{proof}[Proof sketch]
We show that each $\logrel^n$ ($n<\omega$) is a congruence; this
implies that $\logrel^\omega$ is a congruence, since congruences are
closed under intersection. The proof is by induction on $n$. The base
case~\mbox{$n=0$} is trivial, since $\logrel^0=\top$. For the inductive
step $n\to n+1$, suppose that $\logrel^n$ is a congruence. To show that $\logrel^{n+1}$ is a congruence, we need to prove that for each $\Sigma$-operation $\mathsf{f}\colon \tau_{1} \times \dots \times \tau_{k} \to \tau$ and all
  terms $t_{i},s_{i} \c \tau_{i}$ such that
  $\logrel^{n+1}_{\tau_{i}}(t_{i},s_{i})$ for $i=1,\ldots,k$, we have 
$\logrel^{n+1}_{\tau}(\mathsf{f}(t_{1},\dots,t_{k}),
    \mathsf{f}(s_{1},\dots,s_{k}))$. This is equivalent to
  \begin{enumerate}
  \item \label{cong:enum1a}
    $\logrel^{n}_{\tau}(\mathsf{f}(t_{1},\dots,t_{k}),
    \mathsf{f}(s_{1},\dots,s_{k}))$.
  \item \label{cong:enum2a}
    $(\mathsf{f}(t_{1},\dots,t_{k}),
    \mathsf{f}(s_{1},\dots,s_{k}))\in \mathcal{E}_{\tau}(\logrel^{n})$;
  \item \label{cong:enum3a} $(\mathsf{f}(t_{1},\dots,t_{k}),
    \mathsf{f}(s_{1},\dots,s_{k}))\in \mathcal{V}_{\tau}(\logrel^n,\logrel^{n})$.
  \end{enumerate}
Statement \ref{cong:enum1a} holds because $\logrel^{n+1}\seq
\logrel^n$ and $\logrel^n$ is a congruence. Statements
\ref{cong:enum2a} and \ref{cong:enum3a} require a long case
distinction over the 15 operation symbols $\f$ of $\Sigma$.
Let us illustrate the case of application. We need to show that
    $\logrel^{n+1}_{\arty{\tau_{1}}{\tau_{2}}}(t_{1},s_{1})$ and $\logrel^{n+1}_{{\tau_{1}}}(t_{2},s_{2})$   implies \ref{cong:enum2a} $(t_1 \app{}{} t_2, s_1\app{}{} s_2)\in \mathcal{E}_{\tau_2}(\logrel^{n})$ and \ref{cong:enum3a} $(t_1 \app{}{} t_2, s_1\app{}{} s_2)\in \mathcal{V}_{\tau_2}(\logrel^n,\logrel^{n})$. Note that \ref{cong:enum3a} holds vacuously as applications do not admit labelled transitions. For \ref{cong:enum2a}, suppose that $t_{1} \app{}{} t_{2} \to t$. By the semantics of application (\autoref{fig:skirules}) such a transition may occur for two possible reasons:
   \begin{itemize} \item \textbf{Case 1:}
      $t_{1} \to t_{1}'$ and $t = t_{1}' \app{}{} t_{2}$. By
      $\logrel^{n+1}_{\arty{\tau_{1}}{\tau_{2}}}(t_{1},s_{1})$, there exists
      $s_{1}'$ such that $s_{1} \To s_{1}'$ and
      $\logrel^{n}_{\arty{\tau_{1}}{\tau_{2}}}(t_{1}',s_{1}')$. Since $s_{1} \app{}{} s_{2} \To s_{1}'\app{}{}s_{2}$, it suffices to show $\logrel^{n}_{\tau_{2}}(t_{1}' \app{}{}
      t_{2},s_{1}'\app{}{}s_{2})$. This holds because $\logrel^{n}_{\arty{\tau_{1}}{\tau_{2}}}(t_{1}',s_{1}')$ and $\logrel^{n+1}_{{\tau_{1}}}(t_{2},s_{2})$ (hence $\logrel^{n}_{{\tau_{1}}}(t_{2},s_{2})$) and $\logrel^{n}$ is a congruence.
      \item \textbf{Case 2:} $t_{1} \xto{t_{2}} t$.
      Since $\logrel^{n+1}_{\arty{\tau_{1}}{\tau_{2}}}(t_{1},s_{1})$ and $\logrel^{n+1}(t_2,s_2)$ (hence $\logrel^{n}(t_2,s_2)$), there exists $s$ such that $s_{1} \xTo{s_{2}} s$ and $\logrel^n(t,s)$. Moreover $s_1\app{}{} s_2\To s$, which proves the claim.\qedhere
\end{itemize}
\end{proof}

A complete proof can be found in \ifarx{the appendix}{\cite[Appendix A]{gmtu24_arxiv}}. An immediate consequence of \Cref{th:l-is-cong} is the so-called
\emph{Fundamental Property} of $\logrel^{\omega}$, namely that all terms are related to themselves:

\begin{corollary}
  The relation $\logrel^{\omega}$ is reflexive.
\end{corollary}
\noindent
Indeed, all congruences on an initial (term) algebra are reflexive (\autoref{sec:categories}). As another important corollary of \Cref{th:l-is-cong}, we conclude that $\logrel$ is sound for the contextual preorder:
\begin{corollary} \label{cor:sound} For all $\tau\in \Ty$ and terms $t,s\in \Tr_\tau$, we have
\begin{equation}\label{eq:soundness} \logrel^{\omega}_{\tau}(t,s)
  \implies t \lesssim_{\tau} s.
\end{equation}
\end{corollary}
\begin{proof}
\begin{enumerate}
\item Note first that $\L^\omega$ is \emph{adequate} w.r.t.\ termination, i.e.\ 
 $\logrel^{\omega}_{\tau}(t,s)$ implies that if $t{\Downarrow}$ then $s {\Downarrow}$.
To see this, let $\logrel^{\omega}_{\tau}(t,s)$ and suppose that $t{\Downarrow}$, say $t$ terminates in $n$ steps in the term $t'$. Since $\L^{n+1}_\tau(t,s)$, there exists $s'$ such that $s\To s'$ and $\L^{1}_\tau(t',s')$. Since $t'$ has a labelled transition, we get $s'{\Downarrow}$ by definition of $\L^{1}$, hence $s{\Downarrow}$.
\item To prove \eqref{eq:soundness}, suppose that $\logrel^{\omega}_{\tau}(t,s)$. Since $\logrel^\omega$ is a congruence, it follows that $\logrel^\omega_{\tau'}(C[t],C[s])$ for every context $C\colon \cty{\tau}{\tau'}$. Hence, if $C[t]{\Downarrow}$ then $C[s]{\Downarrow}$ by adequacy, which proves $t\lesssim_\tau s$.\qedhere
\end{enumerate}
\end{proof}

\begin{remark}\label{rem:logrel-backwards-closed}
In order to prove $t\lesssim s$, it thus suffices to prove $\L^\omega(t,s)$, or equivalently $\L^n(t,s)$ for all $n<\omega$. 
A useful observation for such proofs is that $\logrel^n$ is backwards closed under silent transitions: if $t\To t'$ and $s\To s'$, 
then $\logrel^n(t',s')$ implies $\logrel^n(t,s)$.
\end{remark}

\begin{example}\label{ex:eta-left-to-right}
We put \autoref{cor:sound} to the test by proving 
\begin{equation}\label{eq:eta-left-to-right}
  f \lesssim S \app{}{} (K \app{}{} I) \app{}{} f
  \quad  \text{for all $f \c \arty{\tau_{1}}{\tau_{2}}$}.
\end{equation}
The term $S \app{}{} (K \app{}{} I) \app{}{} f$ behaves like the $\lambda$-term $\lambda x.\, f x$, as
\[ 
(S \app{}{} (K \app{}{} I) \app{}{} f)\, t\To (K \app{}{}
I\, t)(f\, t)\To I\, (f \, t)\to f\,t,
\]
so \eqref{eq:eta-left-to-right} is an analogue of one half of the $\eta$-law of the $\lambda$-calculus. It suffices to prove by induction that $\logrel^n(f,S \app{}{} (K \app{}{} I) \app{}{} f)$ for all $n<\omega$. We prove a slightly stronger statement, namely $\logrel^n(t,S \app{}{} (K \app{}{} I) \app{}{} f)$ whenever $f\To t$. The base case $n=0$ is trivial, as $\logrel^0=\top$. For the inductive step $n\to n+1$, we consider two cases:
\begin{enumerate}
\item $t\to t'$. Then $S \app{}{} (K \app{}{} I)\,f \To S \app{}{} (K \app{}{}
    I)\, f$, and $\logrel_{\arty{\tau_{1}}{\tau_{2}}}^{n}(t', S \app{}{} (K
    \app{}{} I)\, f)$ by the inductive hypothesis as $f \To t'$, hence $\logrel^{n+1}(t,S \app{}{} (K \app{}{} I) \app{}{} f)$.
\item $t \xrightarrow{e} t'$. Then
$ S \app{}{} (K \app{}{} I)\, f
    \xTo{e'} (K \app{}{} I \app{}{} e') \app{}{} (f \app{}{} e') \To f
    \app{}{} e' \To t \app{}{} e' \to t''$ for every $e'$,
where $t\xto{e'} t''$. To prove $\logrel^{n+1}(t,S \app{}{} (K \app{}{} I) \app{}{} f)$ we need to show that if $\logrel^n(e,e')$ then $\logrel^n(t', (K \app{}{} I \app{}{} e') \app{}{} (f \app{}{} e') )$. By \autoref{rem:logrel-backwards-closed} it suffices to show $\logrel^n(t', t'')$; this holds as $\logrel^{n+1}(t,t)$ by reflexivity.
\end{enumerate}
\end{example}
The other direction  $S \app{}{} (K \app{}{} I) \app{}{} f \lesssim f$ of the $\eta$-law generally fails because $S \app{}{} (K \app{}{} I) \app{}{} f$ always terminates, while $f$ may diverge. However, we shall see below that it does hold when restricting to contexts of ground type like the type of
booleans. Specifically, we define the \emph{ground contextual preorder} and \emph{ground contextual
equivalence} by
\begin{align}
  t \lesssim_{\tau}^{\mathrm{\booltype}} s &\text{\quad if \quad}
 \forall C \c \cty{\tau}{\booltype}.\, C[t] {\Downarrow}
  \ \implies\, C[s] {\Downarrow},\label{eq:ctxproxground}\\
  t \simeq_{\tau}^{\mathrm{\booltype}} s&\text{\quad if \quad}
  \forall C \c \cty{\tau}{\booltype}.\, C[t] {\Downarrow}
  \iff C[s] {\Downarrow}.  \label{eq:ctxequivground}
\end{align}
The ground contextual preorder is (subtly) coarser than the contextual preorder, see e.g.~\cite[\textsection 5]{pitts1997operationally}. To get a sufficiently expressive logical relation for $\lesssim^\booltype$, we extend the weak transition relation~$\xTo{l}$:
\begin{notation}
  \label{not:myarr}
For $t,s\in \Tr_{\arty{\tau_1}{\tau_2}}$ and $e\in \Tr_{\tau_1}$ put
\[ t\myarr{e} s \quad\iff\quad \exists t'.\, (t\To t') \wedge (t'\xto{e} s \vee s=t'\app{}{}e).\]
For terms which are not of function type, put $t\myarr{l} s$ iff $t\xTo{l} s$. 
\end{notation}

\begin{definition}\label{def:logreltwo}
The step-indexed logical relation $\logreltwo$ for \stscr is constructed like $\logrel$, except that all occurrences of weak labelled transitions $\xTo{l}$ in the definitions of $\mathcal{E}$ and $\mathcal{V}$ are replaced by $\myarr{l}$.
\end{definition}

Again we obtain compatibility with the language operations:
\begin{theorem}
  \label{th:l2-is-cong}
  For all $\alpha\leq \omega$ the relation $\logreltwo^{\alpha}$ is a congruence on~$\Tr$.
\end{theorem}
\noindent
The proof works similarly to that of \Cref{th:l-is-cong}, modulo handling the extra
transitions in cases such as application. The relation $\logreltwo$ is sound for the ground contextual preorder:

\begin{corollary} \label{cor:sound2} For all $\tau\in \Ty$ and terms $t,s\in \Tr_\tau$, we have
\[ \logreltwo^{\omega}_{\tau}(t,s)
  \implies t \lesssim_{\tau}^\booltype s.\]
\end{corollary}
\begin{proof}
  Unlike $\logrel^{\omega}$, the relation $\logreltwo^{\omega}$ is not adequate to observe
  termination at all types. It is, however, adequate at $\booltype$:
 \[
    \logreltwo^{\omega}_{\booltype}(t,s) \implies (t{\Downarrow} \implies s {\Downarrow}).
 \]
The remaining argument is like in the proof of \autoref{cor:sound}.
\end{proof}

The extra permissiveness of
$\myarr{l}$ allows potentially diverging terms $t \c \arty{\tau_{1}}{\tau_{2}}$ to
be tested and compared against terminating terms, in a manner similar to
\cite[\textsection 4]{DBLP:journals/tcs/Gordon99}. In particular, it
helps us to establish the other direction of the $\eta$-law:
\begin{example}\label{ex:eta-full}
We prove the full $\eta$-law w.r.t.\ ground equivalence:
\begin{equation}
  f \simeq^\booltype S \app{}{} (K \app{}{} I) \app{}{} f
  \quad  \text{for all $f \c \arty{\tau_{1}}{\tau_{2}}$}.
\end{equation}
We already know $f\lesssim S \app{}{} (K \app{}{} I) \app{}{} f$ from \autoref{ex:eta-left-to-right}, which implies $f\lesssim^\bool S \app{}{} (K \app{}{} I) \app{}{} f$ because $\lesssim^\bool$ is coarser than $\lesssim$. For the reverse direction $S \app{}{} (K \app{}{} I) \app{}{} f \lesssim^\bool f$, it suffices to show $\logreltwo^{n+1}(S \app{}{} (K \app{}{} I) \app{}{} f,  f)$ for all $n<\omega$. We have
\[ S \app{}{} (K \app{}{} I) \app{}{} f \To S''(K\app{}{} I, f) \xto{e} (K\, I\, e)\, (f\, e) \To f\, e. \]
Note that analogous to \autoref{rem:logrel-backwards-closed}, the
relations $\logreltwo^{(-)}$ are backwards closed under silent
transitions. Hence it is enough to prove $\logreltwo^{n+1}(S''(K\, I,
f), f)$. Thus let $\logreltwo^n(e,e')$. Then $f\myarr{e'} f\, e'$, and
it remains to prove $\logreltwo^n((K\, I\, e)\, (f\, e),f\, e')$. For
this it suffices to prove $\logreltwo^n(f\, e, f\, e')$, again using
backwards closure, and this holds because $\logreltwo^n(f,f)$ by reflexivity,  $\logreltwo^n(e,e')$, and
$\logreltwo^n$ is a congruence.
\end{example}

At this point it is worth mentioning some of the difficulties arising from working with logical relations, which the developments for \stscr help
identify and which our theory aims to address. 

First, the principal technical challenge associated with any logical relation is establishing its congruence property. The argument typically follows the structure of the proof of \autoref{th:l-is-cong}: one proceeds via an ultimately straightforward, yet tedious, structural induction over the syntax of the language, requiring meticulous case distinctions along the various operations and their operational rules. In the literature, the individual cases are often organized into separate \emph{compatibility lemmas}. The length and complexity of the corresponding proofs makes them error-prone and arguably hard to trust without formalization in a proof assistant.

Second, another layer of reasoning is required for proving soundness w.r.t.\ contextual equivalence, as in \Cref{cor:sound,cor:sound2}.

Lastly, logical relations are tailor-made for the language under consideration and for the desired notion of contextual equivalence. Every small variation of the setting requires the construction of a separate logical relation, along with new and largely repetitive proofs of the congruence and soundness properties, as illustrated above when passing from $\logrel$ to $\logreltwo$.  

In the remainder we will demonstrate how our categorical approach to operational semantics can help mitigate these issues: for all languages modeled in higher-order abstract GSOS, a \emph{generic} construction of step-indexed logical relations is available, associated with \emph{generic} congruence and soundness theorems.

\section{Logical Relations, Abstractly}\label{sec:logrel}
We present our main technical contribution, a theory of contextual preorders and
(step-indexed) logical relations in the framework of higher-order abstract GSOS. To ensure a convenient calculus of relations, we work under the following global assumptions:

\begin{assumptions}
  \label{assumptions}
  We hereafter fix $\C$ to be a category such that: (1)~$\C$ is complete, (2)~$\C$ has finite coproducts, (3)~$\C$ is well-powered, (4)~monomorphisms are
  stable under finite coproducts, and (5)~strong epimorphisms are stable
  under pullbacks.
\end{assumptions}

\begin{remark}
\begin{enumerate}
\item Recall that an epimorphism $e$ is \emph{strong} if for every commutative square $g\cdot e = m\cdot f$ with $m$ monic, there exists $d$ such that $f=d\cdot e$ and $g=m\cdot d$. Condition (5) means that for every pullback square $e\cdot \ol{f} = f\cdot \ol{e}$, if $e$ is strongly epic, then so is $\ol{e}$.
\item Since $\C$ is complete and well-powered, every morphism $f$ admits a (strong epi, mono)-factorization $f=m\cdot e$~\cite[Prop.~4.4.3]{borceux94}. The subobject represented by $m$ is called the \emph{image} of $f$. All our results can be extended to arbitrary proper factorization systems~\cite{ahs90}.
\end{enumerate}
\end{remark}

\begin{example}\label{ex:categories}
  Categories satisfying our assumptions include the category $\Set$ of
  sets and functions, the category $\Set^{\D}$ of (covariant)
  presheaves on a small category $\D$ and natural transformations, and
  the categories of posets and monotone maps, nominal sets and
  equivariant maps, and metric spaces and non-expansive maps.
\end{example}

\subsection{Relations in Categories}\label{sec:relations}

We review some basic terminology concerning relations in abstract categories. A \emph{relation} on $X\in \C$ is a subobject $\langle \outl_{R}, \outr_{R} \rangle \colon R
\pred{} X \times X$; the projections $\outl_R$ and $\outr_R$ are usually left implicit. A \emph{morphism} from a relation $R
\pred{} X \times X$ to another relation $S
\pred{} Y \times Y$ is given by a morphism $f \colon X \to Y$ in $\C$
such that there exists a (necessarily unique) morphism
$R \to S$ rendering the square below commutative:
\[
\begin{tikzcd}[column sep=5em]
	R \ar[d,tail,"{\langle \outl_{R}, \outr_{R} \rangle}"'] \ar[r,dashed] & S
  \ar[d,tail,"{\langle \outl_{S}, \outr_{S} \rangle}"]\\
	X \times X \ar[r,"f \times f"] & Y \times Y
\end{tikzcd}
\]
We write $\RelCat \C$ for the category of relations in $\C$ and their morphisms. Finite products and coproducts in $\RelCat{\C}$ are formed like in $\C$; for coproducts this is due to Assumption~\ref{assumptions}(4). For $X\in \C$ we denote by
$\RelCat[X]{\C} \subto \Rel(\C)$
the \emph{fiber} at $X$, viz. the non-full subcategory consisting of all relations $R \monoto X
\times X$ and morphisms on the identity $\id_{X} \c X \to X$.
Each fiber $\RelCat[X]{\C}$ is a poset; we write $R\leq S$ for its
partial order. As $\C$ is complete and well-powered, $\RelCat[X]{\C}$
is in fact a complete lattice. We denote the top and bottom element by
$\top$ and $\bot$, respectively, meets (which are given by pullbacks)
by $\bigwedge$ and joins by $\bigvee$.
For relations on (sorted) sets, joins and meets are (sortwise)
union and intersection.

Given $f,g\colon X\to Y$ in $\C$ and a relation $R\monoto X\times X$, we write $(f\times g)_\star[R]\monoto Y\times Y$ for the \emph{direct image} of $R$ under $f\times g$, i.e.\ the image of the morphism $\begin{tikzcd}[cramped] R \ar{r}{\langle \outl_R, \outr_R\rangle} & X\times X\ar{r}{f\times g} & Y\times Y\end{tikzcd}$, and we write~$f_\star$ for $(f\times f)_\star$. Similarly, given $S\monoto Y\times Y$ we let $(f\times g)^\star[S]\monoto X\times X$ denote the \emph{preimage} of $S$ under $f\times g$, i.e.\ the pullback of $\langle \outl_S,\outr_S\rangle$ along $f\times g$, and we put $f^\star$ for $(f\times f)^\star$. Note that $(f\times g)_\star$ is a left adjoint of $(f\times g)^\star$: one has $(f\times g)_\star[R]\leq S$ iff $R\leq (f\times g)^\star [S]$.

We denote the \emph{identity relation} $\langle \id,\id\rangle\colon X\monoto X\times X$ by $\Delta_X$, or just $\Delta$. The \emph{composite} of two relations $R,S\monoto X\times X$ is the relation $R\cdot S\monoto X\times X$ given by the image of the morphism \[\langle \outl_R\cdot \ol{\outl}_{R\smc S}, \outr_{S}\cdot \ol{\outr}_{R\smc S} \rangle\colon R\smc S\to X\times X,\] where $\ol{\outl}_{R\smc S}$ and $\ol{\outr}_{R\smc S}$ form the pullback of $\outr_R$ and $\outl_S$:
\[
\begin{tikzcd}[column sep=1em,row sep=0.65em]
  & &
  R\smc S %
  \pullbackangle{-90}
  \ar{dl}[swap]{\ol{\outl}_{R\smc S}}
  \ar{dr}{\ol{\outr}_{R\smc S}}
  \\
  & R
  \ar{dl}[swap]{\outl_R}
  \ar{dr}{\outr_R}
  & &
  S \ar{dl}[swap]{\outl_{S}}
  \ar{dr}{\outr_{S}}
  \\
  X & & X & & X 
\end{tikzcd}
\]
Composition defines a monotone map $\RelCat[X]{\C}\times \RelCat[X]{\C} \to \RelCat[X]{\C}$.
A relation $R$ is \emph{reflexive} if $\Delta\leq R$, \emph{transitive} if $R\cdot R\leq R$, and a \emph{preorder} if it is both reflexive and transitive.

\paragraph*{Preordered objects}
Our abstract congruence results involve objects whose generalized elements are suitably preordered. We recall the required terminology from \citet{UrbatTsampasEtAl23}.

\begin{definition} A \emph{preordered
  object} in $\C$ is a pair $(X,\preceq)$ of an object $X \in \C$ and
a family ${\preceq} = (\preceq_{Y})_{Y\in \C}$, where $\preceq_{Y}$ is a
{preorder} on the hom-set
$\C(Y,X)$ satisfying
\[ f\preceq_Y g \quad\implies\quad f\comp h\preceq_Z g\comp h \quad\text{for all
    $h\c Z\to Y$}.\]
We drop the subscript from $\preceq_Y$ when it is obvious from the context.
\end{definition}

\begin{example}
  \label{ex:preord}
  Every preordered set $(X,\preceq)$ can be viewed as a preordered
  object in $\Set$ by taking the pointwise preorder on $\Set(Y,X)$:
    \[
      f\preceq g \quad\iff\quad \forall y\in Y.\, f(y)\preceq g(y).
    \]
Note that there exist preordered objects in $\Set$ where the order on $\Set(Y,X)$ is not pointwise determined by that on $\Set(1,X)$; hence preordered objects are more general that preordered sets.
\end{example}

\begin{definition}\label{def:good-for-simulations}
  Let $(X,\preceq)$ be a preordered object. A relation $R \monoto X \times X$ is
  \emph{up-closed} if for every span $X\xleftarrow{f} S
  \xrightarrow{g} X$ and every morphism $S\to R$ such that the
  left-hand triangle in the first diagram below commutes, and the right-hand triangle
  commutes laxly as indicated, then there exists a morphism $S\to R$ such that the second
  diagram commutes.%
  \smnote{I expanded this since the formulation here was unclear and confusing.}
\[
  \begin{tikzcd}
    & \ar[bend right=2em]{dl}[swap]{f} \ar[bend left=2em]{dr}{g} \ar{d} S & {~}\\
    X & \ar[phantom]{ur}[description, pos=.35, xshift=-5]{\dleq{45}} \ar{l}[swap]{\outl_R} R \ar{r}{\outr_R} & X
  \end{tikzcd}
  \implies
  \begin{tikzcd}
    & \ar[bend right=2em]{dl}[swap]{f} \ar[bend left=2em]{dr}{g} \ar[dashed]{d} S & {~}\\
    X & \ar{l}[swap]{\outl_R} R \ar{r}{\outr_R} & X
  \end{tikzcd}
\]
\end{definition}
\citet{UrbatTsampasEtAl23} have called up-closed relations \emph{good for simulations}, as they admit a typical property of simulation relations.

\begin{example}\label{ex:up-closed}
Given a preordered set $(X,\preceq)$, regarded as a preordered object in $\Set$ as in \autoref{ex:preord}, a relation $R\seq X\times X$ is up-closed iff $R\smc {\preceq} \, \seq\, R$. For instance:
\begin{enumerate}
\item\label{ex:up-closed-2}
Every relation $R\seq X\times X$ induces a relation $\arP R \seq \Pow X\times \Pow X$ on the power set known as the \emph{(left-to-right) Egli-Milner relation}:
\[ \arP(A,B) \;\iff\; \forall a\in A.\, \exists b\in B.\, R(a,b).\]
This relation is up-closed w.r.t.\ $\seq$.
\item\label{ex:up-closed-3} In contrast, the \emph{two-sided Egli-Milner relation} given by
\[ \ol{\Pow}R(A,B) \;\iff\; \forall a\in A.\, \exists b\in B.\, R(a,b) \wedge \forall b\in B.\, \exists a\in A.\, R(a,b)\]
is not up-closed.
\end{enumerate}
\end{example}

\subsection{Congruences and Contextual Preorders}
\label{sec:ctxeq}
As highlighted in \autoref{sec:tclreason}, logical relations are employed as a sound proof method for contextual preorders. In the following we introduce the abstract contextual preorder that will feature in our generic soundness result. It is based on the categorical notion of congruence for functor algebras, which is most conveniently presented in terms of relation liftings of the underlying functor.

A \emph{relation lifting} of an endofunctor $\Sigma \colon \C \to \C$ is a functor $\overline \Sigma \colon  \RelCat \C \to \RelCat \C$ making the square below commute; here $\under{-}$ denotes the forgetful functor sending $R \pred{} X \times X$ to $X$. 
\begin{equation*}
	\label{eq:liftingRel}
	\begin{tikzcd}
		\RelCat{\C} \ar{d}[swap]{\under{-}}  \ar{r}{\ol{\Sigma}} & \RelCat{\C} \ar{d}{\under{-}}  \\
		\C \ar{r}{\Sigma}  & \C
	\end{tikzcd}
\end{equation*}
Every endofunctor
$\Sigma$ has a \emph{canonical relation lifting} $\overline \Sigma$, which takes a relation $R
\pred{} X \times X$ to the relation $\ol{\Sigma}R\monoto \Sigma X\times \Sigma X$ given by the image of the morphism $\langle \Sigma\outl_{R}, \Sigma\outr_{R}\rangle\colon \Sigma R\to \Sigma X\times \Sigma X$.
\begin{example}\label{ex:endofunctor-liftings}
\begin{enumerate}
\item\label{ex:endofunctor-liftings-polynomial} For a polynomial functor $\Sigma$ on $\Set^S$, the canonical lifting $\ol{\Sigma}$ sends $R\seq X\times X$ to the relation $\ol{\Sigma} R\seq \Sigma X\times \Sigma X$ relating $u,v\in (\Sigma X)_s$ iff $u=\f(x_1,\ldots,x_n)$ and $v=\f(y_1,\ldots,y_n)$ for some $\f\colon s_1\times\cdots \times s_n\to s$ in $\Sigma$, and $R_{s_i}(x_i,y_i)$ for $i=1,\ldots,n$.
\item\label{ex:endofunctor-liftings-pow} The canonical lifting $\ol{\Pow}$ of the power set functor $\Pow\colon\Set\to \Set$ takes a relation $R\monoto X\times X$ to the two-sided Egli-Milner relation $\ol{\Pow} R\seq \Pow X\times \Pow X$ (\autoref{ex:up-closed}\ref{ex:up-closed-3}). Taking the left-to-right Egli-Milner relation yields a non-canonical lifting, denoted by $\arP$.
\item\label{ex:endofunctor-liftings-typed-pow} For every set $S$ of sorts, the pointwise power set functor
\[ \Pow_\star\colon \Set^S\to \Set^S,\qquad \Pow_\star X=(\Pow X_{s})_{s\in S}, \]
has a relation lifting ${\arP}_\star$ taking $R\seq X\times X$ to $\arP_\star R = (\arP R_s)_{s\in S}$.
\end{enumerate}
\end{example}
In our applications, syntax endofunctors $\Sigma$ will always be equipped with their canonical relation lifting. This ensures that free $\ol{\Sigma}$-algebras emerge as liftings of free $\Sigma$-algebras:
\removeThmBraces
\begin{proposition}[{\cite[Prop.~V.4]{UrbatTsampasEtAl23}}]\label{prop:free-monad-lift}
Let\/ $\ol{\Sigma}$ and $\ol{\Sigmas}$ be the canonical liftings of $\Sigma$ and $\Sigmas$. If $\Sigma$ preserves strong epimorphisms, then $\ol{\Sigma}^\star = \ol{\Sigmas}$.
\end{proposition}
\resetCurThmBraces
Using canonical liftings, the notion of congruence from universal algebra extends to functor algebras~\cite{DBLP:books/cu/J2016}:
\begin{definition}\label{def:congruence}
  For an endofunctor $\Sigma$ with canonical lifting $\ol{\Sigma}$,
  a \emph{congruence} on a $\Sigma$-algebra $(A,a)$ is a relation
  $R\monoto A\times A$ such that
\[\fimg{a}{\ol{\Sigma}R}\leq R.\]
\end{definition}

\begin{example}
For a polynomial functor $\Sigma$ on $\Set^S$, a relation $R\seq A\times A$ is a congruence as per \autoref{def:congruence} iff it is a congruence in the usual sense, i.e.\ compatible with all $\Sigma$-operations. 
\end{example}

We will mainly consider congruences on the initial algebra $\mu \Sigma$, which admit a useful property; see e.g.\ \cite[Ex.~3.3.2]{DBLP:books/cu/J2016}:
\begin{proposition}\label{prop:cong-refl}
Every congruence on $\mu \Sigma$ is reflexive.
\end{proposition}
We are prepared to set up our abstract notion of contextual preorder, which is parametric in a choice of admissible observations.

\begin{definition}
Let $O\monoto \mS\times \mS$ be a preorder.
\begin{enumerate}
\item A relation $R \monoto \mS\times \mS$ is \emph{$O$-adequate} if $R\leq O$.
\item If a greatest $O$-adequate congruence on $\mS$ exists, then it
  is said to be the \emph{contextual preorder w.r.t.\ $O$} and denoted by \[{\cprd}\monoto\mS\times \mS.\]
\end{enumerate} 
\end{definition}
This terminology is justified by the following result:
\begin{proposition}\label{prop:contextual-preorder}
If\/ $\Sigma$ preserves strong epimorphisms and the relation
$\cprd$ exists, then it is a preorder.
\end{proposition}

\begin{example}\label{ex:observations}
\begin{enumerate}
\item\label{ex:observations-1} For every $S$-sorted signature $\Sigma$ and preorder $O\seq \mS\times \mS$, the contextual preorder $\cprd$ exists; it is given by
\begin{equation}\label{eq:cont-preord-concrete} t\cprd_{\tau} s \quad\text{iff}\quad \forall \tau'.\forall\, C \c \cty{\tau}{\tau'}.\, O_{\tau'}(C[t], C[s]) \end{equation}
in sort $\tau$. To see this, we take \eqref{eq:cont-preord-concrete} as the definition of $\cprd$ and prove it to be the greatest $O$-adequate congruence. Indeed, every $O$-adequate congruence $R$ is contained in $\cprd$: if $R_\tau(t,s)$  and $C \c \cty{\tau}{\tau'}$ we have $R_{\tau'}(C[t],C[s])$ because $R$ is a congruence, hence $O_{\tau'}(C[t], C[s])$ because $R$ is $O$-adequate, which proves $t\cprd_{\tau} s$. Conversely, it is not difficult to verify that the relation $\cprd$ defined by \eqref{eq:cont-preord-concrete} is itself an $O$-adequate congruence. Hence it is the greatest such relation.
\item\label{ex:observations-2} Taking $O_\tau=\{ (t,s) \mid t{\Downarrow}\,\To\, s{\Downarrow}  \}$\sgnote{Use $\mid$ instead of $\colon$ for consistency?} for all $\tau$, the relation $\cprd$ coincides with the contextual preorder of \stscr given by \eqref{eq:ctxprox}. 
\item\label{ex:observations-3}  Similarly, for $O_\booltype = \{ (t,s) \mid t{\Downarrow}\,\To\, s{\Downarrow}  \}$ and $O_{\tau}=\top_\tau$ if $\tau\neq \booltype$, we recover the ground contextual preorder given by \eqref{eq:ctxproxground}.
\end{enumerate}
\end{example}
In general the existence of a greatest $O$-adequate congruence~$\cprd$ is not obvious because a join of congruences need not be a congruence. However, we show next that it always exists for well-behaved categories and syntax functors:

\begin{rem}\label{rem:cat-properties}
\begin{enumerate}
\item\label{rem:cat-properties-1} A category $\gcat$ is
\emph{infinitary extensive}~\cite{cbl93} if it has small coproducts
and for every set-indexed family of objects $(X_i)_{i\in I}$ the functor
$E\colon \prod_{i\in I} \gcat/X_i \to \gcat/\coprod_{i\in I} X_i$ sending
$(p_i\colon Y_i\to X_i)_{i\in I}$ to
$\coprod_{i\in I} p_i\colon \coprod_i Y_i \to \coprod_i X_i$ is an
equivalence of categories. Extensivity hence ensures that coproducts behave
like disjoint unions. %
\item\label{rem:cat-properties-2} A diagram $D\colon I\to \C$ is \emph{directed} if $I$ is a directed poset, that is, $I$ is non-empty and any two elements $i,j\in I$ have an upper bound $k\geq i,j$. A \emph{directed colimit} is a colimit of a directed diagram.
\item\label{rem:cat-properties-3} The category $\C$ has \emph{smooth monomorphisms}~\cite{amm21} if for every object the join of a directed family of subobjects is given by the colimit of the corresponding directed diagram in $\C$.
\end{enumerate}
\end{rem}

\begin{theorem}[Existence of $\cprd$]\label{prop:existence-of-contextual-preorder}
Suppose that $\C$ is infinitary extensive and has smooth monomorphisms, and that $\Sigma$ preserves strong epimorphisms, monomorphisms, and directed colimits. Then for every preorder $O\monoto \mS\times \mS$ the contextual preorder $\cprd\,\monoto\ \mS\times \mS$ exists.
\end{theorem}
All categories of \autoref{ex:categories} satisfy the assumptions on $\C$, and all polynomial functors on $\Set^S$ satisfy those on $\Sigma$. The proof is a categorical generalization of an argument by \citet[Thm.~7.5.3]{10.5555/1076265}  for a fragment of the \textsf{ML} language. It rests on the lemma below; here $R^\star$ denotes the least preorder containing a given relation $R$.

\begin{lemma}\label{lem:union-transhull-cong}
In the setting of \autoref{prop:existence-of-contextual-preorder}, if\, $\R$ is a set of reflexive congruences on a $\Sigma$-algebra $(A,a)$, then $(\bigvee_{R\in \R} R)^\star$ is a congruence.
\end{lemma}

\begin{proof}[Proof of \autoref{prop:existence-of-contextual-preorder}]
Let $\R$ be the set of all $O$-adequate congruences on $\mS$. Then $\bigvee_{R\in \R} R$ exists as $\Rel_A(\C)$ is a complete lattice, $S=(\bigvee_{R\in \R}R)^\star$ is a congruence by \autoref{prop:cong-refl} and \autoref{lem:union-transhull-cong}, and $S$ is $O$-adequate because 
$S=\big(\bigvee_{R\in \R} R\big)^\star \leq O^\star = O$, using that each $R\in \R$ is adequate and $O$ is a preorder. Hence $S\in \R$, so $S$ is the greatest $O$-adequate congruence.
\end{proof}

\subsection{Logical Relations via Liftings}\label{sec:cong-logrel}
Next we develop our abstract categorical notion of (step-indexed) logical relation, with the goal of exposing a sound proof method for the contextual preorders $\cprd$ introduced above. Just like congruences rely on relation liftings of syntax endofunctors, logical relations are based on relation liftings of behaviour bifunctors. 

A \emph{relation lifting} of a mixed
variance bifunctor ${B\c \C^\opp\times \C\to \C}$ is a functor $\overline B \colon
{\RelCat \C}^\opp \times {\RelCat \C} \to \RelCat \C$ such that the square
\begin{equation*}
	\begin{tikzcd}
		\RelCat{\C}^\opp\times \RelCat{\C}
		\ar{d}[swap]{\under{-}^\opp\times \under{-}} \ar{r}{\ol{B}} & \RelCat{\C} \ar{d}{\under{-}}  \\
		\C^\opp \times \C \ar{r}{B}  & \C
	\end{tikzcd}
\end{equation*}
commutes.
Similar to the case of endofunctors, every bifunctor ${B\c \C^\opp\times \C\to \C}$ has a \emph{canonical relation
lifting} $\overline B$; it takes two relations $R\monoto X\times X$ and $S\monoto Y\times Y$ to the relation $\overline{B}(R,S) \monoto B(X,Y) \times B(X,Y)$ given by the image of the morphism $u_{R,S}$ in the pullback below.
This is an equivalent, albeit simplified,
version of a construction due to \citet[Sec.~C.2]{UrbatTsampasEtAl23}.
\begin{equation*}
  \label{eq:liftingpb}
  \hspace{-9pt}\begin{tikzcd}[column sep=.5em, row sep=.4ex]
    &[-.25em]& {T_{R,S}}
    \pullbackangle{-45}
    &&&& {B(R,S)} \\
    \\
    {\overline{B}(R,S)} \\
    \\
    && {B(X,Y) \times B(X,Y)} &&&& {B(R,Y) \times B(R,Y)}
    \arrow[two heads, from=1-3, to=3-1]
    \arrow[tail, pos=.8, from=3-1, to=5-3]
    \arrow["v_{R,S}", from=1-3, to=1-7]
    \arrow["u_{R,S}"', from=1-3, to=5-3]
    \arrow["{\langle B(\id,\outl_{S}), B(\id,\outr_{S}) \rangle}", from=1-7, to=5-7]
    \arrow["{B(\outl_{R},\id) \times B(\outr_{R},\id)}"{yshift=5pt}, from=5-3, to=5-7]
  \end{tikzcd}
\end{equation*}

\begin{example}\label{ex:lift-beh}
For the behaviour bifunctor $B$ on $\Set^\Ty$ given by~\eqref{eq:beh}, the canonical lifting $\ol{B}$ sends a pair of relations
 $R\seq X\times X$ and $S\seq Y\times Y$ to the relation $\ol{B}(R,S)$ on $B(X,Y)$ defined as follows:
\begin{enumerate}
\item\label{it:b-lift-func} $\ol{B}(R,S)_{\arty{\tau_{1}}{\tau_{2}}}\seq (Y_{\arty{\tau_{1}}{\tau_{2}}} + Y_{\tau_2}^{X_{\tau_1}})^2$ contains all $(u,v)$ such that either $u,v\in Y_{\arty{\tau_{1}}{\tau_{2}}}$ and $S_{\arty{\tau_{1}}{\tau_{2}}} (u,v)$, or
$u,v\in Y_{\tau_2}^{X_{\tau_1}}$ and for all $x,x'\in X_{\tau_1}$, if $R_{\tau_1}(x,x')$ then $S_{\tau_2}(u(x),v(x'))$.
\item $\ol{B}(R,S)_{\mu \alpha.\tau} \seq (Y_{\mu\alpha.\tau}+Y_{\tau[\mu\alpha.\tau/\alpha]})^2$ contains all such $(u,v)$ that either $u,v\in Y_{\mu\alpha.\tau}$ and $S_{\mu\alpha.\tau}(u,v)$, or $u,v\in Y_{\tau[\mu\alpha.\tau/\alpha]}$ and $S_{\tau[\mu\alpha.\tau/\alpha]}(u,v)$.
\item $\ol{B}(R,S)_{\tau_1\boxplus \tau_2}\seq (Y_{\tau_1\boxplus \tau_2} + Y_{\tau_1} + Y_{\tau_2}))^2$ contains all $(u,v)$ such that either $u,v\in Y_{\tau_1\boxplus \tau_2}$ and $S_{\tau_1\boxplus \tau_2}(u,v)$, or $u,v\in Y_{\tau_1}$ and $S_{\tau_1}(u,v)$, or $u,v\in Y_{\tau_2}$ and $S_{\tau_2}(u,v)$.
\item $\ol{B}(R,S)_{\tau_1\boxtimes \tau_2}\seq (Y_{\tau_1\boxtimes \tau_2} + Y_{\tau_1} \times Y_{\tau_2}))^2$ contains all $(u,v)$ such that either $u,v\in Y_{\tau_1\boxtimes \tau_2}$ and $S_{\tau_1\boxtimes \tau_2}(u,v)$, or $u=(u_1,u_2), v=(v_1,v_2)\in Y_{\tau_1} \times Y_{\tau_2}$ and $S_{\tau_i}(u_i,v_i)$ for $i=1,2$.
\end{enumerate}
Note how these clauses match the definition of the logical relation~$\logrel$ in \autoref{def:logrel}. In particular, clause \ref{it:b-lift-func} captures the requirement that for functions, related inputs should lead to related outputs.
\end{example}

Using relation liftings of bifunctors we can introduce bisimulations and logical relations for higher-order coalgebras. 
Both notions are parametric in a pair of coalgebras $c,\wt{c}\colon X\to B(X,X)$, with~$c$ thought of as the operational model of some higher-order GSOS law and $\wt{c}$ as some form of weak transition system associated to $c$ (see \autoref{ex:logrel} below for illustration).

\begin{definition}[Bisimulation, Logical Relation]\label{def:bisim-logical-relation}
Let $B\colon \C^\op\times \C\to \C$ with a relation lifting $\ol{B}$, and let $c,\wt{c} \colon X \to B(X,X)$ be coalgebras.
\begin{enumerate} 
\item\label{def:bisim-logical-relation-1} A \emph{bisimulation} for $(c,\wt{c})$ is a relation $R\monoto X\times X$ such that
\[R \le \iimg{(c\times \wt{c})}{\overline B(\Delta,R)}.\] 
\item\label{def:bisim-logical-relation-2} A \emph{logical relation} for $(c,\wt{c})$ is a relation $R\monoto X\times X$ such that 
\[R \le \iimg{(c\times \wt{c})}{\overline B(R,R)}.\]
\item\label{def:bisim-logical-relation-3} A \emph{step-indexed logical relation} for $(c,\wt{c})$ is an ordinal-indexed family of relations $(R^\alpha\monoto X\times X)_\alpha$ that forms a decreasing chain
(i.e.\ $R^\alpha\leq R^\beta$ for all $\beta<\alpha$) and satisfies
\[ R^{\alpha+1}\leq \iimg{(c\times \wt{c})}{\ol{B}(R^\alpha,R^\alpha)}\quad \text{for all $\alpha$}.\] 
\end{enumerate}
\end{definition}
 Informally, a (step-indexed) logical relation requires that, on term-labeled transitions, related input terms should lead to related output terms, while a bisimulation only considers identical inputs.

\begin{rem}\label{rem:logrel}
The above concepts are related as follows:
\begin{enumerate}
\item\label{rem:logrel-1} For $c=\wt{c}$ a bisimulation corresponds to the familiar notion of (Hermida-Jacobs) bisimulation~\cite{hj98} for the endofunctor $B(X,-)$ and its relation lifting $\ol{B}(\Delta,-)$. If $\wt{c}$ is some weakening of $c$, one obtains an abstract notion of \emph{applicative bisimulation}~\cite{UrbatTsampasEtAl23}.
\item\label{rem:logrel-2} Every reflexive logical relation is a bisimulation.
\item\label{rem:logrel-3} Every logical relation $R$ can be regarded as a step-indexed logical relation by putting $R^\alpha=R$ for all $\alpha$.
\item\label{rem:logrel-4} Conversely, every step-indexed logical relation $(R^\alpha)_\alpha$ induces a logical relation: since $\RelCat[X]{\C}$ is a small complete lattice, the decreasing chain $(R^\alpha)_\alpha$  eventually stabilizes, i.e.\ there exists a (least) ordinal $\nu$ such that $R^{\nu+1}=R^\nu$. An upper bound to $\nu$ is given by the cardinality of $\RelCat[X]{\C}$. Then $R^\nu$ is a logical relation because
\[ R^\nu = R^{\nu+1} \leq \iimg{(c\times \wt{c})}{\ol{B}(R^\nu,R^\nu)}.\]
\end{enumerate}
\end{rem}

\begin{example}\label{ex:logrel}
Let $\gamma\colon \Tr\to B(\Tr,\Tr)$ be the operational model of $\stscr$ given by \eqref{exa:gamma}. By postcomposing with the map $b\mapsto \{b\}$ we regard $\gamma$ as a nondeterministic transition system $\gamma\colon \Tr\to \Pow_\star B(\Tr,\Tr)$, and we denote by $\wt{\gamma}\colon \Tr\to \Pow_\star B(\Tr,\Tr)$ its corresponding \emph{weak transition system}, given for each $\tau\in \Ty$ and $t\in \Tr_\tau$ by
\[ \wt{\gamma}_\tau(t)= \{ t\} \cup \bigcup_{t\To s} \gamma_\tau(s). \]
The bifunctor $\Pow_\star\cdot B$ has a relation lifting $\arP_\star\cdot \ol{B}$ where $\ol{B}$ is the canonical lifting of $B$ (\autoref{ex:lift-beh}) and $\arP_\star$ is the left-to-right Egli-Milner lifting of $\Pow_\star$ (\autoref{ex:endofunctor-liftings}\ref{ex:endofunctor-liftings-typed-pow}). The relations $(\logrel^\alpha)_{\alpha}$ given by \autoref{def:logrel} for $\alpha\leq \omega$ and by $\logrel^\alpha=\logrel^\omega$ for $\alpha>\omega$ form a step-indexed logical relation for $(\gamma,\wt{\gamma})$ w.r.t.\ the lifting $\arP_\star\cdot \ol{B}$.
\end{example}

\subsection{Constructing Step-Indexed Logical Relations}\label{sec:step-indexed}

In the following we present a general construction of a step-indexed logical relation that applies to arbitrary higher-order coalgebras, in particular to operational models of higher-order GSOS laws. 
\begin{construction}[Step-Indexed Henceforth]
\label{def:stepindexed}
Let $B\colon \C^\op\times \C\to \C$ with a relation lifting $\ol{B}$, and let $c,\wt{c} \colon X \to B(X,X)$ be coalgebras.
For every $R\monoto X\times X$ we define the step-indexed logical relation $(\,\invp^{\ol{B},c,\wt{c},\alpha} R\monoto X\times X\,)_\alpha$
by transfinite induction:
\begin{align*}
\invp^{\ol{B},c,\wt{c},0} R	 				&=\, R,\\
\invp^{\ol{B},c,\wt{c},\alpha+1} R 	&=\, \invp^{\ol{B},c,\wt{c},\alpha} R \wedge \iimg{(c\times \wt{c})}{\ol{B}(\invp^{\ol{B},c,\wt{c},\alpha}R,\invp^{\ol{B},c,\wt{c},\alpha}R)},\\
\invp^{\ol{B},c,\wt{c},\alpha} R 		&=\, \bigwedge_{\beta<\alpha} \invp^{\ol{B},c,\wt{c},\beta} R \qquad \text{for limit ordinals $\alpha$}.
\end{align*}
We usually write $\square^\alpha R$ for $\square^{\ol{B},c,\wt{c},\alpha} R$, and we let $\nu$ denote the least ordinal such that $\square^{\nu+1} R=\square^\nu R$ (\autoref{rem:logrel}\ref{rem:logrel-4}).
\end{construction}
The construction of $\square^\alpha R$ thus takes, at every non-zero ordinal $\alpha$, the greatest relation satisfying the conditions of \autoref{def:bisim-logical-relation}\ref{def:bisim-logical-relation-3}.

\begin{example}\label{ex:henceforth}
Taking $B$, $\ol{B}$, $\gamma$, $\wt{\gamma}$ as in \autoref{ex:logrel} and $R=\top$ (the total relation on $\Tr$), we see that $(\square^{\ol{B},\gamma,\wt{\gamma},\alpha} \top)_\alpha$ coincides with the step-indexed logical relation $(\logrel^\alpha)_\alpha$ for \stscr given by \autoref{def:logrel}. Note that $(\square^\alpha \top)_\alpha$ stabilizes after $\omega$ steps (\autoref{rem:logrel-props}\ref{rem:logrel-props-2}).
\end{example}

We will now apply \autoref{def:stepindexed} to the operational model of a higher-order GSOS law. Let us fix the required setup:
\begin{assumptions}\label{asm-cong}
In the remainder we assume that we are given:
\begin{enumerate}
\item\label{asm-cong-1} an endofunctor $\Sigma\colon \C\to \C$ that preserves strong epimorphisms and admits free algebras $\Sigmas X$, with its canonical lifting $\ol{\Sigma}$;
\item\label{asm-cong-2} a bifunctor $B\colon \C^\op \times \C \to \C$ with a relation lifting $\ol{B}$;
\item\label{asm-cong-3} for each $X,Y\in \C$ a preorder $\preceq$ on $B(X,Y)$ such that each relation $\ol{B}(R,S)$ is up-closed w.r.t. $\preceq$ (\autoref{def:good-for-simulations});
\item\label{asm-cong-4} a $V$-pointed higher-order GSOS law $\rho$ of $\Sigma$ over $B$ that lifts to a $\Delta_V$-pointed higher-order GSOS law $\ol{\rho}$ of $\ol{\Sigma}$ over $\ol{B}$.  
\end{enumerate}
\end{assumptions}

\begin{remark}\label{rem:conditions}
\begin{enumerate}
\item We stress that the relation lifting $\ol{B}$ of $B$ need not be canonical, while we always choose the canonical lifting  of $\Sigma$.
\item The requirement that $\Sigma$ preserves strong epimorphism ensures that the free monad $\Sigmas$ lifts (\autoref{prop:free-monad-lift}) and is rather innocent: it holds, e.g., for all set functors, all left adjoints (including binding functors on presheaf categories figuring in \autoref{sec:applications}), for coproducts of such functors, and for products in categories where strong epimorphisms are stable under products. This covers a wide range of syntax functors for higher-order languages, which typically combine polynomial and binding constructions.   
\item \label{rem:rho-lift} Condition \ref{asm-cong-4} states that $\under{\ol{\rho}_{R,S}} = \rho_{X,Y}$ for all relations $R\monoto X\times X$ and $S\monoto Y\times Y$, where $\under{-}\colon \Rel(\C)\to\C$ is the forgetful functor. Note that $\under{\ol{\rho}_{R,S}}$ has the type~\eqref{eq:rhotcl} 
by \autoref{prop:free-monad-lift}. Note also that a lifting $\ol{\rho}$ is necessarily unique because $\under{-}$ is faithful. If $\ol{B}$ is the canonical lifting of $B$, a corresponding higher-order GSOS law always has a relation lifting~\cite[Constr.~D.5]{utgms23_arxiv}.
Generally, the existence of $\ol{\rho}$ can be understood as a monotonicity condition on the rules represented by $\rho$. For instance, for functors $B$ modelling nondeterministic behaviours and whose relation lifting involves the left-to-right Egli-Milner lifting $\arP$, it entails the absence of rules with negative premises~\cite{fs10}.
\end{enumerate}
\end{remark}

\begin{example}\label{ex:mutcl-asm}
For \stscr we instantiate the above data to
\begin{enumerate}
\item the signature functor $\Sigma\colon \Set^\Ty\to \Set^\Ty$ of \stscr;
\item the behaviour $\Pow_\star\cdot B$ and its lifting $\arP_\star \cdot \ol{B}$ as in \autoref{ex:logrel};
\item the order $\seq$ on $\Pow_\star B(X,Y)$ given by pointwise inclusion;
\item the ($0$-pointed) higher-order GSOS law $\rho'$ of $\Sigma$ over $\Pow_\star\cdot B$ whose component at $X,Y\in \Set^\Ty$ is given by
\begin{equation}\label{eq:rho0-to-rho}
  \begin{tikzcd}[column sep=50, row sep=20]
    \Sigma(X\times \Pow_\star B(X,Y)) \ar{r}{\Sigma \st_{X,B(X,Y)}} &
    \Sigma\Pow_\star(X\times B(X,Y)) \ar{dl}[description]{\delta_{X\times B(X,Y)}} \\
    \Pow_\star\Sigma(X\times B(X,Y)) \ar{r}{\Pow_\star\rho_{X,Y}} & \Pow_\star B(X,\Sigmas(X+Y))
  \end{tikzcd}
\end{equation}
Here $\rho$ is the higher-order GSOS law for \stscr given by \eqref{eq:rhotcl}, 
the natural transformation $\st$ is the (pointwise) canonical strength
\[ \st_{X,Y}\colon X\times \Pow_\star Y \to \Pow_\star(X\times Y),\quad (\st_{X,Y})_\tau(x,A) = \{ (x,a)\mid a\in A \},  \]
and the natural transformation $\delta \c \Sigma \Pow_\star \to \Pow_\star\Sigma$ is given by
\[
  \f(U_1,\ldots,U_n) \mapsto \{\, \f(u_1,\ldots,u_n) \mid u_1\in U_1,\ldots,u_n\in U_n\}.
\]
Hence, $\rho'$ is essentially the law $\rho$ regarded as a law for the non\-deterministic behaviour functor $\Pow_\star\cdot B$. Accordingly, the operational model of $\rho'$ is simply the composite
\[
  \mS \xto{\gamma} B(\mS,\mS)\monoto \Pow_\star
  B(\mS,\mS),
\]
of the operational model $\gamma$ of $\rho$ given by \eqref{exa:gamma} with the map $b\mapsto \{b\}$.

Let us verify that $\rho'$ has a relation lifting (Assumption~\ref{asm-cong}\ref{asm-cong-4}). Since $\ol{B}$ is the canonical lifting, the law $\rho$ lifts (\autoref{rem:conditions}.\ref{rem:rho-lift}), so we only need to show that $\st$ and $\delta$ lift. This means that for all $R\monoto X\times X$ and $S\monoto Y\times Y$ the maps
\[ \st_{X,Y}\colon R\times \arP_\star S \to \arP_\star (R\times S) \qand \delta_X\colon \ol{\Sigma}\,\arP_\star R\to \arP_\star \ol{\Sigma} R  \]
are $\Rel(\C)$-morphisms, which easily follows from the definitions.
\end{enumerate}
\end{example}

The key ingredient to our general congruence result for logical relations is a
higher-order version of \emph{lax
  bialgebras}~\cite{DBLP:conf/concur/BonchiPPR15}:
\removeThmBraces
\begin{definition}[\cite{gmstu23,UrbatTsampasEtAl23}]\label{D:lax-bialg}
 A \emph{lax $\rho$-bialgebra} $(X,a,c)$
  is given by an object $X\in \C$ and morphisms $a\c \Sigma X\to X$ and $c\c
  X\to B(X,X)$ such that the diagram below commutes laxly:
  \begin{equation*}\label{eq:lax-bialgebra}
    \begin{tikzcd}[column sep=1.95em]
      \Sigma X \ar{r}{a} \ar{d}[swap]{\Sigma \langle \id, c\rangle}
      &
      X \ar{r}{c}
      &[2em]
      B(X,X)
      \\
      \Sigma(X\times B(X,X))
      \ar{r}{\rho_{X,X}}
      &
      B(X,\Sigmas(X+X))
      \ar[phantom]{u}[description]{\dgeq{-90}}
      \ar{r}{B(\id,\Sigmas\nabla)}
      &
      B(X,\Sigmas X) \ar{u}[swap]{B(\id,\hat{a})}
    \end{tikzcd}
  \end{equation*}
If the diagram commutes strictly, then $(X,a,c)$ is a \emph{$\rho$-bialgebra}.
\end{definition}
\resetCurThmBraces

Informally, in a lax bialgebra the operational rules corresponding to the law $\rho$ are sound (every $c$-transition required by some rule exists), and in a bialgebra they are also complete (no other $c$-transitions exists). The operational model $(\mS,\ini,\gamma)$ of $\rho$ (\autoref{fig:gamma}) is a $\rho$-bialgebra, and if $\wt{\gamma}$ is some weakening of $\gamma$, laxness of $(\mS,\ini,\wt{\gamma})$ means that the rules given by $\rho$ are sound for weak transitions.

\begin{example}\label{ex:lax-bialgebra}
The weak operation model $\wt{\gamma}$ for \stscr (\autoref{ex:logrel}) forms a lax $\rho'$-algebra, with $\rho'$ given by \eqref{eq:rho0-to-rho}. Indeed, all the rules of \autoref{fig:skirules} remain sound when strong transitions are replaced by weak ones\,\footnote{This alludes to the rules being \emph{cool} in the sense of the
  cool congruence formats of \citet{DBLP:journals/tcs/Bloom95} and
  \citet{DBLP:journals/tcs/Glabbeek11}.}. For illustration, consider the two rules for application and their respective weak versions: 
\begin{align*}
	\inference{t\to t'}{t \app{\tau_{1}}{\tau_{2}} s\to t' \app{\tau_{1}}{\tau_{2}} s} 
\quad
	\inference{t\xto{s} t'}{t \app{\tau_{1}}{\tau_{2}} s\to t'} 
\quad
	\inference{t\To t'}{t \app{\tau_{1}}{\tau_{2}} s\To t' \app{\tau_{1}}{\tau_{2}} s}
\quad
	\inference{t \stackon[.5ex]{~\To~}{\scalebox{.75}{$\scriptsize s$}} t'}{t \app{\tau_{1}}{\tau_{2}} s\To t'}
\end{align*}
The third rule is sound because it emerges via repeated application of the first one. The fourth rule is sound as it follows from the second and third rule. Similarly for the other rules of \autoref{fig:skirules}.
\end{example}
The soundness condition modeled by lax bialgebras give rise to a natural criterion for the logical relation $\square^\nu R$ to be a congruence:

\begin{theorem}\label{thm:main}
Let $(X,a,c)$ be a $\rho$-bialgebra, and let $(X,a,\wt{c})$ be a lax
$\rho$-bialgebra. For every congruence $R\monoto X\times X$ on $(X,a)$ and every ordinal $\alpha$, the
relation $\square^{\alpha}R=\square^{\ol{B},c,\wt{c},\alpha}R$ is a congruence
on $(X,a)$.
\end{theorem}

\begin{proof}
We proceed by transfinite induction. The base case is immediate since $\square^0R = R$ is a congruence by assumption. The limit step follows because meets of congruences are congruences. For the successor step $\alpha\to\alpha+1$, suppose that $\square^\alpha R$ is a congruence. We are to prove that $\square^{\alpha+1}R$ is congruence, i.e.\ $\fimg{a}{\ol{\Sigma}(\square^{\alpha+1}R)} \leq \square^{\alpha+1} R$.
By definition of $\square^{\alpha+1}R$, this is equivalent to showing
\begin{align}
\fimg{a}{\ol{\Sigma}(\square^{\alpha+1}R)} &\leq \square^{\alpha} R, \label{eq:proof-main-1} \\
\fimg{a}{\ol{\Sigma}(\square^{\alpha+1}R)} &\leq \iimg{(c\times \wt{c})}{\ol{B}(\square^\alpha R, \square^\alpha R)}.\label{eq:proof-main-2}
\end{align}
The inequality \eqref{eq:proof-main-1} holds because $\square^{\alpha+1}R\leq \square^{\alpha}R$ and $\square^\alpha R$ is a congruence by induction. To prove \eqref{eq:proof-main-2}, we first observe that there exists a $\C$-morphism $\ol{\Sigma}(\square^{\alpha+1} R)\to \ol{B}(\square^\alpha R,\square^\alpha R)$ such that
\begin{equation}\label{eq:proof-main-diag}
  \begin{tikzcd}
    & \ar[bend right=2em]{dl}[swap]{c\cdot a \cdot \outl} \ar[bend left=2em]{dr}{\wt{c}\cdot a\cdot \outr} \ar[dashed]{d} \ol{\Sigma}(\square^{\alpha+1}R) & {~}\\
    B(X,X) & \ar[phantom]{ur}[description, pos=.35, xshift=-10]{\dleq{45}} \ar{l}[swap]{\outl}  \ol{B}(\square^\alpha R, \square^\alpha R) \ar{r}{\outr} & B(X,X)
  \end{tikzcd}  
\end{equation}
commutes laxly, where $\outl$ and $\outr$ are the projections of the respective relations (omitting subscripts).
This follows from the diagram in \autoref{fig:proof-main}; all its cells commute (laxly) as indicated. Here $p$ is the pairing of the $\Rel(\C)$-morphisms witnessing that $\square^{\alpha+1}R\leq \square^{\alpha}R$ and $\square^{\alpha+1}R\leq \iimg{(c\times \wt{c})}{\ol{B}(\square^\alpha R,\square^\alpha R)}$, and $@\colon \ol{\Sigma} R \to R$ is the $\ol{\Sigma}$-algebra structure witnessing that $R$ is a congruence on $(X,a)$. (We write $f_1\colon S\to T$ for the $\C$-morphism witnessing that $f\colon Y\to Z$ is $\Rel(\C)$-morphism from $S\monoto Y\times Y$ to $T\monoto Z\times Z$.)
    \begin{figure*}
      \begin{tikzcd}
        X
        \ar{ddd}[swap]{c}
        & \Sigma X 
        \ar[swap]{l}{a}
        \ar{d}[swap]{\Sigma\langle\id,c\rangle} & & \overline{\Sigma}(\invp^{\alpha+1}R)
        \ar[swap]{ll}{\outl}
        \arrow{d}{(\ol{\Sigma} p)_1}
        \ar{rr}{\outr}
        & & \Sigma X  \ar[swap]{d}{\Sigma\langle \id,\widetilde{c}\rangle}
        \ar{r}{a}
        \ar[phantom]{dddr}[description, pos=.4, xshift=10]{\dleq{45}}
        & X \ar{ddd}{\widetilde{c}} \\
        &
        \Sigma(X \times B(X,X))
        \ar{d}[swap]{\rho_{X,X}}
        & & \ar[swap]{ll}{\outl}
        \overline{\Sigma}(\invp^{\alpha}R \times
        \overline{B}(\invp^{\alpha}R, \invp^{\alpha}R))
        \arrow{d}{(\ol{\rho}_{\square^\alpha R, \square^\alpha R})_1}
        \ar{rr}{\outr}
        & & \Sigma(X \times B(X,X)) \ar{d}[swap]{\rho_{X,X}}
        & \\
        &
        B(X, \Sigma^{\star}(X + X))
        \ar{dl}{B(\id,\widehat{a} \comp \Sigmas \nabla)}
        & &
        \ar[swap]{ll}{\outl}
        \overline{B}(\invp^{\alpha}R,
        \overline{\Sigma}^{\star}(\invp^{\alpha}R + \invp^{\alpha}R))
        \ar{d}{(\ol{B}(\id,\widehat{@}\cdot \ol{\Sigma}^\star \nabla))_1}
        \ar{rr}{\outr}
        & &
        B(X, \Sigma^{\star}(X + X))
        \ar[swap]{dr}{B(\id,\widehat{a} \comp \Sigmas \nabla)}
        \\
        B(X,X)
        & & &
        \ar[swap]{lll}{\outl}
        \overline{B}(\invp^{\alpha}R, \invp^{\alpha}R)
        \ar{rrr}{\outr}
        & & &
        B(X,X)
      \end{tikzcd}
\caption{Diagram for the proof of \autoref{thm:main}}\label{fig:proof-main}
    \end{figure*}
Since $\ol{B}(\square^\alpha R,\square^\alpha R)$ is up-closed (Assumption~\ref{asm-cong}\ref{asm-cong-3}), it follows from \eqref{eq:proof-main-diag} that there exists a morphism $\ol{\Sigma}(\square^{\alpha+1}R)\to \ol{B}(\square^\alpha R,\square^\alpha R)$ such that the diagram below commutes strictly:
\[
  \begin{tikzcd}
    & \ar[bend right=2em]{dl}[swap]{c\cdot a \cdot \outl} \ar[bend left=2em]{dr}{\wt{c}\cdot a\cdot \outr} \ar[dashed]{d} \ol{\Sigma}(\square^{\alpha+1}R) & {~}\\
    B(X,X) & \ar{l}[swap]{\outl}  \ol{B}(\square^\alpha R, \square^\alpha R) \ar{r}{\outr} & B(X,X)
  \end{tikzcd} 
\]
This commutative diagram is equivalent to \eqref{eq:proof-main-2}.
\end{proof}

The main interest of~\autoref{thm:main} is the case where $R=\top$ and $(X,a,c)$ is the operational model of the higher-order GSOS law $\rho$. This leads us to the main result of our paper, a congruence result for our generic step-index logical relation $(\square^\alpha \top)_\alpha$. We stress that our results rely on the \autoref{assumptions} and \ref{asm-cong}.

\begin{corollary}[Congruence]\label{cor:main}
Let $(\mS,\ini,\gamma)$ be the operational model of $\rho$, and let $(\mS,\ini,\wt{\gamma})$ be a lax $\rho$-bialgebra. Then for every ordinal $\alpha$ the relation $\square^{\alpha}\top=\square^{\ol{B},\gamma,\wt{\gamma},\alpha}\top$ is a congruence on  $\mu\Sigma$.
\end{corollary}
Since every congruence on $\mS$ is reflexive (\autoref{prop:cong-refl}), we immediately conclude:

\begin{corollary}[Fundamental Property]\label{cor:fundamental-property}
In the setting of \autoref{cor:main}, the logical relation $\square^\nu \top$ is reflexive.
\end{corollary}

Finally, we deduce a simple criterion for the logical relation $\square^\nu \top$ to yield a sound proof method for a contextual preorder:

\begin{corollary}[Soundness]\label{cor:soundness}
Let $O\monoto \mS\times \mS$ be a preorder such that the contextual preorder $\cprd$ exists. In the setting of \autoref{cor:main}, if the logical relation $\square^\nu \top$ is $O$-adequate, it is contained in $\cprd$.
\end{corollary}
Recall from \autoref{prop:contextual-preorder} that $\cprd$ always exists under natural assumptions on the category $\C$ and the syntax functor $\Sigma$.

\begin{example}
\label{ex:dbtilde}
  Choose the data for \stscr as in \autoref{ex:mutcl-asm}.
\begin{enumerate}
\item By taking $O_\tau=\{ (t,s) \mid t{\Downarrow}\,\To\, s{\Downarrow}  \}$ as in \autoref{ex:observations}\ref{ex:observations-2}, we recover the results of \autoref{th:l-is-cong} and \autoref{cor:sound}: the relation $\square^\nu \top = \logrel^\omega$ is a congruence, and sound for the contextual preorder.
\item\label{ex:dbtilde:2} To capture the ground contextual preorder, we extend the weak transition system $\wt{\gamma}$ of \autoref{ex:logrel} as follows: put
\[ \dbtilde{\gamma}_{\arty{\tau_1}{\tau_2}}(t)= \wt{\gamma}_{\arty{\tau_1}{\tau_2}}(t) \cup \bigcup_{t\To s}\{ \lambda e\colon \tau_1.\,({s}\,{e}) \}  \]
at function types and $\dbtilde{\gamma}_\tau = \wt{\gamma}_\tau$ at all other types. This amounts to extending the weak transition relation $\xTo{l}$ to $\myarr{l}$ as required in \autoref{def:logreltwo}. Observe that $(\mS,\ini,\dbtilde{\gamma})$ is still a lax bialgebra, i.e.\ the rules remain sound w.r.t.\ the extended weak transition relation. The step-indexed logical relation $(\square^\alpha \top)_\alpha = (\square^{\ol{B},\gamma,\dbtilde{\gamma},\nu}\top)_\alpha$ coincides with $(\logreltwo^\alpha)_{\alpha}$ given in \autoref{def:logreltwo}. Taking $O_\booltype \{ (t,s) \mid t{\Downarrow}\,\To\, s{\Downarrow}  \}$ and $O_{\tau}=\mS_\tau\times \mS_\tau$ for $\tau\neq \booltype$ as in \autoref{ex:observations}\ref{ex:observations-2}, we recover the results of \autoref{th:l2-is-cong} and \autoref{cor:sound2}: the relation $\logreltwo^\omega$ is a congruence, and sound for the ground contextual preorder.
\end{enumerate}
\end{example}

The key insight to draw from the above results is that the lax bialgebra condition forms the language-specific core of compatibility and soundness of logical relations, while their boiler-plate part comes for free thanks to the categorical setup. Checking the lax bialgebra condition itself for a given language usually boils down to a straightforward analysis of its rules, as in \autoref{ex:lax-bialgebra}. 

\begin{remark}\label{rem:logrel-vs-howe}
We highlight some noteworthy connections with the recent work of Urbat et al.~\cite{UrbatTsampasEtAl23}, which lifts another operational method -- namely Howe's method~\cite{DBLP:conf/lics/Howe89, DBLP:journals/iandc/Howe96} -- to the level of higher-order abstract GSOS. Briefly, Howe's method is used to show that applicative (bi)similarity is a congruence, while logical relations are a form of relation between programs that is essentially designed with the congruence property in mind. There are certain similarities between the two methods, which our present work makes explicit:
\begin{enumerate}
\item Logical relations and Howe's method are regarded as independent techniques in the literature, each with its own use cases.
Our present results demonstrate that logical relations for recursive types can be modeled in the same categorical framework as applicative simulations for untyped languages. 
Even more remarkably, the respective congruence results (\autoref{cor:main} and \cite[Thm.~VIII.6]{UrbatTsampasEtAl23}) boil down to the same abstract (lax-bialgebra) condition. This exposes a formal and explicit connection between the two most widely used higher-order operational techniques.

\item Our present technical setup is slightly simpler compared to~\cite{UrbatTsampasEtAl23}:
\begin{enumerate}
\item Extensivity is not needed for the congruence result (just for constructing the contextual preorder).
\item Weakenings are arbitrary coalgebras, while in~\cite{UrbatTsampasEtAl23} they are restricted. For instance, the coalgebra $\dbtilde{\gamma}$ of \autoref{ex:dbtilde}\ref{ex:dbtilde:2} is not a weakening in the sense of~\cite{UrbatTsampasEtAl23}.
\item The congruence proof proceeds directly via structural induction and does not involve Howe's closure.
\end{enumerate}
These simplifications bear a relevant insight on their own: they reflect, on a categorical level, the common wisdom that congruence proofs for logical relations are (and in fact should be) less complex than for applicative simulations.
\end{enumerate}
\end{remark}

\section{Binding and Nondeterminism}
\label{sec:applications}
We conclude with a case study of our general framework, namely the functional
language \fpc, which is a typed
$\lambda$-calculus with recursive types. We consider a nondeterministic
variant of the standard \fpc~\cite{Gunter92,FioreP94} to emphasize that higher-order abstract GSOS can handle effectful settings. The overall treatment is similar to that of untyped~\cite{gmstu23,UrbatTsampasEtAl23} and simply typed~\cite{gmstu24} $\lambda$-calculi in earlier work; hence we focus on the core ideas and omit technical details.

The types of \fpc are the same as those of \stscr, defined earlier
\eqref{eq:type-grammar}. The language \fpc lacks combinators but instead
includes variables, $\lambda$-abstractions, sums, products conditionals and a
nondeterministic choice operator. We skip the term formation rules of \fpc, as
they are standard; the (call-by-name, open evaluation) operational semantics is
presented in \Cref{fig:fpcrules}.

\begin{figure*}[t]
  \centering
  \columnwidth=\linewidth
  \begin{gather*}
  \inference{t\to t'}{t \app{\tau_{1}}{\tau_{2}} s\to t'
    \app{\tau_{1}}{\tau_{2}} s}
  \qquad
  \inference{}{(\mathsf{lam}\,x\c\tau.\,t) \app{}{} s \to t[x/s]} \qquad
  \inference{t \to t'}{\mathsf{unfold}(t) \to \mathsf{unfold}(t')} \qquad
  \inference{}{\mathsf{unfold}(\mathsf{fold}(t)) \to t}
  \qquad
  \\[1ex]
  \inference{t \to t'}{\mathsf{fst}(t) \to \mathsf{fst}(t')}
  \qquad
  \inference{t \to t'}{\mathsf{snd}(t) \to \mathsf{snd}(t')}
  \qquad
  \inference{}{\mathsf{fst}(\mathsf{pair}(t,s)) \to t}
  \qquad
  \inference{}{\mathsf{snd}(\mathsf{pair}(t,s)) \to s}
  \\[1ex]
  \inference{}{t \oplus s \to t}
  \qquad
  \inference{}{t \oplus s \to s}
  \qquad
  \inference{t\to t'}{\mathsf{case}(t,s,r)\to
    \mathsf{case}(t',s,r)} \qquad
  \inference{}{\mathsf{case}(\inl(t),s,r)\to
    s\app{}{}t} \qquad
  \inference{}{\mathsf{case}(\inr(t),s,r)\to
    r\app{}{}t}
\end{gather*}
\caption{Call-by-name operational semantics of \fpc. Metavariables $r,s,t,t'$ range over
  possibly open terms.}
  \label{fig:fpcrules}
\end{figure*}

To implement \fpc in higher-order abstract GSOS, we build on the
categorical approach to abstract syntax and variable binding using
presheaves~\cite{DBLP:journals/mscs/Fiore22,DBLP:conf/lics/FiorePT99}. Let
$\fset$ be the category of finite cardinals and functions, and regard
the set $\Ty$ of types as a discrete category. An object $\Gamma\colon
n\to \Ty$ of the slice category $\fset/\Ty$ is a \emph{typed variable context} associating to each variable $x\in n$ a type; we put $\under{\Gamma}=n$. The fundamental operation of \emph{context extension} $(-
+ \check\tau) \c \fset/{\Tyl} \to \fset/{\Tyl}$ extends a context with
a new variable $x \c \tau$.

The base category for modeling $\fpc$, and typed languages in general, is the presheaf category $\C=(\Set^{\fset/{\Ty}})^{\Ty}$. Informally, a presheaf $X\in (\Set^{\fset/{\Ty}})^{\Ty}$ associates to each type $\tau$ and  context $\Gamma$ a set $X_\tau(\Gamma)$ of terms of type $\tau$ in context $\Gamma$. For example,
the presheaf $V$ of \emph{variables} is given by $V_{\tau}(\Gamma)=\{x\in |\Gamma| \mid \Gamma(x) =\tau\}$, and \fpc-terms form a presheaf $\Lambda$ given by $\Lambda _\tau(\Gamma) = \{ t \mid \text{$t$ is an \fpc-term and } \Gamma \vdash t\colon \tau \}$, with terms taken modulo $\alpha$-equivalence.
The presheaf $\Lambda$ carries the initial algebra for the endofunctor $\Sigma$
on $(\Set^{\fset/{\Ty}})^{\Ty}$ corresponding to the \emph{binding
  signature}~\cite{DBLP:conf/lics/FiorePT99} of \fpc:
\begin{equation*}
  \begin{aligned}
    & \Sigma_{\tau} X =
      V_{\tau} + X_{\tau} \times X_{\tau}
      + \Sigma^{1}_{\tau}X + \Sigma^{2}_{\tau}X + \Sigma^{3}_{\tau}X + \Sigma^{4}_{\tau}X, \\
    & \Sigma^{1}_{\tau} X =
      \coprod_{\tau' \in \Tyl}(X_{\arty{\tau'}{\tau}} \times X_{\tau'})
      + X_{\tau' \boxtimes \tau} + X_{\tau \boxtimes \tau'},\\
    & \Sigma^{2}_{\tau_{1} \boxplus \tau_{2}} X= X_{\tau_{1}} + X_{\tau_{2}},
      \qquad
      \Sigma^{2}_{\tau_{1} \boxtimes \tau_{2}} X = X_{\tau_{1}} \times X_{\tau_{2}},
    \\
    & \Sigma^{2}_{\mu\alpha.\,\tau} X= X_{\tau[\mu\alpha.\,\tau/\alpha]},
      \qquad
   \Sigma^{2}_{\arty{\tau_{1}}{\tau_{2}}}X= X_{\tau_{2}} \comp (- + \check\tau_{1}),
    \\
    & \Sigma^{3}_{\tau} X = \coprod_{\tau_{1} \in \Tyl}
      \coprod_{\tau_{2} \in \Tyl}X_{\tau_{1} \boxplus \tau_{2}} \times X_{\arty{\tau_{1}}{\tau}}
      \times X_{\arty{\tau_{2}}{\tau}},
    \\
    & \Sigma^{4}_{\tau} X =
      \coprod_{\sigma \c \tau = \sigma[\mu \alpha.\sigma/\alpha]}X_{\mu\alpha.\sigma}.
  \end{aligned}
\end{equation*}
In the definition of $\Sigma^{4}$, $\sigma$ ranges over types with one free
variable.
The bifunctor $B$ for \fpc (compare with \eqref{eq:beh}) is given by
\begin{equation*}
  \begin{aligned}
    B(X,Y)
    &\;= \llangle X,Y \rrangle \times \Pow_{\star}(Y + D(X,Y)),\\
    D_{\arty{\tau_{1}}{\tau_{2}}}(X,Y)
    &\;= Y_{\tau_{2}}^{X_{\tau_{1}}},
    &\hspace{-5.5em}
      D_{\tau_1\boxplus\tau_2}(X,Y)
    &\;= Y_{\tau_1}+Y_{\tau_2}, \\
        D_{\mu\alpha.\,\tau}(X,Y)
                    &\,= Y_{\tau[\mu\alpha.\tau/\alpha]},&\hspace{-6em}
    D_{\tau_1\boxtimes\tau_2}(X,Y) &\,= Y_{\tau_1}\times Y_{\tau_2}.
  \end{aligned}
\end{equation*}
Here, $\Pow_{\star}$ is the pointwise
power set functor $X\mapsto \Pow\cdot X$, we write $Y_{\tau_{2}}^{X_{\tau_{1}}}$
for the exponential in $\Set^{\fset/{\Ty}}$,
and $\llangle \argument,\argument \rrangle$ is given by
\[
  \llangle X,Y \rrangle_{\tau}(\Gamma) = \Set^{\fset/{\Tyl}}\Bigl(\prod_{x \in
    |\Gamma|}X_{\Gamma(x)}, Y_{\tau}\Bigr).
\]
The bifunctor $\llangle \argument,\argument \rrangle$ models
\emph{substitution}. For example, there is a morphism $\gamma_{0}
\c \Lambda \to \llangle \Lambda,\Lambda \rrangle$ mapping terms of \textbf{FPC}
to their substitution structure: given $t\in \Lambda_\tau(\Gamma)$, the natural transformation $(\gamma_0)_\tau(t)\colon \prod_{x \in
    |\Gamma|}\Lambda_{\Gamma(x)}\to \Lambda_{\tau}$ is given at component $\Delta\in \fset/\Ty$ by
\[ \vec{u}\in \prod_{x \in
    |\Gamma|}\Lambda_{\Gamma(x)}(\Delta) \quad\mapsto \quad t[\vec{u}]\in \Lambda_\tau(\Delta),\] i.e.\ the simultaneous substitution of $\vec{u}$ for the variables of $t$. 

There is a $V$-pointed higher-order order GSOS law $\rho$ of $\Sigma$ over $B$ modelling
\textbf{FPC} \ifarx{(\Cref{sec:fpc-extended})}{\cite[Appendix C]{gmtu24_arxiv}} which, as in the case of \stscr, encodes the operational
rules of \fpc into maps, taking into account how the individual constructors
affect the substitution structure of terms,
The canonical operational model of $\rho$ is the coalgebra
\begin{equation*}
  \langle \gamma_{0},\gamma \rangle \c \Lambda \to \llangle \Lambda,\Lambda \rrangle \times
  \Pow_{\star}(Y + D(\Lambda,\Lambda))
\end{equation*}
where $\gamma_0$ is as described above, and $\gamma$ is the transition system
that models $\beta$-reduction according to the operational semantics and
identifies values similarly to \stscr in \Cref{sec:mu-ctl}. For instance,
\[ \gamma_{\tau_{1} \boxtimes \tau_{2}}(\Gamma)(\mathsf{pair}(t,s)) = \{(t,s)\}
  \text{, \,else~\,}
  \gamma_{\tau_{1} \boxtimes \tau_{2}}(\Gamma)(t) = \{ t' \mid  t\to t' \}.\]
Here, $\gamma$ identifies $\mathsf{pair}(t,s)$ as a value with
projections $t$ and $s$.
Replacing $\to$ by its reflexive transitive hull $\To$ we obtain the \emph{weak operational model} $\langle \gamma_0, \widetilde{\gamma}\rangle \c \Lambda \to
  \llangle \Lambda,\Lambda \rrangle \times \Pow_{\star}(Y + D(\Lambda,\Lambda))$.
Similar to the \stscr case, we choose the relation lifting $\arP_\star \cdot \ol{B}$ of $\Pow_\star\cdot B$ where $\arP_\star$ is the left-to-right Egli-Milner lifting and $\ol{B}$ is the canonical lifting. The generic logical relation $(\logrel^\alpha)_\alpha=(\square^\alpha)_\alpha$ w.r.t.\ $(\gamma,\wt{\gamma})$ is then given by $\logrel^0=\top$, $L^\alpha=\bigcap_{\beta<\alpha} L^\beta$ at limit ordinals $\alpha$, and 
\begin{align*}
  \kern.5em\logrel^{\alpha+1}_{\tau} =&\; \logrel^{\alpha}_{\tau}
     \cap \mathcal{S}_{\tau}(\logrel^{\alpha},\logrel^{\alpha})
      \cap \mathcal{E}_{\tau}(\logrel^{\alpha})
      \cap \mathcal{V}_{\tau}(\logrel^{\alpha},\logrel^{\alpha})
\end{align*}
where $\mathcal{S}$, $\mathcal{E}$,
$\mathcal{V}$ are the maps on $\Rel_{\Lambda}$ given by
  \begin{align*}
    & \mathcal{S}_{\tau}(\Gamma)(Q,R)
      = \{(t,s) \mid \text{for all $\Delta$ and $Q_{\Gamma(x)}(\Delta)(u_x,v_x)$ ($x\in \under{\Gamma}$)},  \\
    & \phantom{\mathcal{S}_{\tau}(\Gamma)(Q,R)= \{}
      \text{one has $R_{\tau}(\Delta)(t[\vec{u}],s[\vec{v}])$} \}, \\
    & \mathcal{E}_{\tau}(\Gamma)(R) = \{(t,s) \mid \text{if $t \to t'$ then $\exists s'.\,s \To s'
      \land R_{\tau}(\Gamma)(t',s')$}\}, \\
    & \mathcal{V}_{\arty{\tau_{1}}{\tau_{2}}}(\Gamma)(Q,R) =
      \{(t,s) \mid
      \text{for all $Q_{\tau_{1}}(\Gamma)(e,e'),$} \\
    & \quad \text{if $t=\lambda x.t'$ then $\exists s'.\,s \To \lambda x.s'\,\wedge\,R_{\tau_2}(\Gamma)(t'[e/x], s'[e'/x])$} \}, \\
    & \text{and similarly for the other types.}
  \end{align*}

\begin{remark}
 Note that substitution closure is directly built into the definition of $\L$. In the literature, logical relations are commonly constructed in two steps, namely by first defining them on closed terms, and then extending
  to open terms via (closed) substitutions.
\end{remark}
By instantiating our results of \autoref{sec:logrel} to \fpc, we obtain:
\begin{corollary}
The logical relation $\logrel$ for $\fpc$ is a congruence, and sound for the contextual preorder w.r.t.\ the termination predicate.
\end{corollary}
Again, this amounts to observing that the operational rules of \fpc remain sound for weak transitions, which is easily verified. 
\vspace{-0.2cm}
\section{Conclusion and Future Work}
We have developed a language-independent theory of step-indexed
logical relations and contextual equivalence based on higher-order abstract
GSOS. We have shown that logical relations arise naturally via a
generic construction and that soundness for contextual equivalence is contingent
of a simple condition on the operational semantics. We expect that
our theory will lead to a higher level of automation of proofs via logical
relations and that its scaling nature contributes towards efficient reasoning on
real-world languages.

Let us identify several directions for future work that will
further enrich our theory.
Currently, we have investigated logical relations in higher-order abstract
GSOS for call-by-name languages, and define them in terms of their small-step
($\to$) and weak $(\To)$ semantics. In the future, we will look to explore
call-by-value logical relations, and make use of \emph{big-step} semantics
$(\Downarrow)$. For the latter, the challenge would be to find a suitable analogue of our lax-bialgebra condition.

Another direction is to leverage our abstract coalgebraic theory to effectful, probabilistic and
differential settings, and also to generalize to \emph{Kripke} logical
relations~\cite{DBLP:journals/iandc/OHearnR95,DBLP:conf/popl/HurDNV12,DBLP:journals/jfp/DreyerNB12,DBLP:conf/popl/HurD11}.
Another compelling prospect is to apply our theory to
\emph{cross-language logical
  relations}~\cite{DBLP:conf/popl/DevriesePP16,10.1145/1596550.1596567,DBLP:conf/popl/HurD11,
  DBLP:conf/esop/PercontiA14,10.1145/3434302}, specifically
to reason about (secure) compilers. Such
an effort would likely involve the development of notions of \emph{morphisms of
  higher-order GSOS laws}, analogously to the first-order
case~\cite{DBLP:journals/entcs/Watanabe02}, a notion that has been used as
criterion of secure compilation in the
past~\cite{DBLP:conf/cmcs/0001NDP20,DBLP:conf/aplas/AbateBT21}.

In the present paper, we have focused on logical relations that are
sound for contextual equivalence; investigating completeness at similar generality is a sensible (and likely challenging) next step.

Finally, ever since their conception, logical relations have been used
for
\emph{parametricity}~\cite{DBLP:conf/ifip/Reynolds83,DBLP:journals/corr/AnandM17,DBLP:journals/jfp/BernardyJP12,DBLP:conf/csl/KellerL12,DBLP:journals/pacmpl/AhmedJSW17}.
We will aim to leverage our theory to study and reason about
parametricity at the level of generality of higher-order abstract
GSOS.\stnote{Mechanization?}

\begin{acks}
The first author would like to thank Paul Blain Levy for helpful discussions on the 
reasoning about program equivalence.
\end{acks}

\clearpage

\bibliography{mainBiblio}

\appendix
\onecolumn

\takeout{
\section{Optional Content non-indexed logical relations}

\subsection*{From \autoref{sec:relations}}
\begin{definition}
	Let $B \colon \C^\opp \times \C \to \C$ be a functor equipped with a relation lifting $\overline B \colon  {\RelCat \C}^\opp \times {\RelCat \C} \to \RelCat \C$, $c \colon Y \to B(X,Y)$ be a coalgebra, and $S \pred{} X \times X$ be a relation on $X$. We say that a relation $Q \pred{} Y \times Y$ is a $S$-\emph{relative} ($\overline B$-)bisimulation (for $c$) if 
	\[
	Q \le \iimg{c}{\overline B(S,Q)} \quad \text{or equivalently} \quad \fimg{c}{Q} \le \overline B(S,Q).
	\]
\end{definition}
We then obtain the usual notion of ``logical relation'' as a relation that is preserved by the coalgebra steps in terms of self-relative bisimulations:

Following previous work~\cite{FoSSaCS24}, we introduce a well-behavedness condition
for relation liftings, needed for computing their fixpoints, and abstracting structural 
induction over types. Let us briefly recall some terminology.
A metric space $(X,\, d\colon X\times X\to \mathbb{R})$ is \emph{$1$-bounded} if $d(x,y)\leq 1$ for all $x,y$, an \emph{ultrametric space} if $d(x,y)\leq\max\{d(x,z),d(z,y)\}$ for all $x,y,z$, and \emph{complete} if every Cauchy sequence converges. A map $f\colon (X,d)\to (X',d')$ between metric spaces
is \emph{nonexpansive} if $d'(f(x),f(y))\leq d(x,y)$ for all $x,y$, and \emph{contractive}
if there exists $c\in [0,1)$, called a \emph{contraction factor}, such that $d'(f(x),f(y))\leq c\cdot d(x,y)$
for all $x,y$. 
By Banach's fixed point theorem, every contractive endomap $f\colon X\to X$ on a complete metric space has a unique fixed point.

\begin{defn}\label{ass:contr}
The category $\C$ is \emph{predicate-contractive} if
\begin{enumerate}
  \item every $\Pred[X]{\C}$ carries the structure of a complete $1$-bounded ultrametric space;
  \item for every $f\c X\to Y$ in $\C$, the map $\iimg{f}{\argument}\c\Pred[Y]{\C}\to\Pred[X]{\C}$
  is non-expansive;
  \item for any two co-well-ordered families $(P^i\pred{} X)_{i\in I}$ and $(Q^i\pred{} X)_{i\in I}$ of predicates,
  \begin{displaymath}\textstyle
    d\bigl(\bigand_{i\in I} P^i,\bigand_{i\in I} Q^i\bigr)\leq \sup_{i\in I} d(P^i,Q^i).
  \end{displaymath} 
  Here $(P^i\pred{} X)_{i\in I}$ is \emph{co-well-ordered} if each nonempty subfamily has a greatest element.
\end{enumerate}
\end{defn}
\begin{example}
  The category $\C=\Set^\Ty$ is predicate-contractive when equipped with the
  ultrametric on $\Pred[X]{\C}$ given by $d(P,Q)=2^{-n}$ for
  $P,Q\monoto X$, where $n=\inf\{\sharp\tau\mid P_\tau\neq Q_\tau\}$
  and $\sharp\tau$ is the size of~$\tau$, defined by $\sharp\booltype=1$
  and $\sharp(\arty{\tau_1}{\tau_2})=\sharp\tau_1+\sharp\tau_2$. By
  convention, $\inf\,\emptyset=\infty$ and $2^{-\infty}=0$. To see
  predicate-contractivity, first note that a function
  $\mathcal{F}\c\Pred[Y]{\C}\to\Pred[X]{\C}$ is non-expansive iff
\begin{flalign*}
\quad\inf\{\sharp\tau \mid (\mathcal{F} P)_\tau \neq (\mathcal{F} Q)_\tau\}
\geq \inf\{\sharp\tau \mid P_\tau \neq Q_\tau\} && (P, Q\monoto Y)
\end{flalign*}
and contractive (necessarily with factor at most $1/2$) iff that inequality holds strictly.

This immediately implies clause~(2) of~\Cref{ass:contr}: inverse images in $\Set^\Ty$ are computed 
pointwise, and $\iimg{f_\tau}{P_\tau} \neq\iimg{f_\tau}{Q_\tau}$
implies $P_\tau\neq Q_\tau$ for $f\colon X\to Y$ and $P,Q\monoto Y$. Similarly, since intersections are computed pointwise, clause~(3) amounts to
\begin{align*}
\inf\Bigl\{\sharp\tau\mid \bigcap_{i\in I}P^i_\tau\neq \bigcap_{i\in I}Q^i_\tau\Bigr\}
\geq&\;\inf\{\sharp\tau\mid \exists i\in I: P^i_\tau\neq Q_\tau^i\},
\end{align*}
which is clearly true, for if $\bigcap_{i\in I}P^i_\tau\neq \bigcap_{i\in I}Q^i_\tau$ then 
$P^i_\tau\neq Q^i_\tau$ for some $i\in I$.
\end{example}
Since relations are special kind of predicates, we apply predicate-contractivity 
to relations in the obvious way.
For the rest of the section, we fix a mixed variance functor $B \colon \C^\opp \times \C \to \C$,
and a corresponding lifting $\ol B$, which satisfies the following 
\begin{assumption}\label{ass:lics}
\begin{enumerate}
 \item $\ol B$ is weakly contractive;
 \item For any $R\pred{} X\times X$ and any $S,S'\pred{} Y\times Y$, $\ol{B}(R,S)\comp\ol{B}(\Delta,S')\leq \ol{B}(R,S\comp S')$.
\end{enumerate}
\end{assumption}
\begin{example}
Consider the type grammar~\eqref{eq:type-grammar}. Given a type $\tau\in\Ty$, 
let $FL(\tau)$ be the Fischer-Ladner closure of $\tau$, i.e.\ the smallest set
of (closed) types, which is closed under direct subformulas of $\tau_1\boxplus\tau_2$,
$\tau_1\boxtimes\tau_2$, $\arty{\tau_{1}}{\tau_2}$, and under the following unfolding 
operation 
\begin{align*}
\mu\alpha.\,\tau'\leadsto\tau[\mu\alpha.\,\tau'/\alpha].
\end{align*}
It is known that $FL(\tau)$ is finite. 
We proceed to make it into a coalgebra in $\Set^\Ty$, by defining the transitions:
\begin{itemize}
    \item $\tau_1\boxplus\tau_2\xto{\boxplus_1}\tau_1$, $\tau_1\boxplus\tau_2\xto{\boxplus_2}\tau_2$;
    \item $\tau_1\boxtimes\tau_2\xto{\boxtimes_1}\tau_1$, $\tau_1\boxtimes\tau_2\xto{\boxtimes_2}\tau_2$;
    \item $\arty{\tau_{1}}{\tau_2}\xto{\arty{}{}_1}\tau_1$, $\arty{\tau_{1}}{\tau_2}\xto{\arty{}{}_2}\tau_2$;
    \item $\tau'\xto{\ell}\tau'''$ if $\tau'\leadsto^\star\tau''$, and $\tau''\xto{\ell}\tau'''$,
    is one of the transitions, defined above.
\end{itemize} 
For every $\tau\in\Ty$, let $\sharp\tau$ be the size of the bisimilarity 
quotient of the subcoalgebra of $FL(\tau)$ generated by $\tau$. This makes $\C=\Set^\Ty$ 
into a predicate-contractive category equipped with the ultrametric on $\Pred[X]{\C}$ 
given by $d(P,Q)=2^{-n}$ for $P,Q\monoto X$, where $n=\inf\{\sharp\tau\mid P_\tau\neq Q_\tau\}$.

By definition, 
every $\mu\alpha.\,\tau$ is bisimilar to its unfoldings, and hence $\sharp\mu\alpha.\,\tau=\sharp\tau[\mu\alpha.\,\tau/\alpha]$.
It is also clear that $\sharp\tau\leq\sharp\tau'$ for every direct subtype $\tau'$
of $\tau$. 
\end{example}

\paragraph{Henceforth modality} We can define a henceforth modality just like in the predicate case: given a coalgebra $c \colon Y \to B(X,Y)$, where $B \colon \C^\opp \times \C \to \C$ has a relation lifting $\overline B$, and a relation $S \pred{} X \times X$, the $S$-relative henceforth modality sends $R \pred{}Y \times Y$ to $\invp^{\ol{B},c} (S, R)
\pred{} Y \times Y$, which is the supremum of all $S$-relative bisimulations contained in~$R$:
\begin{equation}\label{eq:henceRel}
	\invp^{\ol{B},c} (S, R) = \bigvee\{ Q \leq R \mid Q
	\text{ is an $S$-relative $\ol{B}$-bisimulation for $c$} \}.
\end{equation}
We will omit the superscripts $\ol B,c$ if they are irrelevant, or clear from the context.

Again by the Knaster-Tarski theorem we have that $\invp (S,R)$ is itself an $S$-relative bisimulation, because the correspondence $G\mapsto R\land\iimg{c}{\ol{B}(S,G)}$
is a monotone endomap on the complete lattice of subobjects of $Y \times Y$, hence the corresponding 
greatest post-fixpoint exists and equals~\eqref{eq:henceRel}. Symbolically:
\begin{align*}
\invp(S, R) = \nu G.\ R \land \iimg{c}{\overline{B}(S,G)}.
\end{align*}
We thus obtain again an analogous result to \Cref{prop:invarprop}.
\begin{proposition}
\label{prop:invarpropRel}
Given $S\pred{} X \times X$ and $R\pred{} Y \times Y$,
\begin{enumerate}
	\item \label{strongerthanpRel} $\invp (S, R) \leq R$,
	\item \label{squareisinvRel} $\invp (S, R)$ is an $S$-relative $\ol{B}$-bisimulation,
	\item $\invp (S, R)$ is antitone in $S$ and monotone in $R$.
\end{enumerate}
\end{proposition}
Because $\RelCat[X]{\C} = \Pred[X \times X]{\C}$, \Cref{ass:contr} implies its
analogous version for relations. 

In the unary case, typically, $\overline{B}(S,\top)=\top$ for all $S$, which entails
that $\logp\top=\top$. 
Indeed, since $\top\leq\top\land \iimg{c}{\overline{B}(\logp\top,\top)}$, $\top$ 
is a $\logp\top$-relative invariant, and therefore it is smaller than the greatest 
such, i.e.\ $\logp\top$. Thus, in the unary case $\logp P$ is typically only of 
interest if $P$ is non-trivial, such as the strong normalization predicate. 
Contrastingly, in the binary case, $\logp\top$ is typically not $\top$, and in fact 
an extremely important relation, which, analogously to the case of first-order
bisimilarity, can arguably be regarded as the finest notion of program equivalence.

\begin{lemma}\label{lem:RR-refl}
Let $R\pred{}\mS\times\mS$ be a $\ol\Sigma$-congruence.
Then 
\begin{displaymath}
  \Delta\leq\iimg{\gamma}{\ol B(R,R)}.
\end{displaymath}
\end{lemma}
\begin{proof}
By assumption, $R$ is a $\ol\Sigma$-algebra in $\RelCat \C$. Note that $\Delta$ 
is the initial $\ol\Sigma$-algebra in $\RelCat\C$, and therefore, by instantiating 
a previous general result~\cite[Lemma 4.13]{gmstu23} in~$\RelCat \C$, we obtain 
\begin{equation*}
\qquad\begin{tikzcd}[column sep=4em, row sep=normal]
	\Delta 
	  \dar[tail] 
	  \rar["R^\clubsuit"] 
	& 
	  \ol B(R,R) 
	  \dar[tail]
\\
	\mS \times \mS
    \rar["\gamma \times \gamma"] 
	& 
	  B(\mS,\mS) \times B(\mS,\mS)
\end{tikzcd}
\end{equation*}
where $R^\clubsuit$ is the unique such morphism $R\to\ol B(R,R)$ that the
diagram commutes. Commutativity of this diagram is equivalent to the goal.
\end{proof}

\begin{lemma}\label{lem:popl}
Let $R\pred{}\mS\times\mS$ be a transitive $\ol\Sigma$-congruence.
Then 
\begin{displaymath}
  \iimg{\gamma}{\ol B(\Delta,R)}\leq\iimg{\gamma}{\ol B(R,R)}.
\end{displaymath}
\end{lemma}
\begin{proof}
We use the following instance of \Cref{ass:lics}~(2):
\begin{align*}
\ol{B}(R,R)\comp\ol{B}(\Delta,R)\leq\ol{B}(R,R\comp R).
\end{align*}
Now,
\begin{flalign*}
\quad\iimg{\gamma}{\ol B(\Delta,R)}
\leq\;&\iimg{\gamma}{\ol{B}(R,R)}\comp\iimg{\gamma}{\ol{B}(\Delta,R)}&\by{\Cref{lem:RR-refl}}\\
\leq\;&\iimg{\gamma}{\ol{B}(R,R)\comp\ol{B}(\Delta,R)}\\
\leq\;&\iimg{\gamma}{\ol{B}(R,R\comp R)}&\!\!\by{\Cref{ass:lics}~(2)}\\
\leq\;&\iimg{\gamma}{\ol{B}(R,R)}, &\by{transitivity}
\end{flalign*}
as desired.
\end{proof}

\subsection*{From \autoref{sec:step-indexed}}
\begin{definition}[Indexed Henceforth]\sgnote{Step-indexed?}
\label{def:stepindexed}
Let $R \pred{} X \times X$ be a relation on the state space of a coalgebra $c \c X \to
B(X,X)$ on a bifunctor $B\c \C^\opp\times \C\to \C$ with a contractive predicate 
lifting $\ol{B}$, define~$\invp^{\bar B,c,\alpha} R$ by transfinite induction as follows:
\begin{itemize}
  \item for a successor ordinal $\alpha+1$, $\invp^{\bar B,c,\alpha+1} R$ is the unique 
  solution of the equation
\begin{align*}
\invp^{\bar B,c,\alpha+1} R = R \land\invp^{\bar B,c,\alpha} R\land\iimg{c}{\overline{B}(\invp^{\bar B,c,\alpha+1} R,\invp^{\bar B,c,\alpha} R)};
\end{align*}
\ST{
  \begin{align*}
    \invp^{\bar B,c,\alpha+1} R = R \land\invp^{\bar B,c,\alpha} R\land\iimg{c}{\overline{B}(\invp^{\bar B,c,\alpha} R,\invp^{\bar B,c,\alpha} R)};
  \end{align*}
}
(the result is well-defined thanks to our assumption of weak contractivity of $\bar B$
and predicate-contractivity of the underlying category~$\C$);
  \item for a limit ordinal $\alpha$, $\invp^{\bar B,c,\alpha} R =$ $\bigand_{\beta<\alpha} \invp^{\bar B,c,\beta} R$
  (we regard $0$ as a limit ordinal, and thus $\invp^{\bar B,c,0} R = \top$ as a meet of an empty family).
\end{itemize}
By definition, the family $(\invp^{\bar B,c,\alpha} R)_{\alpha}$ is decreasing along the 
ordinal chain and the members of the family range over subobjects of ${X\times X}$. 
By our global assumptions, subobjects of $X\times X$ form a set. Therefore, there is a least 
ordinal, call it $\nu$, such that $\invp^{\bar B,c,\nu+1} R =$ $\invp^{{\bar B},c,\nu} R$.
\end{definition}
\begin{lemma}\label{lem:box-nu}
For any coalgebra $c\colon\mS\to B(\mS,\mS)$, and any relation $R \pred{} X \times X$,
$\invp^{\bar B,c,\nu}R\leq\invp(\invp^{\bar B,c,\nu} R,R)$.
\end{lemma}
\begin{proof}
By definition, $\invp^{\nu}R  = R\land\invp^{\nu}R \land\iimg{c}{\overline{B}(\invp^{\nu}R ,\invp^{\nu}R )}$,
hence $\invp^{\nu}R  \leq R$ and $\invp^{\nu}R  \leq \iimg{c}{\overline{B}(\invp^{\nu}R, \invp^{\nu}R)}$. 
This entails that $\invp^{\nu}R \leq\invp(\invp^{\nu}R ,R)$, by definition of $\invp$. 
\end{proof}

\begin{theorem}\label{thm:indbox}\sgnote{Explore the case of $\logp P$ with $P\neq\top$;
identify clauses, independent from the existence of a GSOS law.}
Given a coalgebra $\gamma \colon\mS\to B(\mS,\mS)$ induced by a higher-order GSOS law. Then
\begin{enumerate}
  \item\label{it:indbox1} every $\logp^\alpha\top$ is a $\ol\Sigma$-congruence; %
  \item\label{it:indbox2} every $\logp^\alpha\top$ is reflexive; %
  \item\label{it:indbox3} every $\logp^\alpha\top$ is transitive;
  \item\label{it:indbox5} $\invp^\nu\top=\invp(\Delta,\top)$;
  \item\label{it:indbox6} $\invp^\nu\top=\invp\top$ (when\,\,$\invp\top$ exists).
\end{enumerate}
\end{theorem}
\begin{proof}
Clause~\ref{it:indbox1} is shown by transfinite induction on $\alpha$. For a limit ordinal 
$\alpha$ we reduce to the induction hypothesis as follows:
\begin{align*}
\fimg{\iota}{\ol\Sigma(\invp^\alpha\top)} 
=&\;    \fimg{\iota}{\ol\Sigma(\bigand_{\beta<\alpha}\invp^\beta \top)}\\
\leq&\; \bigand_{\beta<\alpha}\fimg{\iota}{\ol\Sigma(\invp^\beta \top)}\\ 
\leq&\; \bigand_{\beta<\alpha}\invp^\beta \top\\
   =&\; \invp^\alpha\top.
\end{align*}
Let us proceed with the successor ordinal 
case. That is, suppose that $\invp^\alpha\top$ is a congruence,
and show that $\invp^{\alpha+1} \top$ is so, equivalently:
\begin{align*}
&\fimg{\iota}{\ol\Sigma(\invp^{\alpha+1}\top)} \leq \logp^{\alpha+1}\top= \invp^{\alpha}\top\land\iimg{\gamma}{\overline{B}(\invp^{\alpha+1}\top,\invp^{\alpha}\top)},
\end{align*}
\ST{
  $\fimg{\iota}{\ol\Sigma(\invp^{\alpha+1}\top)} \leq \logp^{\alpha+1}\top=
  \invp^{\alpha}\top\land\iimg{\gamma}{\overline{B}(\invp^{\alpha}\top,\invp^{\alpha}\top)},$}

which reduces to two subgoals:
\begin{align}
\fimg{\iota}{\ol\Sigma(\invp^{\alpha+1}\top)}\leq&\; \invp^{\alpha}\top\label{eq:abox1}\\
\fimg{\iota}{\ol\Sigma(\invp^{\alpha+1}\top)}\leq&\; \iimg{\gamma}{\overline{B}(\invp^{\alpha+1}\top,\invp^{\alpha}\top)}\label{eq:abox2}
\end{align}

\ST{$\fimg{\iota}{\ol\Sigma(\invp^{\alpha+1}\top)}\leq\; \iimg{\gamma}{\overline{B}(\invp^{\alpha}\top,\invp^{\alpha}\top)}$}

By definition, and by induction hypothesis, 
\begin{align*}
\fimg{\iota}{\ol\Sigma(\invp^{\alpha+1}\top)}\leq \fimg{\iota}{\ol\Sigma(\invp^{\alpha}\top)}\leq\invp^{\alpha}\top,
\end{align*}
which yields~\eqref{eq:abox1}. To prove~\eqref{eq:abox2}, equivalently replace it 
with 
\begin{align*}
\fimg{(\gamma\comp\iota)}{\ol\Sigma(\invp^{\alpha+1}\top)}\leq {\overline{B}(\invp^{\alpha+1}\top,\invp^{\alpha}\top)},
\end{align*}

\ST{$\fimg{(\gamma\comp\iota)}{\ol\Sigma(\invp^{\alpha+1}\top)}\leq {\overline{B}(\invp^{\alpha}\top,\invp^{\alpha}\top)},$}

which is then shown as follows\sgnote{Replace by a diagram (?)}:
\begin{flalign*}
\qquad\fimg{(\gamma\comp\,&\iota)}{\ol\Sigma(\invp^{\alpha+1}\top)} \\
 &\; \leq\fimg{(B(\mS,\hat\iota\comp\Sigmas\nabla)\comp\rho\comp\Sigma\brks{\id,\gamma})}{\ol\Sigma(\invp^{\alpha+1}\top)}\\
&\tag*{\by{Lemma~?}}\\
 &\; \leq\fimg{(B(\mS,\hat\iota\comp\Sigmas\nabla)\comp\rho)}{\fimg{(\Sigma\brks{\id,\gamma})}{\ol\Sigma(\invp^{\alpha+1}\top)}} & \tag*{\by{Lemma~?}}\\
&\;\leq\fimg{(B(\mS,\hat\iota\comp\Sigmas\nabla)\comp\rho)}{\\&\qquad\ol\Sigma{(\invp^{\alpha+1}\top\times\fimg{\gamma}{\invp^{\alpha}\top\land\iimg{\gamma}{\overline{B}(\invp^{\alpha+1}\top,\invp^{\alpha}\top)}})}}\\
&\;\leq\fimg{(B(\mS,\hat\iota\comp\Sigmas\nabla)\comp\rho)}{\\&\textcolor{red}{\qquad\ol\Sigma{(\invp^{\alpha+1}\top\times\fimg{\gamma}{\invp^{\alpha}\top\land\iimg{\gamma}{\overline{B}(\invp^{\alpha}\top,\invp^{\alpha}\top)}})}}}\\
&\;\leq\fimg{(B(\mS,\hat\iota\comp\Sigmas\nabla)\comp\rho)}{\ol\Sigma{(\invp^{\alpha+1}\top,\times\overline{B}(\invp^{\alpha+1}\top,\invp^{\alpha}\top))}}\\
&\textcolor{red}{\;\leq\fimg{(B(\mS,\hat\iota\comp\Sigmas\nabla)\comp\rho)}{\ol\Sigma{(\invp^{\alpha+1}\top\times\overline{B}(\invp^{\alpha}\top,\invp^{\alpha}\top))}}}\\   
&\tag*{\by{Lemma~?}}\\
 &\;\leq\fimg{B(\mS,\hat\iota\comp\Sigmas\nabla)}{\fimg{\rho}{\ol\Sigma{(\invp^{\alpha+1}\top\times\overline{B}(\invp^{\alpha+1}\top,\invp^{\alpha}\top))}}}& \\
&\tag*{\by{\eqref{eq:gsoslifting}}}\\
&\textcolor{red}{\;\leq\fimg{B(\mS,\hat\iota\comp\Sigmas\nabla)}{\fimg{\rho}{\ol\Sigma{(\invp^{\alpha+1}\top\times\overline{B}(\invp^{\alpha}\top,\invp^{\alpha}\top))}}}}&\\
&\textcolor{red}{\;\leq\fimg{B(\mS,\hat\iota\comp\Sigmas\nabla)}{\fimg{\rho}{\ol\Sigma{(\invp^{\alpha}\top\times\overline{B}(\invp^{\alpha}\top,\invp^{\alpha}\top))}}}}&\\    
 &\;\leq \fimg{B(\mS,\hat\iota\comp\Sigmas\nabla)}{\ol{B}(\logp^{\alpha+1}\top, \ol{\Sigma}^\star (\logp^{\alpha+1}\top+\invp^{\alpha}\top))}\\
&\tag*{\by{Lemma~?}}\\
&\textcolor{red}{\;\leq \fimg{B(\mS,\hat\iota\comp\Sigmas\nabla)}{\ol{B}(\logp^{\alpha}\top, \ol{\Sigma}^\star (\logp^{\alpha}\top+\invp^{\alpha}\top))}}\\  
 &\;\leq \fimg{B(\mS,\hat\iota)}{\fimg{B(\mS,\Sigmas\nabla)}{\ol{B}(\logp^{\alpha+1}\top, \ol{\Sigma}^\star (\logp^{\alpha+1}\top+\invp^{\alpha}\top))}}\\
&\textcolor{red}{\;\leq \fimg{B(\mS,\hat\iota)}{\fimg{B(\mS,\Sigmas\nabla)}{\ol{B}(\logp^{\alpha}\top, \ol{\Sigma}^\star (\logp^{\alpha}\top+\invp^{\alpha}\top))}}}\\
&\;\leq \fimg{B(\mS,\hat\iota)}{\ol{B}(\logp^{\alpha+1}\top, \ol{\Sigma}^\star (\invp^{\alpha}\top))}&\\
&\textcolor{red}{\;\leq \fimg{B(\mS,\hat\iota)}{\ol{B}(\logp^{\alpha}\top, \ol{\Sigma}^\star (\invp^{\alpha}\top))}}&\\
&\;\leq \ol{B}(\logp^{\alpha+1}\top, \fimg{\hat\iota}{\ol{\Sigma}^\star (\invp^{\alpha}\top)})&\\
&\tag*{\by{induction}}\\
&\textcolor{red}{\;\leq \ol{B}(\logp^{\alpha}\top, \fimg{\hat\iota}{\ol{\Sigma}^\star (\invp^{\alpha}\top)})}&\\  
&\;\leq {\overline{B}(\invp^{\alpha+1}\top,\invp^{\alpha}\top)},& \\
&\textcolor{red}{\;\leq {\overline{B}(\invp^{\alpha}\top,\invp^{\alpha}\top)}},&
\end{flalign*}
and we are done with the proof of~\ref{it:indbox1}. Clause~\ref{it:indbox2} now follows by noting that $\Delta=\fimg{\hat\iota}{\ol\Sigma^\star\emptyset}$.

Let us show~\ref{it:indbox3} by transfinite induction over $\alpha$. The limit ordinal 
case is clear. Let us handle the successor ordinal case. Suppose that $\logp^{\alpha}\top$
is transitive. Then
\begin{flalign*}
 \logp^{\alpha}&{}^{\,+1}\top\comp \logp^{\alpha+1}\top \\
  &\;= (\logp^{\alpha}\top\land\iimg{\gamma}{\ol{B}(\invp^{\alpha+1}\top,\invp^{\alpha}\top)})\\
  &\qquad\comp (\logp^{\alpha}\top\land\iimg{\gamma}{\ol{B}(\invp^{\alpha+1}\top,\invp^{\alpha}\top)} ) &\\
  &\;\leq (\logp^{\alpha}\top\comp\logp^{\alpha}\top)\land \iimg{\gamma}{\ol{B}(\invp^{\alpha+1}\top,\invp^{\alpha}\top)}\comp \iimg{\gamma}{\ol{B}(\invp^{\alpha+1}\top,\invp^{\alpha}\top)}\\
  &\tag*{\by{induction}}\\
  &\;\leq \logp^{\alpha}\top\land \iimg{\gamma}{\ol{B}(\invp^{\alpha+1}\top,\invp^{\alpha}\top)}\comp \iimg{\gamma}{\ol{B}(\invp^{\alpha+1}\top,\invp^{\alpha}\top)}&\\
  &\;\tag*{\by{reflexivity}}\\
  &\;\leq \logp^{\alpha}\top\land \iimg{\gamma}{\ol{B}(\invp^{\alpha+1}\top,\invp^{\alpha}\top)}\comp \iimg{\gamma}{\ol{B}(\Delta,\invp^{\alpha}\top)}& \\
  &\;\leq \logp^{\alpha}\top\land \iimg{\gamma}{\ol{B}(\invp^{\alpha+1}\top,\invp^{\alpha}\top)\comp\ol{B}(\Delta,\invp^{\alpha}\top)}&\\
  &\tag*{\by{\Cref{ass:lics}~(2)}}\\
  &\;\leq \logp^{\alpha}\top\land \iimg{\gamma}{\ol{B}(\invp^{\alpha+1}\top,\invp^{\alpha}\top\comp\invp^{\alpha}\top)}&\\
  &\;\leq \logp^{\alpha}\top\land \iimg{\gamma}{\ol{B}(\invp^{\alpha+1}\top,\invp^{\alpha}\top)}&\\
  &\tag*{\by{induction}}\\*
  &\;= \logp^{\alpha+1}\top,
\end{flalign*}
and we are done.

We show~\ref{it:indbox5} by mutual inequality. The inequality $\invp^\alpha\top\leq\invp(\Delta,\top)$ 
follows from~\ref{it:indbox2} and~\Cref{lem:box-nu}. Let us show $\invp(\Delta,\top)\leq\invp^\alpha\top$ by transfinite induction 
over $\alpha$. The limit ordinal case is clear. To show the successor ordinal
case, suppose that $\invp(\Delta,\top)\leq\invp^\alpha\top$. Then
\begin{flalign*}
&& \invp(\Delta,\top) &\;= \iimg{\gamma}{\ol{B}(\Delta,\invp(\Delta,\top))} &\\*
&&  &\;\leq \logp^{\alpha}\top\land\iimg{\gamma}{\ol{B}(\Delta,\invp^{\alpha}\top)}&\by{induction}\\
&&  &\;\leq \logp^{\alpha}\top\land\iimg{\gamma}{\ol{B}(\invp^{\alpha}\top,\invp^{\alpha}\top)} &\by{\Cref{lem:popl}}\\
&&  &\;\leq \logp^{\alpha}\top\land\iimg{\gamma}{\ol{B}(\invp^{\alpha+1}\top,\invp^{\alpha}\top)}\\
&&  &\;=    \invp^{\alpha+1}\top
\end{flalign*}
where the call of \Cref{lem:popl} is justified by~\ref{it:indbox1} and~\ref{it:indbox3}.

Finally, let us show~\ref{it:indbox6}. By~\ref{it:indbox5}, to show~\ref{it:indbox6}, is to show $\logp\top = \logp(\Delta,\top)$. 
Since $\logp\top$ is the unique solution of the equation $\logp\top=\logp(\logp\top,\top)$,
it suffices to show that $\logp(\Delta,\top)=\logp(\logp(\Delta,\top),\top)$. 
By combining clauses~\ref{it:indbox2} and~\ref{it:indbox5}, we obtain $\Delta\leq\logp(\Delta,\top)$, which entails $\logp(\logp(\Delta,\top),\top)\leq\logp(\Delta,\top)$.
We are thus left to prove the converse inequality: $\logp(\Delta,\top)\leq\logp(\logp(\Delta,\top),\top)$.
To that end, it suffices to show that $\logp(\Delta,\top)$ is a $\logp(\Delta,\top)$-relative 
invariant, i.e.\ that 
\begin{displaymath}
  \logp(\Delta,\top)\leq\iimg{\gamma}{\ol B(\logp(\Delta,\top),\logp(\Delta,\top))}.
\end{displaymath}
By clauses~\ref{it:indbox1},~\ref{it:indbox3} and~\ref{it:indbox5}, $\logp(\Delta,\top)$
is a transitive congruence. Therefore, using \Cref{lem:popl},
\begin{displaymath}
  \logp(\Delta,\top) = \iimg{\gamma}{\ol B(\Delta,\logp(\Delta,\top))}\leq\iimg{\gamma}{\ol B(\logp(\Delta,\top),\logp(\Delta,\top))},
\end{displaymath}
and we are done.
\end{proof}
}

\section*{Appendix: Omitted proofs and details}
\label{app:proofs}
This appendix provides all omitted proofs and additional details on our examples, ordered by section.

\section{Details for Section 3}

\subsection*{Full Definition of the Higher-Order GSOS Law \texorpdfstring{$\boldsymbol{\rho}$}{$\rho$} for \texorpdfstring{$\boldsymbol{\mu}$}{$\mu$}CTL (\autoref{sec:mu-ctl})}

\begin{equation*}
\begin{aligned}
  &\rho_{X,Y} \c&\hspace{0em} \Sigma(X \times B(X,Y)) & \to B(X, \Sigma^\star (X+Y)) \\*
  &\rho_{X,Y}\makebox[0pt][l]{$(tr) = \texttt{case}~tr~\texttt{of}$} \\
  &\quad\makebox[0em][l]{$S_{\tau_{1},\tau_{2},\tau_{3}}$} 				&& \mapsto~ \lambda e.\,S'_{\tau_{1},\tau_{2},\tau_{3}}(e)    \\
	&\quad\makebox[0em][l]{$S'_{\tau_{1},\tau_{2},\tau_{3}}(t,f)$} 	&& \mapsto~ \lambda e.\,S''_{\tau_{1},\tau_{2},\tau_{3}}(t,e) \\
	&\quad\makebox[0em][l]{$S''_{\tau_{1},\tau_{2},\tau_{3}}((t,f),(s,g))$}  
	 																																&& \mapsto~ \lambda e.\,\mathsf{app}_{\tau_{2},\tau_{3}}(\mathsf{app}_{\tau_1,\arty{\tau_{2}}{\tau_{3}}}(t, e),\mathsf{app}_{\tau_1,\tau_2}(s, e)) \\  
  &\quad\makebox[0em][l]{$K_{\tau_{1},\tau_{2}}$} 								&& \mapsto~ \lambda e.\,K'_{\tau_{1},\tau_{2}}(e) \\
  &\quad\makebox[0em][l]{$K'_{\tau_{1},\tau_{2}}(t,f)$} 					&& \mapsto~ \lambda e.\,t \\
  &\quad\makebox[0em][l]{$I_{\tau}$} 															&& \mapsto~ \lambda e.\,e \\
  &\quad\makebox[0em][l]{$\inl_{\tau_1,\tau_2}(t,f)$} 						&& \mapsto~t \in X_{\tau_{1}}\\
  &\quad\makebox[0em][l]{$\inr_{\tau_1,\tau_2}(t,f)$} 						&& \mapsto~t \in X_{\tau_{2}}\\
  &\quad\makebox[5em][l]{$\mathsf{case}_{\tau_1,\tau_2,\tau_3}((t,f),(s,g),(r,h))$} 		
  													             													&& \mapsto~ \begin{cases}
																																							 \makebox[10em][l]{$\mathsf{case}_{\tau_1,\tau_2,\tau_3}(f,s,r)$} \text{\quad if $f \in Y_{\tau_1\boxplus\tau_2}$} \\
																																							 \makebox[10em][l]{$\mathsf{app}_{\tau_{1},\tau_{3}}(s,t')$} \text{\quad if $f=\inl_{\tau_{1},\tau_{2}}(t')$} \\
																																							 \makebox[10em][l]{$\mathsf{app}_{\tau_{2},\tau_{3}}(r,t')$} \text{\quad if $f=\inr_{\tau_{1},\tau_{2}}(t')$}
																																						  \end{cases} \\
  &\quad\makebox[0em][l]{$\fst_{\tau_1,\tau_2}(t,f)$} 						&& \mapsto~ \begin{cases}
																																							 \makebox[10em][l]{$\fst_{\tau_1,\tau_2}(f)$} \text{\quad if $f\in Y_{\tau_1\boxtimes\tau_2}$} \\
																																							 \makebox[10em][l]{$\outl(f)$} \text{\quad if $f\in Y_{\tau_1}\times Y_{\tau_2}$}
																																						  \end{cases} \\
  &\quad\makebox[0em][l]{$\snd_{\tau_1,\tau_2}(t,f)$} 						&& \mapsto~ \begin{cases}
																																							 \makebox[10em][l]{$\snd_{\tau_1,\tau_2}(f)$} \text{\quad if $f\in Y_{\tau_1\boxtimes\tau_2}$} \\
																																							 \makebox[10em][l]{$\outr(f)$} \text{\quad if $f\in Y_{\tau_1}\times Y_{\tau_2}$}
																																						  \end{cases} \\
	&\quad\makebox[0em][l]{$\mathsf{pair}_{\tau_1,\tau_2}((t,f),(s,g))$} 
	                                                    						&& \mapsto~ (t,s)\\
  &\quad\makebox[0em][l]{$\mathsf{app}_{\tau_1,\tau_2}((t,f),(s,g))$} 										&& \mapsto~ \begin{cases}
																																							 \makebox[10em][l]{$f(s)$} \text{\quad if $f \in Y_{\tau_{2}}^{X_{\tau_{1}}}$} \\
																																							 \makebox[10em][l]{$\mathsf{app}_{\tau_1,\tau_2}(f, s)$} \text{\quad if $f \in Y_{\arty{\tau_{1}}{\tau_{2}}}$}
																																						  \end{cases} \\
  &\quad\makebox[0em][l]{$\mathsf{fold}_\tau(t,f)$} 							&& \mapsto~ t \in Y_{\tau[\mu\alpha.\tau/\alpha]} \\
  &\quad\makebox[0em][l]{$\mathsf{unfold}_\tau(t,f)$} 					  && \mapsto~ \begin{cases}
      																																				 \makebox[10em][l]{$\mathsf{unfold}_\tau(f)$} \text{\quad if $f \in Y_{\mu \alpha.\,\tau}$} \\
      																																				 \makebox[10em][l]{$f$} \text{\quad if $f \in Y_{\tau[\mu\alpha.\,\tau/\alpha]}$}
    																																				 \end{cases} 
\end{aligned}
\end{equation*}

\subsection*{Proof of \Cref{th:l-is-cong}}
  We show by induction on $n$ that $\logrel^{n}$ is a congruence
  for all $n \in \mathds{N}$. For $n = 0$, $\logrel^{0} = \top$
  and thus the base case is immediate. For the induction step, we
  assume that $\logrel^{n}$ is a congruence.  By the definition of
  $\logrel^{n+1}$, for each operation $\mathsf{f} \in \Sigma$ of
  arity $\tau_{1},\tau_{2}\,\dots,\tau_{n}$ and pairs
  $t_{i},s_{i} \c \tau_{i}$ such that
  $\logrel^{n+1}_{\tau_{1}}(t_{1},s_{1}) \land \dots \land
  \logrel^{n+1}_{\tau_{n}}(t_{n},s_{n})$, we have to show the
  following:
  \begin{enumerate}[label=(\roman*)]
  \item \label{cong:enum1}
    $\logrel^{n}_{\tau_{\mathsf{f}}}(\mathsf{f}(t_{1},\dots,t_{n}),
    \mathsf{f}(s_{1},\dots,s_{n}))$.
  \item \label{cong:enum2}
    If $\mathsf{f}(t_{1},\dots,t_{n}) \to t$, then
    $\exists s.\,\mathsf{f}(s_{1},\dots,s_{n}) \To s$ with
    $\logrel^{n}_{\tau_{\mathsf{f}}}(t,s)$.
  \item \label{cong:enum3} We have four subcases based on the type of
    $\tau_{\mathsf{f}}$.
    \begin{enumerate}[label=(\alph*)]
    \item \emph{Subcase $\tau_{\mathsf{f}} = \tau_{1} \boxplus \tau_{2}$.} We
      have that if
      $\mathsf{f}(t_{1},\dots,t_{n}) \xto{\boxplus_{1}}
      t$ then $\exists s.\,\mathsf{f}(s_{1},\dots,s_{n})
      \xTo{\boxplus_{1}} s$ and $\logrel^{n}_{\tau_{1}}(t,s)$. If
      $\mathsf{f}(t_{1},\dots,t_{n}) \xto{\boxplus_{2}}
      t$ then $\exists s.\,\mathsf{f}(s_{1},\dots,s_{n})
      \xTo{\boxplus_{2}} s$ and $\logrel^{n}_{\tau_{2}}(t,s)$.
    \item \emph{Subcase $\tau_{\mathsf{f}} = \tau_{1} \boxtimes \tau_{2}$.} If
      $\mathsf{f}(t_{1},\dots,t_{n}) \xto{\boxtimes_{1}}
      t_{1}$ and $\mathsf{f}(t_{1},\dots,t_{n}) \xto{\boxtimes_{2}}
      t_{2}$ then $\exists s_{1},s_{2}.\,\mathsf{f}(s_{1},\dots,s_{n})
      \xTo{\boxtimes_{1}}
      s_{1}$ and $\mathsf{f}(s_{1},\dots,s_{n}) \xTo{\boxtimes_{2}}
      s_{2}$; moreover, $\logrel^{n}_{\tau_{1}}(t_{1},s_{1})$ and
      $\logrel^{n}_{\tau_{2}}(t_{2},s_{2})$.
    \item \emph{Subcase $\tau_{\mathsf{f}} = \mu \alpha.\tau$.} If
      $\mathsf{f}(t_{1},\dots,t_{n}) \xto{\mu} t$ then
      $\exists s.\,\mathsf{f}(s_{1},\dots,s_{n}) \xTo{\mu} s$ and
      $\logrel^{n}_{\tau[\mu\alpha.\tau/\alpha]}(t,s)$.
    \item \emph{Subcase $\tau_{\mathsf{f}} = \arty{\tau}{\tau'}$}.
      For all terms $e_{1},e_{2}
      \c \tau$ with $\logrel^{n}_{\tau}(e_{1},e_{2})$,
      if $\mathsf{f}(t_{1},\dots,t_{n}) \xto{e_{1}} t$, then
      $\exists s.\,\mathsf{f}(s_{1},\dots,s_{n}) \xTo{e_{2}} s$ and
      $\logrel^{n}_{\tau'}(t,s)$.
    \end{enumerate}
  \end{enumerate}
  We begin with \ref{cong:enum1}. Let $\mathsf{f} \in \Sigma$ be an operation.
  By the definition of $\logrel^{n+1}$, $\logrel^{n+1} \leq
  \logrel^{n}$ and thus
  \[
    \logrel^{n+1}_{\tau_{1}}(t_{1},s_{1}) \land \dots \implies
    \logrel^{n}_{\tau_{1}}(t_{1},s_{1}) \land \dots
  \]
  By the induction hypothesis we may conclude
  \[
    \logrel^{n}_{\tau_{\mathsf{f}}}(\mathsf{f}(t_{1},\dots,t_{n}),
    \mathsf{f}(s_{1},\dots,s_{n})),
  \]
  finishing the case of \ref{cong:enum1}. For \ref{cong:enum2} and
  \ref{cong:enum3} we proceed by case distinction on the operation
  $\mathsf{f} \in \Sigma$.
  \begin{enumerate}
  \item \emph{Case $\inl_{\tau_{1},\tau_{2}}$.} We assume
    $\logrel^{n+1}_{\tau_{1}}(t,s)$.
    \begin{enumerate}[label=(\roman*)]
      \setcounter{enumii}{1}
    \item This clause is void, as $\inl$ does not reduce.
    \item We have $\inl(t) \xto{\boxplus_{1}} t$, $\inl(s) \xto{\boxplus_{1}} s$
      and $\logrel^{n+1}_{\tau_{1}}(t,s)$, thus
      $\logrel^{n}_{\tau_{1}}(t,s)$.
    \end{enumerate}
  \item \emph{Case $\inr$.} Similar to $\inl$.
  \item \emph{Case of $\mathsf{case}_{\tau_{1},\tau_{2},\tau_{3}}$ expression.}
    We assume $\logrel^{n+1}_{\tau_{1} \boxplus \tau_{2}}(t_{1},s_{1})$,
    $\logrel^{n+1}_{\arty{\tau_{1}}{\tau_{3}}}(t_{2},s_{2})$ and
    $\logrel^{n+1}_{\arty{\tau_{2}}{\tau_{3}}}(t_{3},s_{3})$.
    \begin{enumerate}[label=(\roman*)]
      \setcounter{enumii}{1}
    \item First, $t_{1} \to t_{1}'$ and $\mathsf{case}(t_{1},t_{2},t_{3}) \to
      \mathsf{case}(t_{1}',t_{2},t_{3})$. Thus, by
      $\logrel^{n+1}_{\tau_{1} \boxplus \tau_{2}}(t_{1},s_{1})$,
      $\exists s_{1}'.\,s_{1} \To s_{1}'$ and
      $\logrel^{n}_{\arty{\tau_{1}}{\tau_{2}}}(t_{1}',s_{1}')$.
      Thus, by the semantics we can conclude that
      $\mathsf{case}(s_{1},s_{2},s_{3}) \To \mathsf{case}(s_{1}',s_{2},s_{3})$.
      It now suffices to show that
      $\logrel^{n}_{\tau_{3}}(\mathsf{case}(t_{1}',t_{2},t_{3}),
      \mathsf{case}(s_{1}',s_{2},s_{3}))$, which
      follows from the inductive hypothesis,
      namely that $\logrel^{n}$ is a $\Sigma$-congruence.

      Second, $t_{1} \xto{\boxplus_{1}} t_{1}'$ leading to
      $\mathsf{case}(t_{1},t_{2},t_{3}) \to t_{2}\app{}{}t_{1}'$. By
      $\logrel^{n+1}_{\tau_{1} \boxplus \tau_{2}}(t_{1},s_{1})$,
      $\exists s_{1}'.\,s_{1} \xTo{\boxplus_{1}} s_{1}'$ and
      $\logrel^{n}_{\arty{\tau_{1}}{\tau_{2}}}(t_{1}',s_{1}')$.
      The term $\mathsf{case}(s_{1},s_{2},s_{3})$ will gradually
      evaluate $s_{1}$ until it becomes a value, to which point it will choose
      the ``left'' path. In other words,  $\mathsf{case}(s_{1},s_{2},s_{3}) \To
      s_{2} \app{}{} s_{1}'$. Furthermore, the inductive hypothesis yields
      $\logrel^{n}_{\tau_{3}}(t_{1}' \app{}{} t_{2},s_{1}'\app{}{}s_{2})$.
      The final subclause for $t_{1} \xto{\boxplus_{2}} t_{1}'$ is done
      analogously.
    \item This clause is void.
    \end{enumerate}
  \item \emph{Case $\fst_{\tau_{1},\tau_{2}}$.} We assume
    $\logrel^{n+1}_{\tau_{1} \boxtimes \tau_{2}}(t_{1},s_{1})$.
    \begin{enumerate}[label=(\roman*)]
      \setcounter{enumii}{1}
    \item First, $t_{1} \to t_{1}' \implies \fst(t_{1}) \to \fst(t_{1}')$.
      Hence, by our assumption $\logrel^{n+1}_{\tau_{1} \boxtimes
        \tau_{2}}(t_{1},s_{1})$, $\exists s_{1}'.\,s_{1} \To
      s_{1}'$ and $\logrel^{n}_{\tau_{1} \boxtimes
        \tau_{2}}(t_{1}',s_{1}')$. This yields $\fst(s_{1}) \To \fst(s_{1}')$;
      furthermore, $\logrel_{\tau_{1} \boxtimes
        \tau_{2}}^{n}(\fst(t_{1}'),\fst(s_{1}'))$ by the inductive hypothesis.
    \item This clause is void.
    \end{enumerate}
  \item \emph{Case $\snd$.} This is done similarly to $\fst$.
  \item \emph{Case $\mathsf{pair}_{\tau_{1},\tau_{2}}$.} We assume
    $\logrel^{n+1}_{\tau_{1}}(t_{1},s_{1})$ and
    $\logrel^{n+1}_{\tau_{2}}(t_{2},s_{2})$.
    \begin{enumerate}[label=(\roman*)]
      \setcounter{enumii}{1}
    \item This clause is void.
    \item We see that $\mathsf{pair}_{\tau_{1},\tau_{2}}(t_{1},t_{2})
      \xto{\boxtimes_{1}} t_{1}$ and $\mathsf{pair}_{\tau_{1},\tau_{2}}(t_{1},t_{2})
      \xto{\boxtimes_{2}} t_{2}$. Conversely, $\mathsf{pair}_{\tau_{1},\tau_{2}}(s_{1},s_{2})
      \xto{\boxtimes_{1}} s_{1}$ and $\mathsf{pair}_{\tau_{1},\tau_{2}}(t_{1},t_{2})
      \xto{\boxtimes_{2}} s_{2}$, and our assumptions allow us to conclude this clause.
    \end{enumerate}
  \item \emph{Case of application.} We assume
    $\logrel^{n+1}_{\arty{\tau_{1}}{\tau_{2}}}(t_{1},s_{1})$ and
    $\logrel^{n+1}_{\tau_{1}}(t_{2},s_{2})$.
    \begin{enumerate}[label=(\roman*)]
      \setcounter{enumii}{1}
    \item By the semantics in \Cref{def:logrel}, $t_{1} \app{}{} t_{2} \to t$ in
      two ways. First,
      $t_{1} \to t_{1}'$ and $t = t_{1}' \app{}{} t_{2}$. By
      $\logrel^{n+1}_{\arty{\tau_{1}}{\tau_{2}}}(t_{1},s_{1})$, there exists
      $s_{1}'$ such that $s_{1} \To s_{1}'$ and
      $\logrel^{n}_{\arty{\tau_{1}}{\tau_{2}}}(t_{1}',s_{1}')$. Hence, the
      semantics dictate $s_{1} \app{}{} s_{2} \To s_{1}'\app{}{}s_{2}$. It now
      suffices to show that $\logrel^{n}_{\tau_{2}}(t_{1}' \app{}{}
      t_{2},s_{1}'\app{}{}s_{2})$, which follows from the inductive hypothesis,
      namely that $\logrel^{n}$ is a $\Sigma$-congruence.
      Second, $t_{1} \app{}{} t_{2} \to t$ by $t_{1} \xto{t_{2}} t$.
      By $\logrel^{n+1}_{\arty{\tau_{1}}{\tau_{2}}}(t_{1},s_{1})$, we can
      conclude that $s_{1} \xTo{s_{2}} s$.
      In addition, by virtue of
      $\logrel^{n+1}_{\arty{\tau_{1}}{\tau_{2}}}(t_{1},s_{1})$ and
      $\logrel^{n+1}_{\tau_{1}}(t_{2},s_{2})$ (thus
      $\logrel^{n}_{\tau_{1}}(t_{2},s_{2})$),
      $\logrel^{n}_{\tau_{2}}(t,s)$, finishing the clause.
    \item This clause is void, as application expressions do not perform
      labelled transitions.
    \end{enumerate}
  \item \emph{Case $I$.}
    \begin{enumerate}[label=(\roman*)]
      \setcounter{enumii}{1}
    \item This clause is void.
    \item Given $\logrel^{n}_{\tau_{2}}(e_{1},e_{2})$, and since $I \xto{e}
      e$, it suffices to show that $\logrel^{n}_{\tau_{2}}(e_{1},e_{2})$
      which is a tautology.
    \end{enumerate}
  \item \emph{Case $K_{\tau_{1},\tau_{2}}$.}
    \begin{enumerate}[label=(\roman*)]
      \setcounter{enumii}{1}
    \item This clause is void.
    \item Given $\logrel^{n}_{\tau_{1}}(e_{1},e_{2})$, and since $K \xto{e}
      K'(e)$, it suffices to show that
      $\logrel^{n}_{\arty{\tau_{2}}{\tau_{1}}}(K'(e_{1}),K'(e_{2}))$. The
      latter holds by the induction hypothesis.
    \end{enumerate}
  \item \emph{Case $K'_{\tau_{1},\tau_{2}}$.} We assume
    $\logrel^{n+1}_{\tau_{1}}(t_{1},s_{1})$.
    \begin{enumerate}[label=(\roman*)]
      \setcounter{enumii}{1}
    \item This clause is void.
    \item Given $\logrel^{n}_{\tau_{2}}(e_{1},e_{2})$ and the rule $K'(t)
      \xto{e} t$, it suffices to show
      that $\logrel^{n}_{\tau_{1}}(t_{1},s_{1})$, which holds by the
      assumption $\logrel_{\tau_{1}}^{n+1}(t_{1},s_{1})$.
    \end{enumerate}
  \item \emph{Case $S_{\tau_{1},\tau_{2},\tau_{3}}$.}
    \begin{enumerate}[label=(\roman*)]
      \setcounter{enumii}{1}
    \item This clause is void.
    \item Given
      $\logrel^{n}_{\arty{\tau_{1}}{\arty{\tau_{2}}{\tau_{3}}}}(e_{1},e_{2})$,
      and since $S \xto{e}
      S'(e)$, it suffices to show that
      $\logrel^{n}_{\arty{{(\arty{\tau_{1}}{\tau_{2})}}}{\arty{\tau_{1}}{\tau_{3}}}}
      (S'(e_{1}),S'(e_{2}))$.
      The latter holds by the induction hypothesis.
    \end{enumerate}
  \item \emph{Case $S'_{\tau_{1},\tau_{2},\tau_{3}}$.} We assume
    $\logrel^{n+1}_{\arty{\tau_{1}}{\arty{\tau_{2}}{\tau_{3}}}}(t_{1},s_{1})$.
    \begin{enumerate}[label=(\roman*)]
      \setcounter{enumii}{1}
    \item This clause is void.
    \item Given $\logrel^{n}_{\arty{\tau_{1}}{\tau_{2}}}(e_{1},e_{2})$ and
      since $S'(t) \xto{e} S''(t,e)$, it suffices to show
      $\logrel^{n}_{\arty{\tau_{1}}{\tau_{3}}}(S''(t_{1},e_{1}),S''(t_{2},e_{2}))$,
      which holds by the induction hypothesis.
    \end{enumerate}
  \item \emph{Case $S''_{\tau_{1},\tau_{2},\tau_{3}}$.} Assume
    $\logrel^{n+1}_{\arty{\tau_{1}}{\arty{\tau_{2}}{\tau_{3}}}}(t_{1},s_{1})$
    and $\logrel^{n+1}_{\arty{\tau_{1}}{\tau_{2}}}(t_{2},s_{2})$.
    \begin{enumerate}[label=(\roman*)]
      \setcounter{enumii}{1}
    \item This clause is void.
    \item Given $\logrel^{n}_{\tau_{1}}(e_{1},e_{2})$ and the rule $S'(t,s)
      \xto{e} (t \app{}{} e) \app{}{} (s \app{}{}e)$, it suffices to show
      $\logrel^{n}_{\tau_{3}}((t_{1} \app{}{} e_{1}) \app{}{} (s_{1} \app{}{}e_{1}),
      (t_{2} \app{}{} e_{2}) \app{}{} (s_{2} \app{}{}e_{2}))$, which holds by the
      induction hypothesis.
    \end{enumerate}
  \item \emph{Case $\mathsf{fold}_{\tau}$.} We assume
    $\logrel^{n+1}_{\tau[\mu\alpha.\tau/\alpha]}(t_{1},s_{1})$.
    \begin{enumerate}[label=(\roman*)]
      \setcounter{enumii}{1}
    \item This clause is void.
    \item As $\mathsf{fold}(t) \xto{\mu} t$, it suffices to show
      $\logrel^{n}_{\tau[\mu\alpha.\tau/\alpha]}(t_{1},s_{1})$,
      which holds by the assumption.
    \end{enumerate}
  \item \emph{Case $\mathsf{unfold}_{\tau}$.} We assume
    $\logrel^{n+1}_{\mu\alpha.\tau}(t_{1},s_{1})$.
    \begin{enumerate}[label=(\roman*)]
      \setcounter{enumii}{1}
    \item First, we have $t_{1} \to t_{1}'$, and thus $\mathsf{unfold}(t_{1}) \to
      \mathsf{unfold}(t_{1}')$. By
      $\logrel^{n+1}_{\mu\alpha.\tau}(t_{1},s_{1})$, $\exists s_{1}'.\,
      s_{1} \To s_{1}' \land \logrel^{n}_{\mu\alpha.\tau}(t_{1}',s_{1}')$.
      Hence, we conclude that  $\mathsf{unfold}(s_{1}) \To
      \mathsf{unfold}(s_{1}')$. It suffices to show that
      $\logrel^{n}_{\tau[\mu\alpha.\tau/\alpha]}(\mathsf{unfold}(t_{1}')
      ,\mathsf{unfold}(s_{1}'))$, which holds by the induction hypothesis.

      Second, we have $t_{1} \xto{\mu} t_{1}'$, and thus $\mathsf{unfold}(t_{1})
      \to t_{1}'$. By
      $\logrel^{n+1}_{\mu\alpha.\tau}(t_{1},s_{1})$, $\exists s_{1}'.\,
      s_{1} \xTo{\mu} s_{1}' \land \logrel^{n}_{\mu\alpha.\tau}(t_{1}',s_{1}')$.
      Hence, we conclude that  $\mathsf{unfold}(s_{1}) \To s_{1}'$ and we are
      done.
    \item This clause is void.
    \end{enumerate}
  \end{enumerate}

  We have just proved that $\logrel^{n}$ is a congruence for all $n \in
  \mathds{N}$. By the definition of $\logrel^{\omega}$
  as an intersection and the fact that congruences are closed under arbitrary
  intersections, we conclude that $\logrel^{\omega}$ is a congruence.

\section{Details for Section 4}

\subsection*{Relations in Categories}
We collect a few auxiliary results about relations in categories that we shall need in subsequent proofs. Throughout this section we work under the global \autoref{assumptions} on the category $\C$.

\begin{notation} Recall that given relations $R,S\monoto X\times X$ the composite relation $R\cdot S\monoto X\times X$ is
constructed in two steps:
\begin{enumerate}
\item Form the pullback of $\outr_{R}$ and $\outl_S$:
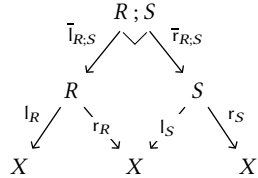
\[
\begin{tikzcd}[column sep=0.65em]
  & &
  R\smc S %
  \pullbackangle{-90}
  \ar{dl}[swap]{\ol{\outl}_{R\smc S}}
  \ar{dr}{\ol{\outr}_{R\smc S}}
  &  & \\
& R \ar{dl}[swap]{\outl_R} \ar{dr}[description]{\outr_R} & & S \ar{dl}[description]{\outl_{S}} \ar{dr}{\outr_{S}}  & \\
X & & X & & X 
\end{tikzcd}
\]
\item The relation $R\cdot S$ is given by the (strong epi, mono)-factorization
\[ \begin{tikzcd}[column sep=4em]
R\smc S \ar[shiftarr = {yshift=15}]{rr}{\langle \outl_R\cdot \ol{\outl}_{R\smc S}, \outr_{S}\cdot \ol{\outr}_{R\smc S} \rangle} \ar[two heads]{r}{e_{R\cdot S}} & R\cdot S \ar[tail]{r}{\langle \outl_{R\cdot S}, \outr_{R\cdot S} \rangle} & X\times X
\end{tikzcd}
\]
\end{enumerate}
\end{notation}

\begin{lemma}\label{lem:distributive}
If $\C$ is infinitary extensive, composition of relations distributes over joins: For every relation $S\monoto X\times X$ and every non-empty family $\R$ of relations over $X$,
\[ (\bigvee_{R\in \R} R)\cdot S = \bigvee_{R\in \R} R\cdot S \qquad\text{and}\qquad S\cdot (\bigvee_{R\in \R} R) = \bigvee_{R\in \R} S\cdot R.  \]
\end{lemma}

\begin{proof}
See \cite[Lem.~A.5]{utgms23_arxiv} for the case of binary joins, under the assumption of finitary extensivity (more generally, local distributivity). The proof for infinite joins is completely analogous. 
\end{proof}

Recall that a relation $R$ is \emph{reflexive} if $\Delta\leq R$, and \emph{transitive} if $R\cdot R\leq R$.
Since intersections of reflexive (transitive) relations are again reflexive (transitive), there exists for every relation $R\monoto X\times X$ a least reflexive and transitive relation $R^\star \monoto X\times X$ such that $R\leq R^\star$.

\begin{lemma}\label{lem:transhull}
Suppose that $\C$ is infinitary extensive.
\begin{enumerate}
\item\label{lem:transhull-1} For every relation $R\monoto X\times X$, we have $R^\star = \bigvee_{n\geq 0} R^n$, where $R^0=\Delta$ and $R^n = R\cdot \cdots \cdot R$ ($n$ factors) for $n\geq 1$.
\item\label{lem:transhull-2}  For every non-empty family $\R$ of relations over $X\in \C$, we have $(\bigvee_{R\in \R} R)^\star = \bigvee_{n\geq 0} \bigvee_{R_1,\ldots,R_n\in \R} R_1\cdot \cdots \cdot R_n$.
\end{enumerate}
\end{lemma}

\begin{proof}
\begin{enumerate}
\item The relation $S=\bigvee_{n\geq 0} R^n$ contains $R$ because $R=R^1\leq S$, it is reflexive because $\Delta=R^0\leq S$, and transitive because
\[ SS = (\bigvee_{m\geq 0} R^m)(\bigvee_{n\geq 0} R^n) = \bigvee_{m,n\geq 0} R^{m+n} = \bigvee_{n\geq 0} R^n = S; \]
the first step uses \autoref{lem:distributive}.
Therefore $R^\star\leq S$. Conversely, we have $R^n\leq R^\star$ for every $n\geq 0$. Indeed, for $n=0$ this follows from $R$ being reflexive, for $n=1$ use that $R\leq R^\star$, and for $n\geq 2$ use that $R^\star$ is transitive. Hence $S\leq R^\star$. 
\item We compute
\[ (\bigvee_{R\in \R} R)^\star = \bigvee_{n\geq 0} (\bigvee_{R\in \R} R)^n = \bigvee_{n\geq 0} \bigvee_{R_1,\ldots,R_n\in \R} R_1\cdots R_n \]
where the first step uses part \ref{lem:transhull-1} and the second one follows from \autoref{lem:distributive}. \qedhere
\end{enumerate}
\end{proof}

\begin{lemma}\label{lem:fimg-pres-joins-comp}
Let $f\colon X\to Y$ in $\C$.
\begin{enumerate}
\item\label{lem:fimg-pres-joins-comp-1} The map $f_\star\colon \RelCat[X]{\C}\to \RelCat[Y]{\C}$ preserves joins.
\item\label{lem:fimg-pres-joins-comp-2} For all relations $R,S\monoto X\times X$ one has $f_\star[R\cdot S]\leq f_\star [R]\cdot f_\star[S]$.
\end{enumerate}
\end{lemma}

\begin{proof}
Item~\ref{lem:fimg-pres-joins-comp-1} holds because $f_\star$ is a left adjoint (with right adjoint $f^\star$). For item~\ref{lem:fimg-pres-joins-comp-2}, see \cite[Lem.~A.7]{utgms23_arxiv}.
\end{proof}

\begin{notation}
  Recall that the canonical relation lifting $\ol{\Sigma}$ of an endofunctor $\Sigma\colon \C\to \C$ maps a relation
  $R\monoto X\times X$ to the relation
  $\ol{\Sigma}\seq \Sigma X \times \Sigma X$ obtained via the (strong
  epi, mono)-factorization
  \begin{equation}\label{eq:fac}
    \begin{tikzcd}[column sep=4em]
      \Sigma R
      \ar[shiftarr = {yshift=15}]{rr}{\langle \Sigma \outl_R, \Sigma \outr_R \rangle }
      \ar[two heads]{r}{e_{\ol{\Sigma} R}}
      &
      \ol{\Sigma} R \ar[tail]{r}{\langle \outl_{\ol{\Sigma}R}, \outr_{\ol{\Sigma}R} \rangle}
      &
      \Sigma X\times \Sigma X
    \end{tikzcd}
  \end{equation}
\end{notation}

\begin{lemma}\label{lem:barSigma-resp-comp}
Given an endofunctor $\Sigma\colon \C\to\C$ preserving strong epimorphisms and relations $R,S\monoto X\times X$, one has $\ol{\Sigma}(R\cdot S)\leq \ol{\Sigma} R \cdot \ol{\Sigma} S$.
\end{lemma}

\begin{proof}
Consider the pullbacks used in the definition of $R\cdot S$ and $\ol{\Sigma} R \cdot \ol{\Sigma} S$:
\begin{equation}\label{eq:comp}
\begin{tikzcd}[column sep=1em]
  & &
  R\smc S %
  \pullbackangle{-90}
  \ar{dl}[swap]{\ol{\outl}_{R\smc S}}
  \ar{dr}{\ol{\outr}_{R\smc S}}
  &  & \\
& R \ar{dl}[swap]{\outl_R} \ar{dr}[description]{\outr_R} & & S \ar{dl}[description]{\outl_{S}} \ar{dr}{\outr_{S}}  & \\
X & & X & & X 
\end{tikzcd}
\qquad\qquad
\begin{tikzcd}[column sep=0.65em]
  & &
  \ol{\Sigma} R\smc \ol{\Sigma} S %
  \pullbackangle{-90}
  \ar{dl}[swap]{\ol{\outl}_{\ol{\Sigma} R\smc \ol{\Sigma} S}}
  \ar{dr}{\ol{\outr}_{\ol{\Sigma} R\smc \ol{\Sigma} S}}
  &  & \\
& \ol{\Sigma} R \ar{dl}[swap]{\outl_{\ol{\Sigma} R}} \ar{dr}[description]{\outr_{\ol{\Sigma} R}} & & \ol{\Sigma} S \ar{dl}[description]{\outl_{\ol{\Sigma} S}} \ar{dr}{\outr_{\ol{\Sigma} S}}  & \\
\Sigma X & & \Sigma X & & \Sigma X 
\end{tikzcd}
\end{equation}

Now consider the following diagram:
\begin{equation}\label{eq:f}
\begin{tikzcd}[column sep = 40]
  \Sigma R
  \ar[two heads]{d}[swap]{e_{\ol{\Sigma}R}}
  \ar[shiftarr = {xshift = -20}]{dd}[swap]{\Sigma \outr_R}
  &
  \Sigma (R\smc S)
  \ar{r}{\Sigma \ol{\outr}_{R\smc S}}
  \ar{l}[swap]{\Sigma \ol{\outl}_{R\smc S}}
  \ar[dashed]{d}{f}
  &
  \Sigma S \ar[two heads]{d}{e_{\ol{\Sigma}S}}
  \ar[shiftarr = {xshift = 20}]{dd}{\Sigma \outl_S}
  \\
  \ol{\Sigma} R
  \ar{d}[swap]{\outr_{\ol\Sigma R}}
  &
  \ol{\Sigma} R \smc \ol{\Sigma} S
  \ar{l}[swap]{\ol{\outl}_{\ol{\Sigma} R;\ol{\Sigma} S}}
  \ar{r}{\ol{\outr}_{\ol{\Sigma} R;\ol{\Sigma} S}}
  &
  \ol{\Sigma} S
  \ar{d}{\outl_{\ol\Sigma S}}
  \\
  \Sigma X
  \ar[equals]{rr}
  &
  &
  \Sigma X
\end{tikzcd}
\end{equation}
Its outside is the commutative square in~\eqref{eq:comp} on the left
under the functor $\Sigma$. The left-hand part is the right-hand
component of~\eqref{eq:fac}; similarly, the right-hand part is the
left-hand component of~\eqref{eq:fac} instantiated for $S$ in lieu of
$R$. Thus, by the universal property of the pullback
in~\eqref{eq:comp} on the right, there exists a
unique $f\colon \Sigma(R\smc S)\to \ol{\Sigma} R\smc \ol{\Sigma} S$
such that the upper two squares commute.

Then the outside of the following diagram commutes, as we explain below.
\[
  \begin{tikzcd}[column sep=80, row sep = 40]
    {\Sigma}(R\smc S)
    \ar{r}{\id}
    \ar[two heads]{d}[swap]{\Sigma e_{R\cdot S}}
    \ar{rdd}[description]{\langle
      \Sigma(\outl_R \cdot \ol{\outl}_{R\smc S}),
      \Sigma(\outr_S \cdot \ol{\outr}_{R\smc S})
      \rangle}
    &
    {\Sigma}(R\smc S)
    \ar{r}{f}
    \ar{d}[swap]{\langle
      e_{\ol\Sigma R} \cdot \Sigma\ol\outl_{R\smc S},
      e_{\ol\Sigma S} \cdot \Sigma\ol\outr_{R\smc S},
      \rangle}
    &
    \ol{\Sigma} R \smc \ol{\Sigma} S
    \ar{dd}{e_{\ol{\Sigma}R \cdot \ol{\Sigma}S}}
    \ar{ldd}[description]{\langle
      \outl_{\ol\Sigma R} \cdot \ol\outl_{\ol\Sigma R \smc \ol\Sigma S},
      \outr_{\ol\Sigma S} \cdot \ol\outr_{\ol\Sigma R\smc \ol\Sigma S}
      \rangle}
    \ar{ld}[swap]{\langle
      \ol\outl_{\ol\Sigma R \smc \ol\Sigma S},
      \ol\outr_{\ol\Sigma R\smc \ol\Sigma S}
      \rangle}
    \\
    {\Sigma}(R\cdot S)
    \ar[two heads]{d}[swap]{e_{\ol{\Sigma}(R\cdot S)}}
    \ar{rd}[swap]{\langle \Sigma \outl_{R \cdot S}, \Sigma\outr_{R \cdot S}\rangle}
    &
    \ol\Sigma R \times \ol\Sigma S
    \ar{d}[description]{\outl_{\ol\Sigma R}\times \outr_{\ol\Sigma S} }
    &
    \\
    \ol{\Sigma}(R\cdot S)
    \ar{r}{\langle \outl_{\ol{\Sigma}(R\cdot S)}, \outr_{\ol{\Sigma}(R\cdot S)} \rangle}
    &
    \Sigma X \times \Sigma X
    &
    \ol{\Sigma} R \cdot \ol{\Sigma} S
    \ar{l}[swap]{\langle
      \outl_{\ol{\Sigma}R\cdot \ol{\Sigma} S},
      \outr_{\ol{\Sigma}R\cdot \ol{\Sigma} S}
      \rangle}
\end{tikzcd}
\]
The lower left-hand triangle commutes by~\eqref{eq:fac}, and the one
above it by the definition of $R \cdot S \monoto X \times X$ as the
image of the pair in the span given in~\eqref{eq:comp} on the
left. For the triangle above that (the upper left-hand one overall) consider the product components
separately: the left-hand one is $\Sigma\ol\outl_{R\smc S}$
precomposed with the left-hand component of~\eqref{eq:fac}; for the
right-hand component one has, similarly, $\Sigma\ol\outr_{R \smc S}$
precomposed with the right-hand component of~\eqref{eq:fac}
instantiated for $S$ in lieu of $R$. The upper right-hand triangle
commutes by using the two top squares in~\eqref{eq:f}, and the one below it is obvious. Finally,
the lower right-hand triangle commutes by the definition of $\ol\Sigma
R \cdot \ol\Sigma S \monoto \Sigma X \times \Sigma X$ as the image of
the outside span in~\eqref{eq:comp} on the right. 

Now note that $\Sigma e_{R\cdot S}$ at the left-hand edge above is a strong
epimorphism because $\Sigma$ preserves strong epimorphisms by
assumption. Hence we obtain a diagonal fill-in witnessing that
$\ol{\Sigma}(R\cdot S)\leq \ol{\Sigma} R \cdot \ol{\Sigma} S$, as
shown in the diagram below:
\[
\begin{tikzcd}[column sep=5em, row sep=3em, baseline = (B.base)]
{\Sigma}(R\smc S) \ar{rr}{f} \ar[two heads]{d}[swap]{\Sigma e_{R\cdot
    S}}
&&
\ol{\Sigma} R \smc \ol{\Sigma} S \ar{d}{e_{\ol{\Sigma}R \cdot
    \ol{\Sigma}S}}
\\
{\Sigma}(R\cdot S) \ar[two heads]{d}[swap]{e_{\ol{\Sigma}(R\cdot S)}}
&&
\ol{\Sigma} R \cdot \ol{\Sigma} S
\ar[>->]{d}{\langle \outl_{\ol{\Sigma}R\cdot \ol{\Sigma} S}, \outr_{\ol{\Sigma}R\cdot \ol{\Sigma} S} \rangle }\\
\ol{\Sigma}(R\cdot S) \ar[dashed]{urr} \ar{rr}{\langle
  \outl_{\ol{\Sigma}(R\cdot S)}, \outr_{\ol{\Sigma}(R\cdot S)}
  \rangle}  && |[alias = B]|\Sigma X \times \Sigma X 
\end{tikzcd}
\qedhere
\]
\end{proof}

\begin{lemma}\label{lem:cong-comp}
Suppose that $\Sigma\colon \C\to \C$ preserves strong epimorphisms. Then for every $\Sigma$-algebra $(A,a)$ the diagonal $\Delta\monoto A\times A$ is a congruence, and the composite $R\cdot S$ of congruences is a congruence. 
\end{lemma}

\begin{proof}
The diagonal is a congruence because
\[ a_\star\ol{\Sigma}\Delta = a_\star \Delta \leq \Delta. \]
Given congruences $R,S\monoto A\times A$, the composite $R\cdot S$ is a congruence since
\[ a_\star\ol{\Sigma}(R\cdot S) \leq a_\star(\ol{\Sigma}R \cdot\ol{\Sigma} S) \leq (a_\star\ol{\Sigma}R) \cdot (a_\star \ol{\Sigma} S) \leq R\cdot S, \]
where the first step uses \autoref{lem:barSigma-resp-comp} and the second step uses \autoref{lem:fimg-pres-joins-comp}.\ref{lem:fimg-pres-joins-comp-2}.
\end{proof}

\begin{lemma}\label{lem:barSigma-pres-directed-joins}
Suppose that $\C$ is infinitary extensive and has smooth monomorphisms, and that $\Sigma\colon \C\to \C$ preserves monomorphisms and directed colimits. Then for every $X\in \C$ the map $\ol{\Sigma}\colon \RelCat[X]{\C} \to \RelCat[\Sigma X]{\C}$ preserves directed joins.
\end{lemma}

\begin{proof}
Let $\R$ be a directed set of relations on $X$. By smoothness, its join $\bigvee_{R\in \R} R$ in $\RelCat[X]{\C}$ coincides with the colimit of the corresponding diagram in $\C$. Since $\Sigma$ preserves that colimit, it follows that $\ol{\Sigma}(\bigvee_{R\in \R} R) = \bigvee_{R\in \R} \ol{\Sigma} R$, as required.
\end{proof}

\subsection*{Proof of \autoref{prop:contextual-preorder}}
This is a consequence of the following observations:
\begin{enumerate}
\item $\Delta$ is an $O$-adequate relation, and composites of
  $O$-adequate relations are $O$-adequate, since $O$ is a preorder
  whence transitive.%
  \smnote{This also uses that relation composition is monotone, which
    we never stated, but we should; it should follow from \Cref{lem:distributive}.}
\item $\Delta$ is a congruence, and composites of congruences are congruences (\autoref{lem:cong-comp}).
\end{enumerate}
Since $\cprd$ is the greatest $O$-adequate relation, it follows from (1) and (2) that both $\Delta$ and $\cprd\cdot \cprd$ are contained in $\cprd$, hence $\cprd$ is a preorder. 

\subsection*{Proof of \autoref{lem:union-transhull-cong}}
We need to prove $a_\star\ol{\Sigma}( \bigvee_{R\in \R} R )^\star \leq (\bigvee_{R\in \R} R)^\star$, which follows from the computation
\begin{align*}
a_\star\ol{\Sigma}( \bigvee_{R\in \R} R^\star ) &= a_\star\ol{\Sigma} \big( \bigvee_{n\geq 0} \bigvee_{R_1,\ldots,R_n\in \R} R_1\cdot \cdots \cdot R_n \big) & \text{\autoref{lem:transhull}\ref{lem:transhull-2}} \\
&= a_\star \big( \bigvee_{n\geq 0} \bigvee_{R_1,\ldots,R_n\in \R} \ol{\Sigma}(R_1\cdot \cdots \cdot R_n) \big) & \text{see below} \\
&= \bigvee_{n\geq 0} \bigvee_{R_1,\ldots,R_n\in \R} a_\star(\ol{\Sigma}(R_1\cdot \cdots \cdot R_n)) & \text{\autoref{lem:fimg-pres-joins-comp}\ref{lem:fimg-pres-joins-comp-1}} \\
&\leq \bigvee_{n\geq 0} \bigvee_{R_1,\ldots,R_n\in \R} a_\star(\ol{\Sigma} R_1\cdot \cdots \cdot \ol{\Sigma}R_n) & \text{\autoref{lem:barSigma-resp-comp}} \\
&\leq \bigvee_{n\geq 0} \bigvee_{R_1,\ldots,R_n\in \R} (a_\star\ol{\Sigma} R_1)\cdot \cdots \cdot (a_\star\ol{\Sigma}R_n) & \text{\autoref{lem:fimg-pres-joins-comp}\ref{lem:fimg-pres-joins-comp-2}} \\
&\leq \bigvee_{n\geq 0} \bigvee_{R_1,\ldots,R_n\in \R} R_1\cdot \cdots \cdot R_n & \text{$R\in \R$ congruence} \\
&= (\bigvee_{R\in \R} R)^\star & \text{\autoref{lem:transhull}\ref{lem:transhull-2}} 
\end{align*}
In the second step, we use that the join $\bigvee_{n\geq 0}
\bigvee_{R_1,\ldots,R_n\in \R} R_1\cdot \cdots \cdot R_n$ is directed:
given $R_1,\ldots,R_n, R_1',\ldots,R_m'\in \R$, the relations
$R_1\cdot\cdots \cdot R_n$ and $R_1'\cdot\cdots\cdot R'_m$ have the
upper bound $R_1\cdot\cdots\cdot R_n \cdot R_1'\cdot\cdots\cdot R_m'$
because all relations in $\R$ are reflexive. Hence, by \autoref{lem:barSigma-pres-directed-joins}, the functor $\ol{\Sigma}$ preserves the join.

\section{Details for Section 5}\label{sec:fpc-extended}
We provide a more detailed account on how to apply the higher-order abstract
GSOS framework to (nondeterministic) \fpc,
which is a typed, call-by-name
$\lambda$-calculus with sums, products, conditionals, recursive types and choice. The type
expressions are given by the grammar
\begin{equation*}
  \tau_1,\tau_2,\tau,\ldots \Coloneqq \alpha \mid\tau_1\boxplus\tau_2\mid\tau_1\boxtimes\tau_2\mid \arty{\tau_1}{\tau_2}\mid \mu\alpha.\,\tau,
\end{equation*}
where $\alpha$ ranges over a fixed countably infinite set of {type variables}.
We write $\Ty$ the set of closed type expressions modulo $\alpha$-equivalence.
The terms of \fpc are constructed under the following typing rules:
\begin{gather*}
  \inference{\Gamma \vdash t \c \arty{\tau_{1}}{\tau_{2}} & \Gamma \vdash s \c
    \tau_{1}}{\Gamma \vdash t \app{}{} s \c \tau_{2}}
  \qquad
  \inference{\Gamma, x \c \tau_{1} \vdash t \c \tau_{2}}
  {\Gamma \vdash \mathsf{lam}\,x\c\tau_{1}.\, t\c \arty{\tau_{1}}{\tau_{2}}}
  \qquad
  \inference{\Gamma \vdash t \c \tau[\mu\alpha.\tau/\alpha]}{\Gamma \vdash
    \mathsf{fold}(t) \c \mu\alpha.\tau}
  \qquad
  \inference{\Gamma \vdash t \c \mu\alpha.\tau}{\Gamma \vdash
    \mathsf{unfold}(t) \c \tau[\mu\alpha.\tau/\alpha]}
  \\[1ex]
  \inference{\Gamma \vdash t \c \tau_{1} & \Gamma \vdash s \c \tau_{2}}
  {\Gamma \vdash \mathsf{pair}(t,s) \c \tau_{1} \boxtimes \tau_{2}}
  \qquad
  \inference{\Gamma \vdash t \c \tau_{1} \boxtimes \tau_{2}}
  {\Gamma \vdash \mathsf{fst}(t) \c \tau_{1}}
  \qquad
  \inference{\Gamma \vdash t \c \tau_{1} \boxtimes \tau_{2}}
  {\Gamma \vdash \mathsf{snd}(t) \c \tau_{2}}
  \qquad
  \inference{\Gamma \vdash t \c \tau & \Gamma \vdash s \c \tau}{\Gamma \vdash t
    \oplus s \c \tau}
  \\[1ex]
  \inference{\Gamma \vdash t \c \tau_{1}}{\Gamma \vdash \inl(t) \c \tau_{1}
    \boxplus \tau_{2}}
  \qquad
  \inference{\Gamma \vdash t \c \tau_{2}}{\Gamma \vdash \inr(t) \c \tau_{1}
    \boxplus \tau_{2}}
  \qquad
  \inference{\Gamma \vdash t \c \tau_{1} \boxplus \tau_{2}
    & \Gamma \vdash s \c \arty{\tau_{1}}{\tau_{3}}
    & \Gamma \vdash r \c \arty{\tau_{2}}{\tau_{3}}}
  {\Gamma \vdash \mathsf{case}(t,s,r) \c \tau_{3}}
\end{gather*}
The operational semantics work in the expected way:
\begin{gather*}
  \inference{t\to t'}{t \app{\tau_{1}}{\tau_{2}} s\to t'
    \app{\tau_{1}}{\tau_{2}} s}
  \qquad
  \inference{}{(\mathsf{lam}\,x\c\tau.\,t) \app{}{} s \to t[x/s]} \qquad
  \inference{t \to t'}{\mathsf{unfold}(t) \to \mathsf{unfold}(t')} \qquad
  \inference{}{\mathsf{unfold}(\mathsf{fold}(t)) \to t}
  \qquad
  \\[1ex]
  \inference{t \to t'}{\mathsf{fst}(t) \to \mathsf{fst}(t')}
  \qquad
  \inference{t \to t'}{\mathsf{snd}(t) \to \mathsf{snd}(t')}
  \qquad
  \inference{}{\mathsf{fst}(\mathsf{pair}(t,s)) \to t}
  \qquad
  \inference{}{\mathsf{snd}(\mathsf{pair}(t,s)) \to s}
  \\[1ex]
  \inference{}{t \oplus s \to t}
  \qquad
  \inference{}{t \oplus s \to s}
  \qquad
  \inference{t\to t'}{\mathsf{case}(t,s,r)\to
    \mathsf{case}(t',s,r)} \qquad
  \inference{}{\mathsf{case}(\inl(t),s,r)\to
    s\app{}{}t} \qquad
  \inference{}{\mathsf{case}(\inr(t),s,r)\to
    r\app{}{}t}
\end{gather*}

\paragraph{Categorical modelling}
Implementing \textbf{FPC} in the style of higher-order abstract
GSOS follows ideas from earlier work~\cite{gmstu23, UrbatTsampasEtAl23, gmstu24}, this
time applied to a typed setting with
recursive types. Let $\fset/{\Tyl}$ be the
slice category of the category of finite cardinals $\fset$ over $\Ty$. The
objects of $\fset/{\Tyl}$ are typed cartesian contexts, i.e. vectors of types
$\Gamma \c n \to \Tyl$; morphisms are type-respecting renamings:
\[
  \fset/{\Tyl}(\Gamma \c n \to \Tyl, \Delta \c m \to
  \Tyl) = \{r \c n \to m \mid \Delta(r(n)) = \Gamma(n)\}
\]
We write $|\Gamma|$ for the domain of a context $\Gamma$.
Each type $\tau \in \Tyl$ induces the single-variable context $\check \tau\c
1 \to \Tyl$; coproducts in~$\fset/{\Tyl}$ are formed by copairing:
\[
  (\Gamma_{1} \c |\Gamma_{1}| \to \Tyl) + (\Gamma_{2} \c |\Gamma_{1}| \to \Tyl)
  = [\Gamma_{1}, \Gamma_{2}] \c |\Gamma_{1}| + |\Gamma_{2}| \to \Tyl.
\]
The fundamental operation of \emph{context extension} $(-
+ \check\tau) \c \fset/{\Tyl} \to \fset/{\Tyl}$ extends a variable context with
a new variable $x \c \tau$, i.e. $\Gamma \xmapsto{(-
  + \check\tau)} \Gamma, x \c \tau$ in a type-theoretic notation.

Our higher-order GSOS law for \textbf{FPC} lives in the category
$(\Set^{\fset/{\Ty}})^{\Ty}$ of type-indexed covariant presheaves over
$\fset/{\Tyl}$. Its objects are families of sets
indexed by contexts $\Gamma \in \fset/{\Tyl}$ and types $\tau
\in \Ty$ that respect context renamings. Two fundamental examples of objects in
$(\Set^{\fset/{\Ty}})^{\Ty}$ are the presheaf of \emph{variables} $V$,
\begin{equation*}
  V_{\tau}(\Gamma)=\{x\in |\Gamma| \mid \Gamma(x)=\tau\} \quad \text{and} \quad
  V_{\tau}(\Gamma)(r) = r,
\end{equation*}
and
the family $\Lambda$ of
well-typed, $\alpha$-equivalent terms in \textbf{FPC}, with
\begin{equation*}
  t \in \Lambda_{\tau}(\Gamma) \iff \Gamma \vdash t \c \tau.
\end{equation*}
The syntax functor $\Sigma$ of our higher-order GSOS law is the endofunctor
corresponding to the $\Ty$-sorted \emph{binding signature} of \textbf{FPC}:
\begin{equation}
  \label{eq:sigmalamapp}
    \begin{aligned}
    & \Sigma_{\tau} X =
      V_{\tau} + X_{\tau} \times X_{\tau}
      + \Sigma^{1}_{\tau}X + \Sigma^{2}_{\tau}X + \Sigma^{3}_{\tau}X + \Sigma^{4}_{\tau}X, \\
    & \Sigma^{1}_{\tau} X =
      \coprod_{\tau' \in \Tyl}(X_{\arty{\tau'}{\tau}} \times X_{\tau'})
      + X_{\tau' \boxtimes \tau} + X_{\tau \boxtimes \tau'},\\
    & \Sigma^{2}_{\tau_{1} \boxplus \tau_{2}} X= X_{\tau_{1}} + X_{\tau_{2}},
      \qquad
      \Sigma^{2}_{\tau_{1} \boxtimes \tau_{2}} X = X_{\tau_{1}} \times X_{\tau_{2}},
    \\
    & \Sigma^{2}_{\mu\alpha.\,\tau} X= X_{\tau[\mu\alpha.\,\tau/\alpha]},
      \qquad
   \Sigma^{2}_{\arty{\tau_{1}}{\tau_{2}}}X= X_{\tau_{2}} \comp (- + \check\tau_{1}),
    \\
    & \Sigma^{3}_{\tau} X = \coprod_{\tau_{1} \in \Tyl}
      \coprod_{\tau_{2} \in \Tyl}X_{\tau_{1} \boxplus \tau_{2}} \times X_{\arty{\tau_{1}}{\tau}}
      \times X_{\arty{\tau_{2}}{\tau}},
    \\
    & \Sigma^{4}_{\tau} X =
      \coprod_{\sigma \c \tau = \sigma[\mu \alpha.\sigma/\alpha]}X_{\mu\alpha.\sigma}.
  \end{aligned}
\end{equation}
The family $\Lambda$ of $\alpha$-equivalent \textbf{FPC}-terms is isomorphic to the
initial algebra $\mS$ of $\Sigma$.
Let $\Pow_{\star}$ be the pointwise
powerset functor $\Pow_{\star}(X) = \Pow \comp X$. The
behaviour bifunctor $B$ is given by
\begin{equation}
  \label{eq:bapp}
  \begin{aligned}
    B(X,Y)
    &\;= \llangle X,Y \rrangle \times \Pow_{\star}(Y + D(X,Y)),\\
    D_{\arty{\tau_{1}}{\tau_{2}}}(X,Y)
    &\;= Y_{\tau_{2}}^{X_{\tau_{1}}},
    &\hspace{-5.5em}
      D_{\tau_1\boxplus\tau_2}(X,Y)
    &\;= Y_{\tau_1}+Y_{\tau_2}, \\
    D_{\mu\alpha.\,\tau}(X,Y)
    &\;= Y_{\tau[\mu\alpha.\tau/\alpha]},
    &\hspace{-5.5em}
      D_{\tau_1\boxtimes\tau_2}(X,Y)
    &\;= Y_{\tau_1} \times Y_{\tau_2}.
  \end{aligned}
\end{equation}
where $Y_{\tau_{2}}^{X_{\tau_{1}}}$ is the exponential in $\Set^{\fset/{\Ty}}$
and the bifunctor $\llangle \argument,\argument \rrangle$ is
\[
  \llangle X,Y \rrangle_{\tau}(\Gamma) = \Set^{\fset/{\Tyl}}\Bigl(\prod_{x \in
    |\Gamma|}X_{\Gamma(x)}, Y_{\tau}\Bigr).
\]
The bifunctor $\llangle \argument,\argument \rrangle$ models
\emph{simultaneous substitution}. For example, there is a morphism $\gamma_{0}
\c \Lambda \to \llangle \Lambda,\Lambda \rrangle$ mapping terms of \textbf{FPC}
to their substitution structure: given $t\in \Lambda_\tau(\Gamma)$, the natural transformation $(\gamma_0)_\tau(t)\colon \prod_{x \in
  |\Gamma|}\Lambda_{\Gamma(x)}\to \Lambda_{\tau}$ is given at component $\Delta\in \fset/\Ty$ by
\[ \vec{u}\in \prod_{x \in
    |\Gamma|}\Lambda_{\Gamma(x)}(\Delta) \quad\mapsto \quad t[\vec{u}]\in \Lambda_\tau(\Delta),\] i.e.\ the simultaneous substitution of $\vec{u}$ for the variables of $t$. 

\paragraph*{Higher-order GSOS law.} The $V$-pointed pointed higher-order GSOS law 
\begin{equation}
  \label{eq:lawapp}
  \rho_{\nu \c V \to X,Y} \c \Sigma(X \times B(X,Y)) \to
  B(X,\Sigma^{\star}(X + Y))
\end{equation}
for \fpc is the pairing of the two components
\[
  \begin{aligned}
  & \rho^{1}_{\nu \c V \to X,Y} \c \Sigma(X \times \llangle X,Y \rrangle
  \times \Pow_{\star}(Y +D(X,Y))) \to
    \llangle X, \Sigma^{\star}(X + Y) \rrangle \\
  & \rho^{2}_{\nu \c V \to X,Y} \c \Sigma(X \times \llangle X,Y \rrangle
    \times \Pow_{\star}(Y +D(X,Y))) \to
    \Pow_{\star}(\Sigma^{\star}(X + Y) + D(X,\Sigma^{\star}(X + Y))). \\
  \end{aligned}
\]
Component $\rho^{1}$ is produced by the canonical \emph{$V$-pointed
  strength}~\cite[p. 6]{DBLP:conf/csl/FioreH10} (see also
\cite{DBLP:conf/lics/Fiore08}) of endofunctor $\Sigma \c
(\Set^{\fset/{\Ty}})^{\Ty} \to (\Set^{\fset/{\Ty}})^{\Ty}$, and makes use of the
adjunction $\argument \otimes Y \dashv \llangle Y, \argument \rrangle$, where
$\argument \otimes \argument$ is the so-called \emph{substitution tensor}, given by
\begin{equation*}
  (X \otimes_{\tau} Y) (\Gamma) = \int^{\Delta \in \Set^{\fset/{\Tyl}}} X_{\tau}(\Delta)
  \times \prod_{i \in |\Delta|}Y_{\Delta(i)}(\Gamma).
\end{equation*}
The unit of of the tensor is the presheaf of variables $V$. A $V$-pointed
strength for an endofunctor $F$ is given by a family of maps
\[
  \mathrm{str}^{F}_{X,\nu \c V \to Y} \c F X \otimes Y
  \to F (X \otimes Y),
\]
natural in X and $Y/V$. Let $\mathrm{str}_{X,\nu \c V \to Y} \c \Sigma X \otimes
Y \to \Sigma (X \otimes Y)$ be the canonical $V$-pointed strength of $\Sigma$.
The component $\rho^{1}$ is given in terms of $\mathrm{str}$ as the adjoint
transpose of the morphism
\[
  \begin{tikzcd}
  \Sigma(X \times \llangle X,Y \rrangle
  \times \Pow_{\star}(Y +D(X,Y))) \otimes X
  \ar{rr}{\Sigma(\outr_{2})\otimes \id}
  & & \Sigma(\llangle X,Y \rrangle) \otimes X
  \ar{rr}{\mathrm{str}_{\llangle X,Y \rrangle,X}}
  & & \Sigma(\llangle X,Y \rrangle \otimes X)
  \ar{rr}{\Sigma(\mathrm{eval_{X,Y}})}
  & & \Sigma X
  \ar{r}{\theta \comp \Sigma \inl}
  & \Sigma^{\star}(X+Y),
  \end{tikzcd}
\]
where $\theta_{X} \c \Sigma \to \Sigma^{\star}$ is the embedding of $\Sigma$ to
its free monad. Note that that applying strength $\mathrm{str}_{\llangle X,Y
  \rrangle,X}$ is correct, as $X$ is $V$-pointed. Our setup ensures that in the canonical operational model
\begin{equation*}
  \langle \gamma_{0},\gamma \rangle \c \Lambda \to \llangle \Lambda,\Lambda \rrangle \times
  \Pow_{\star}(Y + D(\Lambda,\Lambda)),
\end{equation*}
map $\gamma_{0}$ yields the substitution structure of terms in $\Lambda$ (see
e.g. \cite[Prop. 5.9]{gmstu23}).
Family $\rho^{2}$ is responsible for the computational behaviour of terms and
is given below. We keep the injection maps $\inl,\inr$ on the codomain implicit,
and annotate terms by types to improve readability.

\begin{equation*}
  \begin{aligned}
    &\rho^{2}_{\nu \c V \to X,Y} \c
    &\hspace{0em} \Sigma(X \times \llangle X,Y \rrangle
      \times \Pow_{\star}(Y +D(X,Y)))
    & \to \Pow_{\star}(\Sigma^{\star}(X + Y) + D(X,\Sigma^{\star}(X + Y))) \\*
    &\rho^{2}_{\nu \c V \to X,Y}\makebox[0pt][l]{$(tr) = \texttt{case}~tr~\texttt{of}$} \\
    &\quad\makebox[0em][l]{$\mathsf{var}~x : \tau$}
    && \mapsto~ \varnothing \\
    &\quad\makebox[0em][l]{$\mathsf{lam}_{\tau_{1},\tau_{2}} x\c
      \tau_{1}.\,(t ,(v \in \llangle X,Y\rrangle_{\tau_{2}}(\Gamma + \check\tau_{1})),N)$}
    && \mapsto~ \{\lambda (e \in X_{\tau_{1}}(\Gamma)). v(x_{1},\dots,x_{\Gamma(n)},e)\} \\
    &\quad\makebox[0em][l]{$t \oplus_{\tau} s$}
    && \mapsto~ \{t,s\} \\
    &\quad\makebox[0em][l]{$\inl_{\tau_1,\tau_2}(t,v,N)$}
    && \mapsto~ \{t\} \\
    &\quad\makebox[0em][l]{$\inr_{\tau_1,\tau_2}(t,v,N)$}
    && \mapsto~ \{t\} \\
    &\quad\makebox[5em][l]{$\mathsf{case}_{\tau_1,\tau_2,\tau_3}((t,v,N),(s,u,U),(r,w,W))$} 		
    && \mapsto~ \Biggl(t_{f} = \begin{cases}
                  \makebox[10em][l]{$\mathsf{case}_{\tau_1,\tau_2,\tau_3}(f,s,r)$} \text{\quad if $f \in Y_{\tau_1\boxplus\tau_2}$} \\
                  \makebox[10em][l]{$\mathsf{app}_{\tau_{1},\tau_{3}}(s,t')$} \text{\quad if $f=\inl_{\tau_{1},\tau_{2}}(t')$} \\
                  \makebox[10em][l]{$\mathsf{app}_{\tau_{2},\tau_{3}}(r,t')$} \text{\quad if $f=\inr_{\tau_{1},\tau_{2}}(t')$}
                               \end{cases} \Biggr)_{f \in N} \\
    &\quad\makebox[0em][l]{$\fst_{\tau_1,\tau_2}(t,v,N)$}
    && \mapsto~ \Biggl(\begin{cases}
      \makebox[10em][l]{$\fst_{\tau_1,\tau_2}(f)$} \text{\quad if $f\in Y_{\tau_1\boxtimes\tau_2}$} \\
      \makebox[10em][l]{$\outl(f)$} \text{\quad if $f\in Y_{\tau_1}\times Y_{\tau_2}$}
    \end{cases}\Biggr)_{f \in N} \\
    &\quad\makebox[0em][l]{$\snd_{\tau_1,\tau_2}(t,v,N)$}
    && \mapsto~\Biggl(\begin{cases}
      \makebox[10em][l]{$\snd_{\tau_1,\tau_2}(f)$} \text{\quad if $f\in Y_{\tau_1\boxtimes\tau_2}$} \\
      \makebox[10em][l]{$\outr_{2}(f)$} \text{\quad if $f\in Y_{\tau_1}\times Y_{\tau_2}$}
    \end{cases}\Biggr)_{f \in N} \\
    &\quad\makebox[0em][l]{$\mathsf{pair}_{\tau_1,\tau_2}((t,v,N),(s,u,U))$} 
    && \mapsto~ \{(t,s)\}\\
    &\quad\makebox[0em][l]{$\mathsf{app}_{\tau_1,\tau_2}((t,v,N),(s,u,U))$}
    && \mapsto~ \Biggl(t_{f} = \begin{cases}
      \makebox[10em][l]{$f(s)$} \text{\quad if $f \in Y_{\tau_{2}}^{X_{\tau_{1}}}$} \\
      \makebox[10em][l]{$\mathsf{app}_{\tau_1,\tau_2}(f, s)$} \text{\quad if $f \in Y_{\arty{\tau_{1}}{\tau_{2}}}$}
    \end{cases}\Biggr)_{f \in N} \\
    &\quad\makebox[0em][l]{$\mathsf{fold}_\tau(t,v,N)$}
    && \mapsto~ \{t \in X_{\tau[\mu\alpha.\tau/\alpha]}\} \\
    &\quad\makebox[0em][l]{$\mathsf{unfold}_\tau(t,v,N)$}
    && \mapsto~ \Biggl(t_{f} = \begin{cases}
                  \makebox[10em][l]{$\mathsf{unfold}_\tau(f)$} \text{\quad if $f \in Y_{\mu \alpha.\,\tau}$} \\
                  \makebox[10em][l]{$f$} \text{\quad if $f \in Y_{\tau[\mu\alpha.\,\tau/\alpha]}$}
                               \end{cases}\Biggr)_{f \in N}
\end{aligned}
\end{equation*}

The operational semantics of \fpc involves rules such as
\[
  \inference{}{\mathsf{unfold}(\mathsf{fold}(t)) \to t},
  \qquad
  \inference{}{\mathsf{case}(\inl(t),s,r)\to s\app{}{}t}
  \text{\quad and \quad}
  \inference{}{\mathsf{case}(\inr(t),s,r)\to
    r\app{}{}t}
\]
which, strictly speaking, are not GSOS: the rules seemingly ``pattern match''
the shape of subterms to decide the next transition. In such instances, the
respective subterms (e.g. $\mathsf{fold}(t)$, $\inl(t)$ and $\inr(t)$) are
understood as \emph{values}. The way such rules are modelled is by identifying
these terms as values in the behaviour functor, by adding the respective
constructor. For example, with $B(X,Y) = \llangle X,Y \rrangle \times
\Pow_{\star}(Y + D(X,Y))$ and $D_{\tau_1\boxplus\tau_2}(X,Y) =
Y_{\tau_1}+Y_{\tau_2}$, we identify the injections $\inl$ and $\inr$ as value
constructors. In addition, the behaviour of $\lambda$-abstractions is that of a
function on terms of the suitable type.

\paragraph*{Operational model.} The higher-order GSOS law $\rho$ for \fpc induces the operational model
\begin{equation*}
  \langle \gamma_{0},\gamma \rangle \c \Lambda \to \llangle \Lambda,\Lambda \rrangle \times
  \Pow_{\star}(Y + D(\Lambda,\Lambda)).
\end{equation*}
The component $\gamma_0$ is the substitution map described earlier, and the component $\gamma$ models the call-by-name
semantics of \textbf{FPC}. For every $\Gamma \vdash t \c \tau$, we have that
$\gamma_\tau(\Gamma)(t)$ is the smallest subset of
$B(\Lambda,\Lambda)_\tau(\Gamma)$ satisfying the following:
\begin{flalign*}
  \gamma_\tau(\Gamma)(t)\ni\;&t' \text{\qquad if $t \to t'$, for $\Gamma \vdash t'\c\tau$,}\\*
    \gamma_{\arty{\tau_{1}}{\tau_{2}}}(\Gamma)(t) \ni\;& \lambda e.\,s[e/x]
   \text{\qquad if $t = \mathsf{lam}\,x\c\tau_{1}.s$, for $\Gamma, \tau_{1} \vdash s
                                                         \c \tau_{2} $,}\\*
  \gamma_{\tau_1\boxplus\tau_2}(\Gamma)(t)\ni\;&t'
                                                 \text{\qquad if $t = \inl(t')$,} \\*
  \gamma_{\tau_1\boxplus\tau_2}(\Gamma)(t)\ni\;&t'
                                                 \text{\qquad if $t =
                                                 \inr(t')$,} \\*
  \gamma_{\tau_1\boxtimes\tau_2}(\Gamma)(t)\ni\;&(t_{1},s_{1})
                                                  \text{\qquad if $t = \mathsf{pair}(t_{1},s_{1})$,} \\*
    \gamma_{\mu \alpha.\,\tau}(\Gamma)(t) \ni\;& t' \text{\qquad if $t = \mathsf{fold}(t')$}
\end{flalign*}
The weak operational model associated to \textbf{FPC} is $\langle
\gamma_{0},\wt{\gamma} \rangle \c \Lambda \to
\llangle \Lambda,\Lambda \rrangle \times \Pow_{\star}(Y + D(\Lambda,\Lambda))$,
where $\wt{\gamma}$ models $\To$. Alternatively, we could consider a
transition system in the style of $\myarr{}$ from \Cref{not:myarr} to later
obtain a slightly more expressive logical relation.

\paragraph{A free logical relation for \textbf{FPC}}
Having set-up our higher-order GSOS law, we move on to the logical relation.
First we have to specify our relation lifting $\ol{B}$ on the behaviour functor $B$
\eqref{eq:bapp}. We let
\begin{equation}
  \label{eq:liftingapp}
  \ol{B}(Q,R) = \ol{\llangle Q,R \rrangle} \times \ol{\Pow_{\star}}(R + \ol{D}(Q,R)),
\end{equation}
where $Q,R \mapsto \ol{\llangle Q,R \rrangle}$ is the canonical lifting for the
bifunctor $\llangle \argument, \argument \rrangle$, $\ol{\Pow_{\star}}$ is the left-to-right
Egli-Milner relation lifting and $\ol{D}$ the canonical relation lifting of $D$.
Applying \Cref{def:stepindexed} on the two coalgebras
\[
  \langle \gamma_{0},\gamma \rangle ,
  \langle \gamma_{0},\wt{\gamma} \rangle \c \Lambda \to
  \llangle \Lambda,\Lambda \rrangle \times \Pow_{\star}(Y + D(\Lambda,\Lambda))
\]
yields the following step-indexed logical relation. Note that in this
nondeterministic setting \Cref{def:stepindexed} does not terminate at $\omega$ steps.

\begin{definition} The logical relation $\logrel$ for \textbf{FPC} is the family
  of relations $(\logrel^\alpha \monoto \Lambda\times \Lambda)_{\alpha\leq 2^{\omega}}$ defined inductively by
\begin{align*}
 \logrel^{0}_{\tau}(\Gamma) =&\; \top_{\tau}(\Gamma)  = \{(t,s) \mid \Gamma \vdash t,s \c \tau\}&\\
  \kern.5em\logrel^{\alpha+1}_{\tau} =&\; \logrel^{\alpha}_{\tau}
     \cap \mathcal{S}_{\tau}(\logrel^{\alpha},\logrel^{\alpha})
      \cap \mathcal{E}_{\tau}(\logrel^{\alpha})
      \cap \mathcal{V}_{\tau}(\logrel^{\alpha},\logrel^{\alpha})\\
     \logrel^{\alpha}_\tau(\Gamma) =&\; \bigcap_{\beta < \alpha} \logrel^{\alpha}_\tau(\Gamma)
     \quad \text{for limit ordinals $\alpha$},
\end{align*}
where $\mathcal{S}$, $\mathcal{E}$,
$\mathcal{V}$ are the maps on $\Rel_{\Lambda}$ given by
  \begin{align*}
    \mathcal{S}_{\tau}(\Gamma)(Q,R)
      & = \{(t,s) \mid \text{for all $\Delta$ and $Q_{\Gamma(x)}(\Delta)(u_x,v_x)$ ($x\in \under{\Gamma}$)}, \text{ one has $R_{\tau}(\Delta)(t[\vec{u}],s[\vec{v}])$} \}, \\
    \mathcal{E}_{\tau}(\Gamma)(R)
      & = \{(t,s) \mid \text{if $t \to t'$ then $\exists s'.\,s \To s'
      \land R_{\tau}(\Gamma)(t',s')$}\}, \\
    \mathcal{V}_{\tau_{1} \boxplus \tau_{2}}(\Gamma)(Q,R)
      & = \{(t,s)
        \mid \text{if $t = \inl_{\tau_{1},\tau_{2}}(t')$ then }
        \exists s'.\,s \To \inl_{\tau_{1},\tau_{2}}(s')
        \land R_{\tau_{1}}(\Gamma)(t',s')\} \;{\cup} \\ & \phantom{=}\;\; \{(t,s)
        \mid \text{if $t = \inr_{\tau_{1},\tau_{2}}(t')$ then }
        \exists s'.\,s \To \inr_{\tau_{1},\tau_{2}}(s')
        \land R_{\tau_{2}}(\Gamma)(t',s')\}, \\
    \mathcal{V}_{\tau_{1} \boxtimes \tau_{2}}(\Gamma)(Q,R)
      & = \{(t,s) \mid
        \text{if $t = \mathsf{pair}_{\tau_{1},\tau_{2}}(t_{1},t_{2})$ then }
        \exists s_{1},s_{2}.\,s \To \mathsf{pair}_{\tau_{1},\tau_{2}}(s_{1},s_{2})
        \land R_{\tau_{1}}(\Gamma)(t_{1},s_{1}) \land R_{\tau_{2}}(\Gamma)(t_{2},s_{2})\}, \\
    \mathcal{V}_{\mu\alpha.\tau}(\Gamma)(Q,R)
      & = \{(t,s) \mid
        \text{if $t = \mathsf{fold}_{\tau}(t')$ then }
        \exists s'.\,s \To \mathsf{fold}_{\tau}(s')
        \land R_{\tau[\mu\alpha.\tau/\alpha]}(\Gamma)(t',s')\}, \\
    \mathcal{V}_{\arty{\tau_{1}}{\tau_{2}}}(\Gamma)(Q,R)
      & = \{(t,s) \mid
        \text{for all $Q_{\tau_{1}}(\Gamma)(e,e'),$} \text{ if $t=\lambda x.t'$ then $\exists s'.\,s \To \lambda x.s'\,\wedge\,R_{\tau_2}(\Gamma)(t'[e/x], s'[e'/x])$} \}.
  \end{align*}
\end{definition}

Note that termination after $2^\omega$ steps is guaranteed since $|\Rel_\Lambda|\leq 2^\omega$.
\begin{remark}
  Applying \Cref{def:stepindexed} in \textbf{FPC} produces logical relations
  that are closed under arbitrary substitutions by definition, whereas logical
  relations in the literature are often defined on closed terms, then extended
  to open terms in a subsequent step. The two constructions can be shown to be
  equivalent (see e.g. \cite[E.10]{utgms23_arxiv} for a similar argument).
\end{remark}

To prove congruence and reflexivity of $\logrel^{2^{\omega}}$, we first instantiate
the data to \Cref{asm-cong}, in a manner that is largely similar to
\Cref{ex:mutcl-asm}. In detail, we pick
\begin{enumerate}
\item the signature functor $\Sigma$ \eqref{eq:sigmalamapp},
\item the behaviour functor $B(X,Y) = \llangle X,Y \rrangle \times \Pow_{\star}(Y +
  D(X,Y))$~\eqref{eq:bapp} and its lifting $\overline{B}$~\eqref{eq:liftingapp},
\item the higher-order GSOS law $\rho = \langle \rho^{1},\rho^{2} \rangle$~\eqref{eq:lawapp}.
\end{enumerate}
Let us check that the higher-order GSOS law $\rho$ lifts to a
$\Delta_{V}$-pointed higher-order GSOS law
\[
  \ol{\rho}_{\nu \c \Delta_{V} \to Q,R} \c \ol{\Sigma}(Q \times \ol{B}(Q,R)) \to
 \ol{B}(X,\ol{\Sigma}^{\star}(Q + R)).
\]
It suffices to show that the following diagram commutes for each $Q \in
  \Delta_{V}/\RelCat{(\Set^{\fset/{\Ty}})^{\Ty}}$ and $R \in
  \RelCat{(\Set^{\fset/{\Ty}})^{\Ty}}$:
  \begin{equation*}
    \begin{tikzcd}[column sep=5em]
      \ol{\Sigma}(Q \times \ol{B}(Q,R)) \ar[d,tail] \ar[r,dashed]
      & \ol{B}(Q,\ol{\Sigma}^{\star}(Q + R))
      \ar[d,tail]\\
      \Sigma(X \times B(X,Y))^{2} \ar[r,"\rho \times \rho"]
      & B(X,\Sigma^{\star}(X + Y))^{2}
    \end{tikzcd}
  \end{equation*}
The above diagram decomposes into the following two diagrams:
\begin{equation}
  \label{eq:liftingdiag2}
  \begin{tikzcd}[column sep=5em]
    \ol{\Sigma}(Q \times \ol{B}(R,Q)) \ar[d,tail] \ar[r,dashed]
    & \ol{\llangle Q, \ol{\Sigma}^{\star}(Q + R) \rrangle}
    \ar[d,tail]\\
    \Sigma(X \times B(X,Y))^{2} \ar[r,"\rho^{1}\,{\times}\,\rho^{1}"]
    & \llangle X, \Sigma^{\star}(X + Y) \rrangle^{2}
  \end{tikzcd}
  \qquad
  \begin{tikzcd}[column sep=5em]
    \ol{\Sigma}(Q \times \ol{B}(Q,R)) \ar[d,tail] \ar[r,dashed]
    & \ol{\Pow_{\star}}(\ol{\Sigma}^{\star}(Q + R) + \ol{D}(R,\ol{\Sigma}^{\star}(Q + R)))
    \ar[d,tail]\\
    \Sigma(X \times B(X,Y))^{2} \ar[r,"\rho^{2}\,\times\,\rho^{2}"]
    & \Pow_{\star}(\Sigma^{\star}(X + Y) + D(X,\Sigma^{\star}(X + Y)))^{2}
  \end{tikzcd}
\end{equation}
The left diagram commutes by definition of $\rho^{1}$, as it stems from the
$V$-pointed strength $\mathrm{str}$ of $\Sigma$, which admits a canonical
relation lifting. It remains to consider commutativity of the diagram on the right. The
left-to-right Egli-Milner relation lifting $\ol{\Pow_{\star}}$ works as follows:
Given a relation $R \seq X \times X$ then
\[
  (V,U) \in \ol{\Pow_{\star}}_{\tau}(R)(\Gamma) \iff
  \forall t \in V \subseteq X_{\tau}(\Gamma).\,\exists s \in U \subseteq
  X_{\tau}(\Gamma) \text{\quad with \quad} R_{\tau}(\Gamma)(t,s).
\]
With the above in mind, we interpret commutativity of the right diagram in
\eqref{eq:liftingdiag2} as a kind of monotonicity condition on the operational
rules of \textbf{FPC}, applying on each $\mathsf{f} \in \Sigma$. In particular:
\begin{enumerate}
\item Case of variables. Let $x \c \tau \in V_{\tau}(\Gamma)$. The condition
  says that for all $f_{1} \in \rho^{2}(x) = \varnothing$, there exists $f_{2} \in
  \rho^{2}(x) = \varnothing$ such that $\ol{\Sigma}^{\star}(Q + R) +
  \ol{D}(Q,\ol{\Sigma}^{\star}(Q + R))_{\tau}(\Gamma)(f_{1},f_{2})$. This is
  trivially true.
\item Case $\mathsf{app}_{\tau_{1},\tau_{2}}$. Assume terms $t_{1},t_{2}
  \in X_{\arty{\tau_{1}}{\tau_{2}}}(\Gamma)$ and $s_{1},s_{2} \in
  X_{\tau_{1}}(\Gamma)$, with
  $Q_{\arty{\tau_{1}}{\tau_{2}}}(\Gamma)(t_{1},t_{2})$ and
  $Q_{\tau_{1}}(\Gamma)(s_{1},s_{2})$, and behaviours
  $V_{1},V_{2} \in \Pow_{\star}(Y + D(X,Y))$, such that for each $f_{1} \in
  V_{1}$ there exists $f_{2} \in V_{2}$ with $(R +
  \ol{D}(Q,R))_{\tau}(\Gamma)(f_{1},f_{2})$. For each $i \in \{1,2\}$, the
  operational rules maps each term
  $\mathsf{app}_{\tau_{1},\tau_{2}}(t_{i},s_{i}) = t_{i} \app{}{} s_{i}$ to the set
  \[
    Z_{i} = \Biggl(t_{f_{i}} =
    \begin{cases}
      \makebox[10em][l]{$f_{i}(s_{i})$}
      \text{\quad if $f_{i} \in Y_{\tau_{2}}^{X_{\tau_{1}}}$} \\
      \makebox[10em][l]{$\mathsf{app}_{\tau_1,\tau_2}(f_{i}, s_{i})$}
      \text{\quad if $f_{i} \in Y_{\arty{\tau_{1}}{\tau_{2}}}$}
    \end{cases}\Biggr)_{f_{i} \in V_{i}}
  \]
  The lifting condition asserts that for each $t_{f_{1}} \in Z_{1}$ there exists
  a related $t_{f_{2}} \in Z_{2}$, which is true. Intuitively, the
  condition holds because for each derivation of $t_{1} \app{}{} s_{1}$, there
  is always a derivation for $t_{2} \app{}{} s_{2}$ where the conclusions are
  related. For instance:
  \[
    \inference{t_{1}\to f_{1}}{t_{1} \app{\tau_{1}}{\tau_{2}} s_{1}\to f_{1}
      \app{\tau_{1}}{\tau_{2}} s_{1}}
    \text{\quad and \quad}
    \inference{t_{2}\to f_{2}}{t_{2} \app{\tau_{1}}{\tau_{2}} s_{2}
      \to f_{2} \app{\tau_{1}}{\tau_{2}} s_{2}}
    \text{\quad with \quad}
    \ol{\Sigma}^{\star}(R + Q)_{\tau_{2}}(\Gamma)(f_{1}\app{\tau_{1}}{\tau_{2}} s_{1},
    f_{2} \app{\tau_{1}}{\tau_{2}} s_{2}).
  \]
\item Case $\mathsf{lam}_{\tau_{1},\tau_{2}}$. Assume terms $t_{1},t_{2} \in
  X_{\tau_{2}}(\Gamma + \check\tau_{1})$ with $Q_{\tau_{2}}(\Gamma +
  \check\tau_{1})(t_{1},t_{2})$ and substitution structures $v_{1},v_{2} \in
  \llangle X,Y \rrangle_{\tau_{2}}(\Gamma)$ such that, given substitutions
  $u_{1},u_{2}$ whose components are pointwise in
  $Q$, then $R_{\tau_{2}}(\Delta)(v_{1}[u_{1}],v_{2}[u_{2}])$. The condition
  says that for all $e_{1},e_{2} \in X_{\tau_{1}}(\Gamma)$ with
  $Q_{\tau_{1}}(\Gamma)(e_{1},e_{2})$,
  $R_{\tau_{2}}(\Gamma)(v_{1}[x_{1},\dots,x_{|\Gamma|},e_{1}],v_{2}[x_{1},\dots,x_{|\Gamma|},e_{2}])$,
  which is true because $v_{1},v_{2}$ and $e_{1},e_{2}$ are appropriately
  related and $Q$ is pointed, hence contains the variables $x_{i}$.
\end{enumerate}
The rest of the operations work similarly to either the case of applications and
$\lambda$-abstractions.

\begin{theorem}
  Relation $\logrel^{2^{\omega}}$ is a (reflexive) congruence.
\end{theorem}
\begin{proof}
  All rules of \textbf{FPC} are sound for the weak transition system $\To$,
  hence $(\rho, \iota, \langle \gamma_{0},\wt{\gamma} \rangle)$ forms a lax
  $\rho$-bialgebra. As such, applying \Cref{cor:main} and \Cref{cor:fundamental-property}
  gives the desired result.
\end{proof}

Following other examples in the literature of nondeterministic higher-order
languages (see for instance
\cite{DBLP:conf/csl/BirkedalST12,DBLP:journals/corr/BirkedalBS13,DBLP:journals/pacmpl/AguirreB23}),
one suitable choice of observable relation $O \monoto \Lambda \times \Lambda$
for the contextual preorder is \emph{may-termination}. To model
may-termination, we define $O$ as
$O_\tau(\Gamma) \{ (\Gamma \vdash (s,t) \c s{\Downarrow}\,\To\,
t{\Downarrow}  \}$. The category $(\Set^{\fset/{\Ty}})^{\Ty}$ and the signature
endofunctor~$\Sigma$ of \eqref{eq:sigmalamapp} satisfy the assumptions of
\Cref{prop:existence-of-contextual-preorder}, thus the contextual preorder
$\cprd$ exists. It can given in more explicit term using contexts, completely analogous to \eqref{eq:ctxprox}.

\begin{corollary}
  The logical relation $\logrel$ is sound for the contextual preorder.
\end{corollary}

\begin{proof}
  The logical relation $\logrel$ is easily seen to an $O$-adequate congruence. Thus, by
  \Cref{cor:soundness}, it is contained in $\cprd$.
\end{proof}
The \emph{ground contextual preorder}, which restricts contexts to the $\booltype$ type, is handled analogously, cf.\ \autoref{ex:dbtilde}.

\takeout{
\begin{proof}[Proof of \Cref{th:mainweak}]
  In the following we omit superscript $\overline{B},\gamma, \widetilde{\gamma}$
  from $\invp^{\overline{B},\gamma,\widetilde{\gamma},\alpha} \top$. We prove
  that $\ol\Sigma(\invp^\alpha\top)$ is a congruence, i.e.
  $\fimg{\iota}{\ol\Sigma(\invp^\alpha\top)} \leq  \invp^\alpha\top$, by
  transfinite induction on $\alpha$.
  \begin{enumerate}
  \item We assume $\alpha$ is a limit ordinal.
  \item Assuming, $\fimg{\iota}{\ol\Sigma(\invp^\alpha\top)} \leq
    \invp^\alpha\top$, we prove $\fimg{\iota}{\ol\Sigma(\invp^{\alpha + 1}\top)}
    \leq  \invp^{\alpha + 1}\top$, which reduces to the following subgoals by
    the definition of $\invp$.
    \begin{align}
      \fimg{\iota}{\ol\Sigma(\invp^{\alpha+1}\top)}\leq
      &\; \invp^{\alpha}\top\label{eq2:abox1}\\
      \fimg{\iota}{\ol\Sigma(\invp^{\alpha+1}\top)}\leq
      &\; \iimg{(\gamma \times \widetilde{\gamma})}{\overline{B}(\invp^{\alpha+1}\top,\invp^{\alpha}\top)}\label{eq2:abox2}
    \end{align}
    By the definition of $\invp$ and induction hypothesis,
    \begin{align*}
      \fimg{\iota}{\ol\Sigma(\invp^{\alpha+1}\top)}\leq \fimg{\iota}{\ol\Sigma(\invp^{\alpha}\top)}\leq\invp^{\alpha}\top,
    \end{align*}
    which yields~\eqref{eq2:abox1}. For \eqref{eq2:abox2}, we have the following diagram:
    \[
      \begin{tikzcd}
        \mS
        \ar{dddd}{\gamma}
        & \Sigma(\mS)
        \ar[swap]{l}{\iota}
        \ar{d}[swap]{\Sigma(\id,\gamma)} & & \overline{\Sigma}(\invp^{\alpha+1}\top)
        \ar[swap,  tail]{ll}{\outl_{\overline{\Sigma}(\invp^{\alpha+1}\top)}}
        \arrow[d,dashed]
        \ar[tail]{rr}{\outr_{\overline{\Sigma}(\invp^{\alpha+1}\top)}}
        & & \Sigma(\mS) \ar[swap]{d}{\Sigma(\id,\widetilde{\gamma})}
        \ar{r}{\iota}
        \ar[phantom]{ddddr}[description, pos=.4]{\dleq{45}}
        & \mS \ar{dddd}[swap]{\widetilde{\gamma}} \\
        &
        \Sigma(\mS \times B(\mS,\mS))
        \ar{d}[swap]{\rho_{\mS,\mS}}
        & & \ar[swap,tail]{ll}{\outl_{Q}}
        \overline{\Sigma}(\invp^{\alpha+1}\top \times
        \overline{B}(\invp^{\alpha+1}\top, \invp^{\alpha}\top))
        \arrow[d]
        \ar[tail]{rr}{\outr_{Q}}
        & & \Sigma(\mS \times B(\mS,\mS)) \ar{d}[swap]{\rho_{\mS,\mS}}
        & \\
        &
        B(\mS, \Sigma^{\star}(\mS + \mS))
        \ar[equal]{d}
        & & \ar[swap,tail]{ll}{\outl_{Q'}}
        \overline{B}(\invp^{\alpha + 1}\top,
        \overline{\Sigma}^{\star}(\invp^{\alpha + 1}\top + \invp^{\alpha}\top))
        \ar[tail]{rr}{\outr_{Q'}}
        \ar[dashed]{d}
        & & B(\mS, \Sigma^{\star}(\mS + \mS))
        \ar[equal]{d}
        &
        \\
        &
        B(\mS, \Sigma^{\star}(\mS + \mS))
        \ar{dl}{B(\mS,\hat\iota \comp \nabla)}
        & &
        \ar[swap,tail]{ll}{\outl}
        \overline{B}(\invp^{\alpha + 1}\top,
        \overline{\Sigma}^{\star}(\invp^{\alpha}\top + \invp^{\alpha}\top))
        \ar[dashed]{d}
        \ar[tail]{rr}{\outr}
        & &
        B(\mS, \Sigma^{\star}(\mS + \mS))
        \ar[swap]{dr}{B(\mS,\hat\iota \comp \nabla)}
        \\
        B(\mS,\mS)
        & & &
        \ar[tail,swap]{lll}{\outl}
        \overline{B}(\invp^{\alpha + 1}\top, \invp^{\alpha}\top)
        \ar[tail]{rrr}{\outr}
        & & &
        B(\mS,\mS)
      \end{tikzcd}
    \]
  \end{enumerate}
\end{proof}
}
\end{document}